\documentclass[journal,onecolumn]{IEEEtran}
\usepackage{cite}
\usepackage{balance}
\usepackage{amsmath, amssymb}
\usepackage{hyperref}

\usepackage{mathtools}
\usepackage{bbm}
\usepackage{algorithm}
\usepackage{algpseudocode}
\usepackage{amsthm}

\DeclarePairedDelimiter\abs{\lvert}{\rvert}%

\newcommand{\RN}[1]{%
  \textup{\uppercase\expandafter{\romannumeral#1}}%
}
\DeclareMathOperator*{\argmax}{arg\,max}  
 
\newtheorem{theorem}{Theorem}
\newtheorem{lemma}[theorem]{Lemma}
\newtheorem{corollary}{Corollary}[theorem]
\usepackage{amsmath}               
  {
      \theoremstyle{plain}
      \newtheorem{assumption}{Assumption}
  }
\usepackage{color}

\definecolor{bgrd}{rgb}{1,1,1}
\definecolor{grey}{rgb}{0.9,0.9,0.6}
\definecolor{gray}{rgb}{0.5,0.5,0.5}
\definecolor{dkr}{rgb}{0.6,0.2,0.2}
\definecolor{dkg}{rgb}{0,0.5,0}
\definecolor{dkb}{rgb}{0.0,0.1,0.7}
\definecolor{light-gray}{gray}{0.85}

\def\tcdkb{\textcolor{dkb}}

\DeclarePairedDelimiter\floor{\lfloor}{\rfloor}

\begin{document}
\title{Distributed Chernoff Test:\\ Optimal Decision Systems over Networks}
%
%
%

\author{Anshuka Rangi, Massimo Franceschetti and Stefano Marano

\thanks{A. Rangi and M. Franceschetti  are with Department of Electrical and Computer Engineering of University of California, San Diego, USA. 
S. Marano is with Department of Information and Electrical Engineering 
and Applied Mathematics of University of Salerno, Fisciano (SA), Italy. E-mails:
        {\tt\small \{arangi,massimo\}@ucsd.edu, }\tt\small      {marano@unisa.it}}
\thanks{This paper was presented in part at ISIT 2018 and CDC 2018 \cite{rangi2018decentralized,rangi2018consensus}.}}

\maketitle
\begin{abstract}
We study ``active''  decision making over sensor networks where the sensors' sequential probing actions are actively chosen by  continuously learning  from past observations. We consider two network settings: with and without central coordination. In the first case, the network nodes interact with each other through a central entity, which plays the role of a fusion center. In the second case, the network nodes interact  in a fully distributed fashion.
In both of  
these scenarios, we propose sequential and adaptive hypothesis tests extending the classic Chernoff  test.
We compare the performance of the proposed tests to the optimal 
sequential test. In 
the presence of a fusion center, our test achieves the same asymptotic optimality of the Chernoff test, minimizing the risk, expressed by the expected cost required to reach a decision plus the expected cost of making a wrong decision, when the observation cost per unit time tends to zero. 
The  test is also asymptotically optimal in the higher moments of the time required to reach a decision. Additionally, the test is parsimonious in terms of communications, and the expected number of channel uses per network node tends to a small constant.
In the distributed setup, 
our test 
achieves the  same asymptotic optimality of 
 Chernoff's test, up to a multiplicative constant in terms of both risk and the higher moments of the decision time. 
Additionally, the test is parsimonious in terms of communications in comparison to  state-of-the-art schemes proposed in the literature. 
The analysis of these tests is also extended to account for  message quantization and communication over channels with random erasures.
\end{abstract}

\IEEEpeerreviewmaketitle

\section{Introduction}
%
%
%
%

\IEEEPARstart{W}{ith} the boom in the Internet of Things, sensor-network based solutions  for  inference 
systems have become increasingly popular~\cite{atzori2010internet,li2015internet,mainetti2011evolution}. This is mainly due to the decreasing cost of the sensors, their increasing computational capabilities, the availability of high-speed communication channels, and the redundancy provided by the distributed nature of the network~\cite{varshney2012distributed}. Inference systems have two key functionalities: decision making (\emph{viz.}\ hypothesis testing) and estimation.  We focus on  designing optimal tests for sensor networks in  decision-making   scenarios where the sensors actively choose their probing actions by continuously learning from past observations.
 Applications that fall in this framework include intrusion and target detection, and object classification and recognition~\cite{hu2016mobile,liu2016activetrust,aslam2003tracking,kulkarni2005senseye,ren2007design}. 
 
Previous studies are broadly classified into two 
{categories:} fusion-center based and distributed setting. 
In the first case, all the nodes of the network are connected to a fusion center --- and two operative modalities are considered. 
In the first modality, the network nodes simply deliver their observations to the fusion center, where the inference task is performed. In the second modality, the nodes exploit their computational capability to perform preliminary processing of the observations, and only a limited amount of information is delivered to the fusion center for making the final decision. This reduces the communication overhead,
but may also result in a loss of performance. 
In the distributed setup, network nodes are connected to each other via communication links, typically forming a sparse network, and there is no central processing unit. Thus, to perform an inference task, the network nodes need to perform computations locally, share their processed data with neighboring nodes, and collectively reach a decision. 
A natural question in both  settings is what kind of local processing to perform at the nodes, and  what  fusion scheme to adopt at the fusion center or at the network nodes, in order to reduce the communication burden while keeping a high level of performance. In this work, we  address this question  and propose  statistical tests for both settings.

Hypothesis tests can be broadly classified as sequential or non-sequential tests, as well as adaptive or non-adaptive tests.  
In a sequential test the number of observations needed to reach a decision is not fixed in advance, but depends on the realization of the observed data.
The test proceeds to collect and process data until a decision with a prescribed level of reliability can be made, and an important performance figure --- in addition to the probability of correct decision --- is the average number of observations required to end the test.  In an adaptive test,  the sensors' probing actions are chosen on the basis of the collected data in an on-line, causal manner. Hence, the sensors learn from the past, and adapt their future probing actions in a closed-loop fashion. In this case, the sensors are said to be ``active," in the sense that measurement observations are the consequence of the sensors' chosen probing actions. Our focus here is on sequential and adaptive tests. 


We propose a Decentralized Chernoff Test (DCT) for the fusion center based setup, 
and a Consensus-based Chernoff Test (CCT) for the distributed setup. 
We provide bounds on the  performance of the tests in terms of their risk, defined as 
the expected cost required to reach a decision plus the expected cost of making a wrong decision.
We also provide converse results
 showing the best possible performance of \emph{any} adaptive or non-adaptive sequential test over the network. 
We show that DCT is asymptotically optimal in terms of both the risk and the higher moments of the expected decision time, as the observation cost per unit time tends to zero. Additionally, DCT is parsimonious in terms of communication: when the observation cost per unit time tends to zero, the expected number of messages sent per node tends to a small constant. Finally, 
we  show that  CCT   also retains the  asymptotic optimality of Chernoff's original solution, being order optimal up to a multiplicative constant, in terms of both risk and higher moments of decision time.




To ease the presentation, our initial analysis assumes
ideal communication links carrying real-valued messages without errors. 
In a real network,  messages are quantized into packets of a fixed length, and subject to random erasures at each transmission. In the second part of the paper,  we  extend our results to this scenario.

The rest of the paper is organized as follows: Section \ref{sec:RelatedWork} discusses  related work; Section \ref{sec:ProblemFormulation} formulates the problem;  Section~\ref{sec:classicalChernoff} reviews the standard Chernoff test; Section~\ref{sec:DCT} introduces the Decentralized Chernoff Test (DCT); 
Section~\ref{sec:CCT} introduces the Consensus-based Chernoff Test (CCT); 
Section~\ref{sec:TRofCCT} presents theoretical results on DCT and CCT;  Section~\ref{sec:Simulations} presents simulation results; 
Section~\ref{sec:nonideal} extends the analysis to quantized messages and erasure channels; Section~\ref{sec:Conclusion} concludes the work.   The proofs of all results appear in the Appendices.

\section{Related Work}\label{sec:RelatedWork}
 Sequential tests were first introduced by Wald in 1973~\cite{wald1973sequential}.  One of such tests, the Sequential Probability Ratio Test (SPRT) has been proven optimal for binary  hypothesis testing in~\cite{wald1948optimum},  
 and for  multi-hypothesis testing  in~\cite{dragalin2000multihypothesis,draglia1999multihypothesis}. The performance of sequential tests can be further improved by combining them with adaptive schemes. These schemes operate in closed-loop, adapting the  choice of actions to past observations.
 In the case of sequential and adaptive tests,  Chernoff provided the optimal test for binary composite hypotheses in~\cite{chernoff1959sequential}. Its  asymptotic optimality for multi-hypothesis testing was proven in~\cite{nitinawarat2013controlled}; see also~\cite{franceschetti2016chernoff} and  references therein for an application. Later, the sequentiality and adaptivity gains for different classes of tests were studied, and it was established that sequential adaptive tests outperform other classes of tests~\cite{naghshvar2013sequentiality}, and that the gains can vary from application to application~\cite{naghshvar2012extrinsic,naghshvar2013active,goodman2007adaptive,rangi2018multi}. All of these results were established in the case a single agent performs the  test. 
 
 Different works discuss the extension  to an ensemble of networked sensors independently making observations and coordinating to reach a decision~\cite{blum1997distributed,viswanathan1997distributed}. 
 Different techniques for combining the information from different sensors at a fusion center  are considered in ~\cite{chair1986optimal,varshney2012distributed, jiang2005fusion,lin2005decision}.  In this case, minimization of the risk, which depends on both the decision time and the reliability of the decision, requires joint optimization over both the node level computations and the fusion center operations.  Key challenges of this optimization problem are discussed in \cite{veeravalli1993decentralized}, and   
asymptotically optimal sequential (non-adaptive) tests have been developed in~\cite{mei2008asymptotic,wang2011asymptotic}. 

Previous works have not considered the performance of sequential, adaptive tests in a network setting. The DCT proposed here fills this gap for star networks, namely for networks in which each node is directly connected to a fusion center. On the other hand, the CCT proposed here  considers  networks having a general graph structure and no central entity. In this more general case, different non-sequential tests have been developed relying on gossip protocols  for distributed computation~\cite{boyd2006randomized,olfati2007consensus,braca2008enforcing,degroot1974reaching,kempe2003gossip,braca2008running}. These protocols can be broadly classified into two categories: consensus protocols and running-consensus protocols. In consensus protocols, a distributed computation task is performed after the collection of all the measurements  at the network nodes~\cite{boyd2006randomized,olfati2007consensus,degroot1974reaching,kempe2003gossip}. Necessary and sufficient conditions for  convergence are well studied, see e.g.,~\cite{xiao2004fast}.
In running-consensus protocols, the collection of the measurements from the environment and the computation task are performed simultaneously at the network nodes \cite{braca2008enforcing,braca2008running}. Hypothesis testing schemes  typically rely on  
consensus over ``belief vectors.'' In this case, each network node holds a belief vector, whose elements  represent the probability that a certain hypothesis is true, given all the information collected by the node.
Different strategies are then used to transmit and combine the belief vectors over the network, leading to asymptotic learning of the  correct hypothesis \cite{jadbabaie2012non,parasnis2020non,nedic2015nonasymptotic,shahrampour2016distributed,duchi2012dual,lalitha2014social}. For example,
 a strategy based on distributed dual averaging was proposed in \cite{duchi2012dual}, using an optimization algorithm from \cite{shahrampour2013exponentially}. 
The work in \cite{jadbabaie2012non} proposes usage of linear consensus strategies to combine the belief vectors, and \cite{parasnis2020non} extends the results of \cite{jadbabaie2012non} to the case of random time-varying networks. 
Other works consider Bayesian strategies for updating and combining the belief vectors at the nodes \cite{lalitha2014social}.  In~\cite{jadbabaie2012non,lalitha2014social}
the bounds on the asymptotic learning rate are presented in terms of KL-divergences of the beliefs at the different network nodes. 
Under the assumption that the log-likelihood ratio is bounded, finite-time analysis of the KL-divergence cost has been performed in \cite{shahrampour2016distributed}. Similar results have been obtained for networks modeled as time-varying  graphs \cite{nedic2015nonasymptotic,nedic2016network}. 

Despite this huge literature, only limited attention has been given to distributed \emph{sequential} hypothesis testing over general networks, which requires designing an appropriate stopping rule over the network and evaluating the corresponding expected decision time and performance in terms of risk. 
Recently, a sequential (non-adaptive) hypothesis test which is asymptotically optimal among non-adaptive tests  has been proposed \cite{li2018fully}. In the present work, we propose a sequential as well as  \emph{adaptive} hypothesis test in the distributed network setup.  Unlike the previous literature, including \cite{li2018fully}, the proposed test  does not perform consensus over the belief vector, and is parsimonious in terms of communication.   The stopping criterion proposed in \cite{li2018fully} is not applicable to our test. Our test is also asymptotically optimal among all adaptive or non-adaptive sequential tests, under a broad range of conditions. Finally, we point out that unlike our work, all of the above works do not consider the effect of quantization  and erasures occuring over the communication links. 



\section{Problem Formulation} \label{sec:ProblemFormulation}
\subsubsection*{{Testing model}}
We consider an ensemble $\mathcal{L}=\{1,2,\ldots,L\}$  of  sensor  nodes engaged in a multi-hypothesis testing problem.  The state of nature to be detected is one of $M$ exhaustive and mutually exclusive real-valued   hypotheses $\{h_i\}_{i \in [M]}$, where   $[M] =\{1,2,\dots M\}$. Nodes are connected by bi-directional communication links to form a network. At   each discrete  time step $n$  every node  $\ell  \in \mathcal{L}$ can select a probing action $u_{n,\ell} \in S $, where $S$  is a fixed set of cardinality $M$. As a consequence of this action, the node observes the realization of a real-valued random variable $Y_{n,\ell}$ whose distribution  is $p_{i,\ell}^{u_{n,\ell}}$ and is known to node $\ell$ only. The node  can then send one message over each of its incident links, and receive  one message  from each link. The probing actions  and the messages sent  at time $n$ can be selected based on all past observations, actions taken, and messages sent and received up to time $n-1$.  It follows that  the observations at each node can  be dependent across time. On the other hand,
 given the state of nature, we assume that  the observations at different nodes are conditionally independent, but not necessarily identically distributed.

\subsubsection*{{Network model}}
We consider two network setups.
\begin{enumerate}
    \item \emph{Star network.} In this case, the network is composed of the $L$ sensors and of one special node acting as a fusion center. Each sensor is connected to the fusion center  via a communication link, while there  are no  links  between  the sensors.  This setup   is used to introduce our Decentralized Chernoff Test (DCT).

\item{\emph{General network}.} In this case, the network is represented by a connected graph   $\mathcal{G(L,E)}$, where the edges $\{(\ell,j)\} \in \mathcal{E}$, are such that $\ell, j \in \mathcal{L}$, $\ell \neq j$.
Communication and   information processing tasks  are fully distributed and there is no fusion center. This setup is used to introduce our Consensus-based Chernoff Test (CCT).
\end{enumerate}

\subsubsection*{{Communication model}}
We first assume an ideal communication model, where at each time step every node can send and receive a vector composed of $C$ real-values over each of its incident links. The messages sent are received instantaneously and without error.  This  synchronous model of communication with no queueing delay and real vector channels has been widely used in the literature of detection and estimation, see e.g. \cite{mei2008asymptotic,wang2011asymptotic,boyd2006randomized,olfati2007consensus,braca2008enforcing,degroot1974reaching,kempe2003gossip,braca2008running,jadbabaie2012non,nedic2015nonasymptotic,shahrampour2016distributed,duchi2012dual,lalitha2014social,xiao2004fast,li2018fully}. 

We then refine the communication model by taking into account that in a real packet-switched network, links can only carry  a finite number of bits at each transmission, rather than real numbers. In this case,  if there is a  communication link connecting nodes $\ell$ and $j$, then we assume that at each time step   node $\ell$ can transmit a packet of $C$ bits  to node $j$ and at the same time step
node $j$ can transmit a packet of $C$ bits  to node $\ell$.  This accounts for quantization of the real data in the previous model. 
In information-theoretic terms, every link behaves in each direction as a noiseless channel of finite capacity $C$ bits/transmission.   As in the previous model, every packet transmission occurs synchronously in one time step, and there is no queuing delay. Although less popular than the previous one, this refined model has been considered in the context of quantized consensus in~\cite{nedic2009distributed}, and in the context of estimation and detection in~\cite{mei2008asymptotic,wang2011asymptotic,Tsitsiklis88,Luo-univdetect,varshney2012distributed,MaranoSayedIT19}.
 
Finally, we further extend the communication model  by considering random packet erasures. We assume that at any time step any link in the network can fail independently with probability $\epsilon$. When a link fails,   packets travelling on both directions of the link are received as  ``erasures."
In information-theoretic terms, every link behaves in this case in both directions as a $C$-bit erasure channel without feedback,  having capacity
$(1-\epsilon)C$ bits/transmission. As in the previous case, transmissions are synchronous, and there is no queuing delay.   A related model, where links can fail at random times  but carry real numbers rather than quantized packets has been used to study consensus in  \cite{kar2008sensor} and estimation and detection in~\cite{CattivelliSayedEstimation,CattivelliSayedDetection,Kar-Moura2010,shahrampour2016distributed,7384746,6459049,6422410,5404434,5342458,4663899}.

\subsubsection*{{Performance measure}}
Our objective is to design a scheme to select at each step the nodes' probing actions   and the messages to transmit, to eventually decide the state of nature with   sufficiently high reliability. To quantify the performance of the proposed scheme, we let $N$  be the random time  at which all the nodes have reached the same decision and halt   the test. We consider both the expectation and the higher moments of this stopping time.   Following~\cite{chernoff1959sequential}, we also consider the risk, expressed as the   sum of the expected cost required to reach a decision and the expected cost of making a wrong decision. Namely,
under the true hypothesis $H^*=h_i$, we let the risk $\mathbb{R}_{i}^{\delta}$ of a  test $\delta$ be
\begin{equation}\label{risk}
\mathbb{R}_{i}^{\delta} = c\, \mathbb{E}^{\delta}_{i}[N]+\omega_i \, {\mathbb{P}}_{i}^{\delta}(\hat{H} \neq h_{i}),  
\end{equation}
where $c$ is the observation cost per unit time, $\hat{H}$ is the final decision, 
$\mathbb{E}_i$ and $\mathbb{P}_i$ are the expectation and the probability operators computed under $H^*=h_i$,
and $\omega_i$ is the cost of a wrong decision. As in~\cite{chernoff1959sequential}, we evaluate the risk   for all $i \in [M]$, as $c \rightarrow 0$. 

\vspace*{3pt}
\subsubsection*{{Additional notation}} All logarithms are to base $e$, unless otherwise indicated.
For the  general network case, we denote by $d^{\mathcal{G}}$ the diameter of the network, which is the maximum shortest hop-distance between any pair of nodes of $\mathcal{G(L,E)}$. We  denote by $h^{\mathcal{G}}$ the shortest height of all possible spanning trees of $\mathcal{G(L,E)}$. Since the network is connected, $d^{\mathcal{G}}$ and $h^{\mathcal{G}}$ are both finite. For all $\ell \in [L]$, $u\in S$ and $i,j\in [M]$, the KL-divergence between hypotheses $h_i$ and $h_j$ is denoted by $D(p_{i,\ell}^{u}||p_{j,\ell}^{u})$,
and is assumed to be finite over the entire action set $S$. We also assume that for all $\ell \in [L]$ and $i,j\in[M]$, there exists an action $u\in S$ such that $D(p_{i,\ell}^{u}||p_{j,\ell}^{u})>0$. This assumption entails little loss of generality, rules out trivialities, and is commonly adopted in the literature, see e.g.,~\cite{chernoff1959sequential}. For all $\ell \in [L]$, $u\in S$ and $i,j\in [M]$, we also assume $\mathbb{E}[\log(p_{i,\ell}^{u}(Y))/\log(p_{j,\ell}^{u}(Y))]^{2}<\infty$. If $v_{1}=[v_{1,1},\ldots v_{k,1}]$ and $v_{2}=[v_{1,2},\ldots v_{k,2}]$ are two vectors of same dimension, then with $v_{1} \preceq v_{2}$ we mean that for all $i\in [k]$, $v_{i,1}\leq v_{i,2}$.  Finally, we indicate with $|v_{1}|$   the vector of absolute values of the entries of $v_{1}$. 

\section{Standard Chernoff Test}\label{sec:classicalChernoff} 
\label{sec:classical}
We start by describing the Standard Chernoff Test (SCT) for a single sensor $\ell$   attempting to detect the true hypothesis $H^{*}$, having no interactions with other sensors in the network~\cite{chernoff1959sequential}.  For all $n >1$ we let $y^{n}_{\ell}$ $=$ $\{y_{1,\ell}, \ldots, y_{n-1,\ell}\}$, where $y_{i,\ell}$ denotes the realization of the observation collected at time step~$i$,  and let $u^{n}_{\ell}=\{u_{1,\ell},\ldots, u_{n-1,\ell}\}$, where $u_{i,\ell}$ denotes the realization of the action made at step $i$. For $n=1$ we initialize  the set of previous actions $u^{n}_{\ell} = \emptyset$ and previous observations $y^{n}_{\ell}=\emptyset$, and let all posterior probabilities be the same, namely $\mathbb{P}(H^{*}=h_{i}|y^{n}_{\ell},u^{n}_{\ell})=1/M$.

At every step $n \geq 1$, the test  proceeds as follows:

\begin{itemize}
    \item[1)] A temporary decision is made, based on the maximum posterior probability of the hypotheses, given the past observations and actions  of the sensor. Ties are resolved at random. 
    This temporary decision is in favor of  $h_{i^*_{n}}$ if
    \begin{equation}
    i^*_{n}= \argmax_{i\in [M]}\mathbb{P}(H^{*}=h_{i}|y^{n}_{\ell},u^{n}_{\ell}).
     \label{MAP}
    \end{equation}
   \item[2)] A new action
   $u_{n,\ell}$ is randomly chosen among the elements of the action set $S$, according to the Probability Mass Function (PMF) 
   \begin{equation} \label{eq:pmf}
       Q_{i^*_{n}}^{\ell}=\argmax_{q\in \mathcal{Q}}\min_{j \in M_{i^*_{n}}}\sum_{u{\in[M]}}q(u)D(p^u_{i^*_{n},\ell}||p_{j,\ell}^u),
   \end{equation}
   where ${\cal Q}$ denotes the set of all the possible PMFs over the   $[M]$ actions, and $M_{i^*_{n}}=[M] \setminus \{i^{*}_{n}\}$.
   
   \item[3)] As a consequence of this action, a new observation $y_{n,\ell}$ is  collected, and for all $i \in [M]$   the posterior probabilities  
   $\mathbb{P}(H^{*}=h_{i}|y^{n+1}_{\ell},u^{n+1}_{\ell})$ are updated. 
   \item[4)] The test stops if the worst case log-likelihood ratio crosses a prescribed fixed threshold $\gamma$, namely if
   \begin{equation}\label{StoppingRule}
   \log\frac{\mathbb{P}(H^{*}=h_{i^*_{n},}|y^{n+1}_{\ell},u^{n+1}_{\ell})}{\max_{j \neq i^*_{n}}\mathbb{P}(H^{*}=h_{j}|y^{n+1}_{\ell},u^{n+1}_{\ell})}\geq \gamma, 
   \end{equation}
  If the test stops, then the final decision is $h_{i^*_n}$, otherwise $n$ is incremented by one and the procedure continues from 1).
\end{itemize}

\section{Decentralized Chernoff Test}\label{sec:DCT}
We now extend the SCT to a   DCT in the star network configuration. 
We start by noticing that in the SCT  the quantity 
\begin{equation}
\label{eq:value}
v_{i,\ell}=\max_{q\in \mathcal{Q}}\min_{j \neq i}\sum_{u {\in[M]}}q(u)D(p^u_{i,\ell}||p_{j,\ell}^u)
\end{equation}
 is a measure of the capability of node $\ell$
to detect hypothesis $h_i$ (see~\cite{chernoff1959sequential} for a discussion), and plays a critical role for the selection of the action in \eqref{eq:pmf} that is performed at each step and is adapted to the current belief.
In a network setting, the quantity
\begin{equation}
\label{eq:CC}
I(i) = \sum_{\ell=1}^{L}v_{i,\ell} 
\end{equation}
represents  a measure of the cumulative capability of the network to detect hypothesis $h_{i}$ and can be   used for the selection of the threshold of each node in a coordinated fashion to optimize the expected decision time. Accordingly, in 
DCT, the fusion center collects $v_{i,\ell}$ for all $i \in [M]$ and $\ell \in [L]$, computes  $I(i)$ for all $i \in [M]$, and distributes this result to all the nodes to enable their threshold selection. The nodes then perform SCTs in parallel, until they all reach the same decision and terminate the test. The three phases of the test are as follows:

\subsubsection*{Initialization phase}
\begin{enumerate}
\item  Without performing any probing action,
each node~$\ell$ sends the vector $v_\ell=[v_{1,\ell},\dots,v_{M,\ell}]$ to the fusion center.  
\item The fusion center  sends back to each node  the cumulative capability vector  $I=[I(1),\dots,I(M)]$
and, upon reception, each node $\ell$ computes the    
vector $\rho_{\ell}=[\rho_{1,\ell},\dots,\rho_{M,\ell}]$, representing its fraction of network capability, namely  for all $i \in[M]$,  we have
\begin{equation}\label{eq:rhol}
\rho_{i,\ell}= {v_{i,\ell}}/{I(i)}. 
\end{equation}
\end{enumerate}
\subsubsection*{Test phase} Proceeding in parallel, every node $\ell$ performs a SCT using the threshold
\begin{equation}\label{StoppingRuleDec11}
   \gamma  = \rho_{i^*_{n},\ell}\,\abs{\log c}.
\end{equation}
This threshold depends on both the current estimate of the hypothesis and the node  identity, while it was a constant  in~(\ref{StoppingRule}). 
If the log-likelihood ratio  in \eqref{StoppingRule} exceeds the threshold,  node $\ell$ sends its preference for $h_{i^*_{n}}$ to the fusion center and continues to run the test.
Hence, rather than using it as a stopping condition, the threshold is used here as a triggering condition for the communication of a preference by node $\ell$ to the fusion center.  

\subsubsection*{Stopping phase} When   the   preferences expressed by all the  nodes are the same, the fusion center sends a halting message to all the nodes, who stop the test and declare the final decision. 

The proposed DCT only requires the communication of the local preferences from the nodes during the test phase, as well the the messages in the initialization phase, and the halting message.   
We  show below that,  while maintaining the same  asymptotic optimality of the Chernoff test,  the oscillations in the local preferences of the nodes in the test phase vanish as  $c \to 0$ and, if $C\geq M$, each sensor tends to use the channel on average at most four times: two in the initialization phase, one (on average) to communicate the local preference, and one to receive the halting message. In the case $C<M$, the test retains its asymptotic optimality, although the expected number of channel uses per node increases from four to a constant that is at most $2(M+1)$, since in this case multiple transmissions are needed to  communicate each vector in the initialization phase.   

\subsection{Informal Discussion of DCT}\label{sec:InformalDisDCT}
The key idea behind the proposed DCT is to first determine the individual capabilities of the nodes for detecting the hypotheses. These capabilities are captured by the  vector $v_{\ell}$, whose $i^{th}$ element  is a measure of node's $\ell$ capability to detect the hypothesis $h_{i}$. The fusion center gathers this information, and utilizes it to control the threshold at each node through the vector $\rho_{i,\ell}$.  In this context, $I(i)$ is the measure of the cumulative detection capability of the network for hypothesis $h_{i}$
and $\rho_{i,\ell}$  represents  the fraction of this capability contributed by node~$\ell$ for hypothesis $h_{i}$. To minimize the expected time to reach a decision, it is desirable to determine the threshold for each node~$\ell$ such that all the nodes require roughly the same time to reach the triggering condition in (\ref{StoppingRuleDec11}). This is achieved by dividing the task of hypothesis testing among the nodes based on their speed of performing the task, so that all the nodes finish their share of the task at roughly the same time.  

\section{Consensus-Based Chernoff Test}\label{sec:CCT}
We now describe a  CCT in a general network setup, without a fusion center. The main idea is to generalize the DCT to a fully distributed setting.  CCT employs a consensus protocol to agree on the cumulative capability of the network to detect each hypothesis, performs individual SCTs, and then employs another consensus protocol  to finalize the decision. To ease the presentation of CCT, as in~\cite{jadbabaie2012non,nedic2015nonasymptotic,shahrampour2016distributed,duchi2012dual,lalitha2014social}    we now assume   that $C \geq M$, so that consensus can be performed by exchanging real vector messages of size $M$ at every time step. In the case $C<M$ the test proceeds along the same lines, but performing vector communications of size $M$ now requires multiple time-steps,
and the test completion time must be scaled accordingly.
The three phases of the test are as follows:

\subsubsection*{ Initialization Phase} 
The nodes   use a distributed protocol to compute the vector $I=[I(1),\ldots,I(M)]$. Using consensus,
they compute the arithmetic mean $I/L$, and then compute $I$  using their knowledge of $L$. For all $\ell \in [L]$, we let the initial   estimate for $I/L$ at every node  be  $\hat{I}_{\ell}^{0} =[v_{1,\ell},\ldots, v_{M,\ell}]$, which can be computed locally using (\ref{eq:value}). Then, every node $\ell$ runs the following   consensus protocol by iteratively exchanging messages without performing any probing action: for $n \geq0$
\begin{equation}
\label{eq:Consensus1}
\hat{I}_{\ell}^{n+1}=w_{\ell,\ell}\cdot \hat{I}_{\ell}^{n} +\sum_{j \in  \mathcal{N}_{\ell}}w_{\ell,j}\cdot \hat{I}_{j}^{n}, 
\end{equation}
where $\hat{I}_{\ell}^{n} =[\hat{I}_{\ell}^{n}(1),\ldots, \hat{I}_{\ell}^{n}(M)]$ is an estimate of $I/L$ at node $\ell$ and at time $n$, $w_{\ell,j}$ is the weight assigned by node $\ell$ to the estimate received from  node $j$, and $\mathcal{N}_{\ell}=\{j|\{\ell,j\}\in \mathcal{E}\}$ is the set of   neighbors of node $\ell$ in $\mathcal{G(L,E)}$. 
We now rewrite   (\ref{eq:Consensus1})  in matrix form as
\begin{equation}
\label{eq:Consensus2}
\hat{I}^{n+1}=W\cdot \hat{I}^{n}, 
\end{equation}
where $\hat{I}^{n}$
is an $L\times M$ matrix whose $\ell$th row is $\hat{I}_{\ell}^{n}$  and $W$ is an $L\times L$ matrix whose elements satisfy
\begin{equation} \label{eq:FeasibilitySet}
0<w_{\ell,j}<1 \mbox{ if } j\in \mathcal{N}_{\ell}\cup\{\ell\}, \mbox{ otherwise } w_{\ell,j} = 0.
\end{equation}


The following theorem presents  necessary and sufficient conditions for 
the consensus protocol (\ref{eq:Consensus2}) to converge to  $I/L$ as $n \rightarrow \infty$.
\begin{theorem}\label{th:1}
\emph{\cite[Theorem 1]{xiao2004fast}}.  The consensus protocol (\ref{eq:Consensus2})  converges to  $I/L$ as $n \rightarrow \infty$ if and only if  
\begin{equation}
\label{eq:SC11}
\textbf{{1}}_{L\times 1}^T\cdot W=\textbf{{1}}_{L\times 1}^T,
\end{equation}
\begin{equation}
\label{eq:SC22}
W\cdot \textbf{{1}}_{L\times 1}=\textbf{{1}}_{L\times 1},
\end{equation}
and 
\begin{equation}
\label{eq:SC33}
R\left(W-\frac{\textbf{{1}}_{L\times1}\cdot \textbf{{1}}_{1\times L}}{L}\right)< 1,
\end{equation}
where $R(\cdot)$ denotes the spectral radius of a matrix, and $\textbf{{1}}_{A\times B}$ is a $A\times B$ matrix of all ones. Additionally, the rate of convergence  is  proportional to the spectral radius in the left-hand side of (\ref{eq:SC33}).
\end{theorem}
Based on the above theorem, the computation of the weights in the matrix $W$ can be formulated as a convex optimization problem minimizing the spectral radius in $\eqref{eq:SC33}$,  subject 
to~(\ref{eq:FeasibilitySet}), (\ref{eq:SC11}) and~(\ref{eq:SC22}), and can be determined using standard techniques~\cite{xiao2004fast}.  Hence, in the following  we  assume that, in addition to (\ref{eq:FeasibilitySet}),   $W$ verifies the conditions stated in Theorem \ref{th:1}. 

Although the consensus protocol converges to the correct value $I/L$ as $n \rightarrow \infty$, the initialization phase must terminate in finite time and guarantee that consensus has been reached in a suitable approximate fashion.

To characterize approximate consensus, we define a \emph{local} $c$-consensus status  if   for all $\ell\in[L]$ and $j\in \mathcal{N}_{\ell}$, we have
\begin{equation}\label{eq:localConsensus}
    \abs{\hat{I}^{n}_{\ell}-\hat{I}^{n}_{j}}\preceq \frac{c}{L^2}\cdot\textbf{{1}}_{1\times M}.
\end{equation}
We also define a \emph{global} $c$-consensus status if for all $\ell,j\in [L]$, we have
\begin{equation}
\label{eq:error}
  \abs{\hat{I}^{n}_{\ell}-\hat{I}^{n}_{j}}\preceq \frac{c}{L}\cdot\textbf{{1}}_{1\times M}.
\end{equation}
Since  the diameter $d^{\mathcal{G}} \leq L$, it should be clear that local $c$-consensus implies  \emph{global} $c$-consensus.

We employ a stopping rule for the initialization phase that guarantees  global $c$-consensus  as illustrated in  Algorithm \ref{alg:Phase1}. A similar   rule has been previously studied  in  \cite{xie2017stop}.
{In Algorithm \ref{alg:Phase1}, $r^n_{\ell}$ indicates the number of time steps since  node $\ell$ is in  
\emph{local} $c$-consensus, namely satisfies (\ref{eq:localConsensus}). }
The variable $z^n_{\ell}$  is responsible for the percolation of the   consensus information across the network. If at any node $\ell$ we have $z^n_{\ell}>L+1$, then the network has reached  global $c$-consensus and node $\ell$ sends a termination message  $m^{(1)}=1$  to  its neighbors, where the superscript indicates that this is the termination message of the initialization phase. When a node  receives a termination message, it halts the protocol, it scales the final estimate by $L$,  so that
\begin{equation}
\hat{I}^{n}_{\ell} \leftarrow  L \hat{I}^{n}_{\ell},
\end{equation}  and it forwards the termination message to its neighbors. It follows that all the nodes receive a termination message at most~$d^{\mathcal{G}}$ time steps after the first termination message has been sent, and at the end of the initialization phase for all $\ell,j\in [L]$, we have 
\begin{equation}
\label{eq:error2}
  \abs{\hat{I}^{n}_{\ell}-\hat{I}^{n}_{j}}\preceq c\cdot\textbf{{1}}_{1\times M}.
\end{equation}
In the following phases, we let $\hat{I}_{\ell}$ denote  the estimate of vector $I$ at node $\ell$ at the end of the initialization phase. 
\begin{algorithm}[t]
\begin{algorithmic}
\State Initialize $n=0$, and for all $\ell\in [L]$, $\hat{I}_{\ell}^{n}$, $r^n_{\ell}=0$ and $z^n_{\ell}=0$\;
\While{\textbf{True}}
\State For all $\ell\in [L]$, broadcast local information $\hat{I}_{\ell}^{(n)}$ and  $z_{\ell}^n$.
\State Update the local cumulative capability using  (\ref{eq:Consensus1}).
\If{$n\geq 1$}
\State $z^n_{\ell}=\min\{ r^{n-1}_{\ell},\min_{j\in \mathcal{N}_{\ell}\cup \{\ell\}}z^{n-1}_{j}\}+1$
\EndIf
\If{$z^n_{\ell}>L+1$ or $m^{(1)}=1$ is received}
    \State $\hat{I}^{n}_{\ell} \leftarrow  L \hat{I}^{n}_{\ell}$
    \State  Sensor $\ell$ broadcasts $m^{(1)}=1$ and stops updating.
    \State \textbf{Break While};
 \EndIf
 \If {$\max_{j\in \mathcal{N}_{\ell}}\abs{\hat{I}_{\ell}^{(n)}-\hat{I}_{j}^{(n)}}\preceq c\cdot\textbf{{1}}_{1\times M}/L^2$}
 \State $r^n_{\ell}=r^{n-1}_{\ell}+1$
 \Else
 \State $r^n_{\ell}=0$
 \EndIf
 \State  $n=n+1$\;
 \EndWhile
 \caption{Initialization Phase of CCT}
 \label{alg:Phase1}
 \end{algorithmic}
\end{algorithm}
\begin{algorithm}[b]
\begin{algorithmic}
\State For all $i\in [M]$ and $\ell\in[L]$,  $n=0$; $\hat{H}_{\ell}^{n}=\mbox{NULL}$\;
\State Input: Termination message of  stopping phase, i.e.,  $m^{(3)}$
\While{Final decision is not made i.e. $m^{(3)}\neq 1$}
\State For all $\ell\in [L]$, perform SCT with $ \gamma  = \hat{\rho}_{i^*_{n},\ell}\,\abs{\log c}$
\State If $\hat{H}_{\ell}^{n} \not =$NULL, then broadcast $\hat{H}_{\ell}^{n}$ 
\State $n=n+1$
 \EndWhile
 \caption{Test Phase of CCT}
 \label{alg:Phase2}
 \end{algorithmic}
\end{algorithm}
\begin{algorithm}[t]
\begin{algorithmic}
\State For all $\ell\in[L]$, initialize $n=0$; $d_{\ell,n}=x_{\ell}^{n}=0$, $m^{(3)}=0$;\;
\While{TRUE}
\If{$m^{(3)}=1$ is received from neighbor $j$}
\State Set the final decision, i.e., $\hat{H}_{\ell}^{n}=\hat{H}_{j}^{n-1}$
\State Broadcast $m^{(3)}$ and $\hat{H}_{\ell}^{n}$.
\State Break;
\EndIf
\State For all $\ell\in [L]$, update $x_{\ell}^{n}$ according to (\ref{eq:bit2}).
\State For all $\ell\in [L]$, update $d_{\ell}^{n}$ according to (\ref{eq:bit1}).
\If{$d_{\ell}^{N}> L+1$}
\State  $m^{(3)}=1$
\State For all $\ell\in [L]$, broadcast $m^{(3)}$ and $\hat{H}_{\ell}^{n}$.
\Else 
\State For all $\ell\in [L]$, broadcast $d_{\ell}^{n}$ and $\hat{H}_{\ell}^{n}$.
\EndIf
\State $n=n+1$
 \EndWhile
 \caption{Stopping Phase of CCT}
 \label{alg:Phase3}
 \end{algorithmic}
\end{algorithm}

\subsubsection*{ Test Phase}  This phase is illustrated in Algorithm \ref{alg:Phase2} and begins  following the termination of the initialization phase, namely after receiving $m^{(1)}=1$. Every node $\ell$ performs a SCT using the threshold
\begin{equation}\label{StoppingRuleDec111}
   \gamma  = \hat{\rho}_{i^*_{n},\ell}\,\abs{\log c},
\end{equation}
where $\hat{\rho}^{n}_{i^*_{n},\ell}=v_{i^*_{n},\ell}/\hat{I}_{\ell}(i^*_{n})$.
If the log-likelyhood in \eqref{StoppingRule} exceeds the threshold, 
 then node $\ell$ updates its local preference  $\Hat{H}_{\ell}^{n}$ in favor of hypothesis $h_{i^*_{n}}$; otherwise, it sets its local preference   to NULL. Similar to DCT, the node $\ell$ communicates its preference $\Hat{H}_{\ell}^{n}$, if any,   to its neighbors (instead than to the fusion center) and continues to run the test. Hence, rather than using it as a stopping condition, the threshold is used here as a triggering condition for the communication of the preference by node $\ell$ to its neighbors in $\mathcal{N}_{\ell}$. 
 
 

\subsubsection*{ Stopping Phase} 
This phase is illustrated in Algorithm \ref{alg:Phase3} and runs in parallel with the test phase. This phase detects if all the network nodes have reached the same preference, and in this case it halts the test. At every time step $n \geq 1$, every node  $\ell \in [L]$ 
sends $d^{n}_{\ell}$ to its neighbors, which is defined as
\begin{equation}
\label{eq:bit1}
d^{n}_{\ell}=\min\left \{\min_{j\in \mathcal{N}_{\ell}\cup\{\ell\}}d^{n-1}_{j},x^{n-1}_{\ell}\right \}+1,
\end{equation}
where $d_{\ell}^{0}=0$, and
\begin{equation}
\label{eq:bit2}
x^{n}_{\ell}\hspace{-3pt}=\hspace{-2pt}\begin{cases}
         x^{n-1}_{\ell}+1 \quad\hspace{-.5pt}\mbox{ if }\forall j\in \mathcal{N}_{\ell}, \Hat{H}_{\ell}^{n}=\Hat{H}_{j}^{n},  \Hat{H}_{\ell}^{n}=\Hat{H}_{\ell}^{n-1}, \mbox{ and } \Hat{H}_{\ell}^{n} \not =\mbox{NULL},  \\
         1 \qquad\quad\quad\hspace{7pt}\mbox{if }\forall j\in \mathcal{N}_{\ell}, \Hat{H}_{\ell}^{n}=\Hat{H}_{j}^{n}, \Hat{H}_{\ell}^{n}\neq\Hat{H}_{\ell}^{n-1},  \mbox{ and } \Hat{H}_{\ell}^{n} \not =\mbox{NULL},\\
         0 \qquad\quad\quad\hspace{3pt}\mbox{ otherwise},
       \end{cases}
\end{equation}
{where $x_\ell^0=0$.}

 The rationale of (\ref{eq:bit1}) and (\ref{eq:bit2}) is as follows.  Suppose $x^{n}_{\ell}=k$. Then,   for the past $k$ time steps the
local preference of  the neighbors   of node $\ell$ was the same as the local preference $\hat{H}_{\ell}^{n}$ of node $\ell$. The value of $d^{n}_{\ell}$ is responsible for the percolation of this information along the  network. Using (\ref{eq:bit2}),
if node $j\in \mathcal{N}_{\ell}$ does not report any local preference, then the value $x^{n}_{j}=0$ is received by  the neighbors of $j$. If at any node $\ell$ we have $d^{N}_{\ell}>L+1$, then there exists a time $k\leq N$ at which the local decisions of all the nodes are the same, i.e., $\min_{j\in[L]}x_{j}^{k} \geq 1$ (see Lemma \ref{Lemma:stoppingrule} in Appendix \ref{app:Phase3CCT}). This node $\ell$ sends the final decision  $\hat{H}_{\ell}^{N}$ and the termination message $m^{(3)}=1$ to its neighbors, where $m^{(3)}=1$  represents the termination  message for the stopping phase. When a node receives the termination  message and the final decision $\hat{H}_{\ell}^{N}$, it halts the test and forwards $m^{(3)}$ along with $\hat{H}_{\ell}^{N}$ to its neighbors. It follows that all  nodes  receives the termination message and the final decision at most~$d^{\mathcal{G}}$ time steps after the first termination message of the stopping phase has been sent.

\subsection{Informal Discussion of CCT}\label{sec:InformalDisCCT}
As in DCT, the key idea behind    CCT is to first determine the individual capabilities of the nodes for detecting the hypotheses. These capabilities   are captured by  the  vector $v_{\ell}$,  whose $i^{th}$ element  is a measure of node's $\ell$ capability to detect the hypothesis $h_{i}$. However, in contrast to DCT,
 there is no central entity to facilitate the sharing of this information among different nodes, and  a consensus algorithm is used --- in the first phase of CCT --- to gain global knowledge at each node of the capabilities of all the other nodes.   If the consensus algorithm stops at time  $N$, then      $\hat{\rho}^{N}_{i,\ell}$   
denotes the estimated fraction of the capability contributed by node $\ell$ for hypothesis~$h_{i}$. To minimize the expected time to reach a decision, it is desirable to determine this threshold    for each node~$\ell$ such that all the nodes require roughly the same time to reach the triggering condition in (\ref{StoppingRuleDec11}).
This is achieved by dividing the task of hypothesis testing among the nodes based on their speed of performing the task, so  that all the nodes finish their share of the task roughly at the same time. The decision phase is a distributed stopping criterion for the Chernoff test, and ensures that the nodes stop the test as they reach the same decisions. 

\section{Performance analysis}
\label{sec:TRofCCT}

We now present the performance analysis of our tests. The proofs of all theorems are deferred to the Appendices.

\subsection{Lower Bounds for a Sequential and an Adaptive test}
In this section, we present lower bounds on two different performance measures, namely risk and decision time, for \emph{any} sequential and adaptive test. The superscript $\delta$ is appended to quantities that refer to a generic test and $N$ indicates the  time required to take a decision.
\begin{theorem}\label{converse} \emph{(Converse)}
 For any hypothesis testing scheme $\delta$ operating over a network as described in Section \ref{sec:ProblemFormulation}, we have that for all $i\in [M]$,  
 if the probability of missed detection is
 \begin{equation}
  \mathbb{P}^{\delta}_{i}(\hat{H}\neq h_i)=O(c\,|\log c |), \;\; \mbox{ as } c \to 0,
  \label{eq:assumption}
  \end{equation}
   then we have, {for all integers $r \geq 1$,}
\begin{eqnarray}
\mathbb{E}^{\delta}_{i}[N^r]&\geq& \Bigg((1+o(1))\frac{\abs{\log c}}{I(i)}\Bigg)^{r}, \mbox{ as } c \to 0. \label{Econverse}
\end{eqnarray}
Using (\ref{Econverse}) with $r=1$, we also have
\begin{eqnarray}
\mathbb{R}_{i}^{\delta}&\geq& (1+o(1))\frac{c\,\abs{\log c}}{I(i)},  \;\; \mbox{ as } c \to 0. \label{Rconverse}
\end{eqnarray}
\end{theorem}
The lower bounds provided by Theorem \ref{converse} hold for any scheme operating in our problem formulation setting, in both a star network or general network configuration. In the case the network is composed of a single node and $r=1$, they reduce to Chernoff's original results~\cite{{chernoff1959sequential}}.


\subsection{Upper bounds for proposed DCT and CCT schemes}
We now provide upper bounds on the performance of our schemes, starting with DCT.
In the following theorems,  the superscript ${\cal D}$ refers to the DCT. Part $(i)$ of Theorem \ref{DCT}  states that the probability of making a wrong decision can be made as small as desired by an appropriate choice of the observation cost $c$.
Part $(ii)$ provides a bound on the expected time to reach the final decision, and part $(iii)$ bounds  the risk as an immediate consequence of parts $(i)$ and $(ii)$.   Finally, part $(iv)$ presents the bound on the higher moments of the decision time   of DCT.  

\begin{theorem}\label{DCT} \emph{(Direct).}
The following statements hold: \\
$(i)$ For all $c \in (0,1)$ and all $ i\in [M]$, 
the probability that DCT makes an incorrect decision is 
\begin{equation}\label{eq:errorDCT}
    {\mathbb{P}}^{{\cal D}}_i(\hat{H}\neq h_i)\leq \min\{(M-1)c,1\}.
\end{equation}
$(ii)$ For all $\ell \in [L]$, $i,j\in[M]$ and $u\in S$, if $\mathbb{E}\big[\log {p_{i,\ell}^{u}(Y)}/{p_{j,\ell}^{u}(Y)}\big]^2 < \infty$, then
we have
\begin{equation} \label{EDCT1}
\mathbb{E}_{i}^{{\cal D}}[N]\leq (1+o(1))\frac{\abs{\log c}}{I(i)}, \;\; \mbox{ as }  c\to 0.
\end{equation}
$(iii)$ Combining $(i)$ and $(ii)$, we have
\begin{equation} \label{RDCT}
\mathbb{R}_{i}^{{\cal D}}\leq (1+o(1))\frac{c\,\abs{\log c}}{I(i)}, \;\; \mbox{ as }  c \to 0.
\end{equation}
$(iv)$ For all $\ell \in [L]$, $i,j\in[M]$, $u\in S$ and {all integers} $r\geq 2$, if $\mathbb{E}\big[\abs{\log {p_{i,\ell}^{u}(Y)}/{p_{j,\ell}^{u}(Y)}}^{r+1}\big] < \infty$, then we have
\begin{align}\label{eq:momentsDCT}\small
    &\mathbb{E}_{i}^{{\cal D}}[N^{r}]
      \leq \Bigg( (1+o(1))\frac{c\,\abs{\log c}}{I(i)}\Bigg)^{r}, \mbox{ as }  c \to 0.
\end{align}
\end{theorem}
In the above theorem, the bound on the expected decision time in $(ii)$ requires the second moment of the log-likelihood ratio to be finite. Likewise, for all $r\geq 2$, the bound on the $r^{th}$ moment of the decision time requires  the $r+1^{st}$ moment  of the log-likelihood ratio to be finite.

The next result is a   consequence of 
Theorems \ref{converse} and \ref{DCT}. It shows the asymptotic optimality of   DCT, and presents the expected communication overhead, as $c\to 0$.
\begin{theorem}\label{mainDCT} For any hypothesis testing scheme~$\delta$ operating over a network as described in Section~\ref{sec:ProblemFormulation},   we have\\
$(i)$ For all $\ell \in [L]$, $i,j\in[M]$ and $u\in S$, if $\mathbb{E}\big[\log {p_{i,\ell}^{u}(Y)}/{p_{j,\ell}^{u}(Y)}\big]^2 < \infty$, then
we have
\begin{equation} \label{EDCTM}
\lim_{c\to 0}\frac{\mathbb{E}_{i}^{{\cal D}}[N]}
{ \mathbb{E}_{i}^{{\delta}}[N]} \leq 1, \mbox{ and} 
\end{equation}

\begin{eqnarray}
\lim_{c\to 0} \frac{\mathbb{R}_{i}^{{\cal D}}}{  \mathbb{R}_{i}^{{\delta}}} & \leq & 1 \label{RfinalDCT}. \;\; 
\end{eqnarray}
$(ii)$ For all $\ell \in [L]$, $i,j\in[M]$, $u\in S$ and {all integers} $r\geq 2$, if $\mathbb{E}\big[\abs{\log {p_{i,\ell}^{u}(Y)}/{p_{j,\ell}^{u}(Y)}}^{r+1}\big] < \infty$, then we have
\begin{align} \label{EfinalDCT} \small
    &\lim_{c\to 0} \frac{\mathbb{E}_{i}^{{\cal D}}[N^{r}]}{  \mathbb{E}_{i}^{{\delta}}[N^{r}]}
      \leq 1.
\end{align}
$(iii)$ Assuming $C \geq M$,    and letting the communication overhead
$C_O$ be  the number of channel usages by each node,  we have
\begin{equation} \label{eq:average}
\lim_{c \rightarrow 0} \mathbb{E}^{{\cal D}}_{i}[C_O] = 4.
\end{equation}
\end{theorem}

Combining Theorem~\ref{DCT} and Theorem \ref{mainDCT}, it follows that DCT is asymptotically optimal in terms of stopping time and risk, as the observation cost tends to zero. 
This asymptotic optimality,  expressed by \eqref{EDCTM}, \eqref{RfinalDCT}, and \eqref{EfinalDCT},  holds for all values of $C$, although
in the case $C<M$  the expected number of channel uses per node in \eqref{eq:average} increases from four to a constant that is at most $2(M+1)$, due to multiple transmissions required to communicate each vector in the initialization phase.
We also point out  that the performance of DCT  depends only on the cumulative capability $I(i)$ of the network to detect hypothesis $h_{i}$, and is independent of how the capabilities $v_{i,\ell}$ are distributed over the network. If two networks have the same cumulative capabilities, then the minimum expected decision time will be the same for both of them. These results hold irrespective of the number of nodes in the network.

We now provide upper bounds on the performance of CCT. We make use of the following well known lemma:

\begin{lemma}\emph{\cite[Proposition 1]{dong2017flocking}.}\label{lemma:ErgodicCoeff} 
For any a connected graph $\mathcal{G(L,E)}$ with weights assigned to the edges satisfying \eqref{eq:FeasibilitySet}, we have that
\begin{equation*}
0<\eta(W^{h^{\mathcal{G}}})<1,
\end{equation*}
where
\[\eta(W)=\min_{i\neq j}\sum_{k=1}^{L}\min\{w_{i,k},w_{j,k}\}\]
is the ergodic coefficient   of  the weight matrix $W$. 
\end{lemma}

In the following theorems, the superscript ${\cal C}$ refers to the CCT. 
 Part $(i)$ of Theorem \ref{CCT} states that the probability of making a wrong decision can be made as small as desired by an appropriate choice of $c$.
Part $(ii)$ provides a bound on the expected time to reach the final decision, and part $(iii)$ bounds  the risk as an immediate consequence of parts $(i)$ and $(ii)$. Finally, part $(iv)$ presents the bound on the higher moments of the decision time of CCT. 

\begin{theorem}\label{CCT} \emph{(Direct).}
Assuming $C\geq M$, the following statements hold: \\
$(i)$ For all $c \in (0,1)$ and all $ i\in [M]$, 
the probability that CCT makes an incorrect decision is
\[{\mathbb{P}}^{{\cal C}}_i(\hat{H}\neq h_i)\leq \min\left \{(M-1)c^{\frac{1}{1+c/I(i)}},1\right \}.\] 
$(ii)$ For all $\ell \in [L]$, $i,j\in[M]$ and $u\in S$, if $\mathbb{E}\big[\log {p_{i,\ell}^{u}(Y)}/{p_{j,\ell}^{u}(Y)}\big]^2 < \infty$, then as $c \rightarrow 0$ we have

\begin{equation} \label{EDCT}
\begin{split}
    &\mathbb{E}_{i}^{{\cal C}}[N]
    \leq (1+o(1)) \Bigg(\frac{h^{\mathcal{G}}\cdot |\log(c/\max_{j\in [L]} I(j))|}{|\log \big(1-\eta(W^{h^{\mathcal{G}}})\big)|}+\frac{\abs{\log c}}{I(i)-c}\Bigg
    ). 
\end{split}
\end{equation}
$(iii)$ Combining $(i)$ and $(ii)$,  as $c \rightarrow 0$ we have 
\begin{equation} 
\begin{split}
  \mathbb{R}_{i}^{{\cal C}}&\leq\hspace{-2pt} (1+o(1))\Bigg(\hspace{-3pt}\frac{h^{\mathcal{G}}\cdot \abs{\log(1/\max_{j\in [L]} I(j))}}{\abs{\log (1-\eta(W^{h^{\mathcal{G}}}))}}+\frac{1}{I(i)-c}\hspace{-3pt}\Bigg)\cdot c\abs{\log c}.  \qquad
\end{split}
\end{equation} \\
$(iv)$ For all $\ell \in [L]$, $i,j\in[M]$, $u\in S$ and all integers $r\geq 2$, if $\mathbb{E}\big[\abs{\log {p_{i,\ell}^{u}(Y)}/{p_{j,\ell}^{u}(Y)}}^{r+1}\big] < \infty$, then as $c \rightarrow 0$ we have  \begin{align}\label{eq:moments}\small
    &\mathbb{E}_{i}^{{\cal C}}[N^{r}]\nonumber\\
      &\leq\hspace{-3pt} \Bigg(\hspace{-3pt} (1+o(1)) \Bigg(\frac{h^{\mathcal{G}}\cdot |\log(c/\max_{j\in [L]} I(j))|}{|\log (1-\eta(W^{h^{\mathcal{G}}}))|}+\frac{\abs{\log c}}{I(i)-c}\Bigg)\hspace{-3pt}\Bigg)^{r}\hspace{-3pt}.
\end{align}

\end{theorem}
The following result is a consequence of Theorems \ref{converse} and \ref{CCT}, and shows that CCT is asymptotically optimal, up to a constant factor, as the observation cost tends to zero. 

{\begin{theorem}\label{main2}
For any hypothesis testing scheme~$\delta$ operating over a network as described in Section~\ref{sec:ProblemFormulation} and assuming $C\geq M$, we have: \\
$(i)$ For all $\ell \in [L]$, $i,j\in[M]$ and $u\in S$, if $\mathbb{E}\big[\log {p_{i,\ell}^{u}(Y)}/{p_{j,\ell}^{u}(Y)}\big]^2 < \infty$, then 
\begin{equation} \label{ECCT}
\begin{split}
    &\lim_{c\to 0}\frac{\mathbb{E}_{i}^{{\cal C}}[N]}{ \mathbb{E}_{i}^{{\delta}}[N]}\leq \Bigg(\frac{I(i)h^{\mathcal{G}}\cdot \abs{\log(1/\max_{j\in [L]} I(j))}}{\abs{\log (1-\eta(W^{h^{\mathcal{G}}}))}}+1\hspace{-3pt}\Bigg). \;\;
\end{split}
\end{equation}
Additionally, 
\begin{eqnarray}
\lim_{c\to 0}\frac{\mathbb{R}_{i}^{{\cal C}}}{ \mathbb{R}_{i}^{{\delta}}}&\leq& \Bigg(\frac{I(i)h^{\mathcal{G}}\cdot \abs{\log(1/\max_{j\in [L]} I(j))}}{\abs{\log (1-\eta(W^{h^{\mathcal{G}}}))}}+1\hspace{-3pt}\Bigg) \label{RfinalCCT} .
\end{eqnarray}
\\
$(ii)$ For all $\ell \in [L]$, $i,j\in[M]$, $u\in S$ and all integers $r\geq 2$, if $\mathbb{E}\big[\abs{\log {p_{i,\ell}^{u}(Y)}/{p_{j,\ell}^{u}(Y)}}^{r+1}\big] < \infty$, then we have
\begin{eqnarray}
\lim_{c\to 0}\frac{\mathbb{E}^{{\cal C}}_{i}[N^r]
}{ \mathbb{E}^{{\delta}}_{i}[N^r]}&\leq & \Bigg(\frac{I(i)h^{\mathcal{G}}\cdot \abs{\log(1/\max_{j\in [L]} I(j))}}{\abs{\log (1-\eta(W^{h^{\mathcal{G}}}))}}+1\hspace{-3pt}\Bigg)^{r} \label{EfinalCCT} . 
\end{eqnarray}
\end{theorem}}

While Theorems~\ref{CCT} and \ref{main2} provide bounds for the case $C\geq M$, it should be clear from their proof that when  $C<M$  CCT is still asymptotically optimal up to a constant factor, as $c \rightarrow 0$. In this case, the right-hand sides of \eqref{ECCT} and \eqref{RfinalCCT}   are simply scaled by an additional factor that is upper bounded by $M$, due to the multiple transmissions required to complete each vector transmission.  Similarly, the right-hand side of \eqref{EfinalCCT} is scaled by a factor upper bounded by $M^r$.

The decision time of CCT (see (\ref{EDCT}) and (\ref{eq:moments})) depends on two terms: $A_{1}$ and $A_{2}$, where 
\[A_{1}=\frac{h^{\mathcal{G}}\cdot |\log(c/\max_{j\in [L]} I(j))|}{|\log \big(1-\eta(W^{h^{\mathcal{G}}})\big)|},\]
\[A_{2}=\frac{\abs{\log c}}{I(i)-c}.\]
Here, $A_{1}$ corresponds to the expected time of the initialization phase. Since this phase performs consensus over the network, this  time depends on the network parameters $h^{\mathcal{G}}$ and matrix $W$. Similarly, $A_{2}$ corresponds to the expected time of the test phase, where the Chernoff test is performed independently at all the nodes. This  time is independent of the network parameters.  Finally, since the decision phase of CCT begins only after the termination of  the initialization phase and is dependent on the test phase,  the expected decision time of CCT   depends on $A_{1}+A_{2}$. Thus, in Theorem \ref{main2}, the ratio of the performance parameters of CCT and of the optimal test converges to the constant $1+{I(i)h^{\mathcal{G}}\cdot \abs{\log(1/\max_{j\in [L]} I(j))}}/{\abs{\log (1-\eta(W^{h^{\mathcal{G}}}))}}$. It follows that 
the gap between the  CCT and the optimal test is given by   ${I(i)h^{\mathcal{G}}\cdot \abs{\log(1/\max_{j\in [L]} I(j))}}/{\abs{\log (1-\eta(W^{h^{\mathcal{G}}}))}}$, and as the expected time of initialization phase decreases this gap decreases. 


As a final remark, we point out that the star network configuration is a  special case of the distributed setup. In this case, the cumulative capability vector $I$ can be estimated {(with no error)} in two time steps at all the nodes, i.e., for $n=2$ and $\ell,j \in [L]$, the equivalent of  (\ref{eq:error}) is 
\begin{equation}
  \abs{\hat{I}^{n}_{\ell}-\hat{I}^{n}_{j}}\preceq 0\cdot\textbf{{1}}_{1\times M},
\end{equation}
and is independent of~$c$. In the regime of vanishing cost $c\to 0$, we have that $A_{1}+A_{2}=2+A_{2}$, which implies the asymptotic optimality of DCT. 

\section{Numerical Results}\label{sec:Simulations}
In this section, we evaluate the performance of both DCT and CCT by simulations, and compare the results to the theoretical bounds presented in the previous section. The performance of these tests is evaluated for different sizes of networks. In our experiments, the number of hypotheses is $M=3$. The probability distribution $p_{i,\ell}^{u}$, \ {$i=1,\dots,M$,} is 
Bernoulli with parameter~$p$, which is selected uniformly at random from $(0,1/3)$,$(1/3,2/3)$ and $(2/3,1)$  for $i=1,2$ and $3$ respectively.  
\begin{figure*}
\centering
\begin{minipage}{0.48\textwidth}
  \centering
  \includegraphics[width=\linewidth]{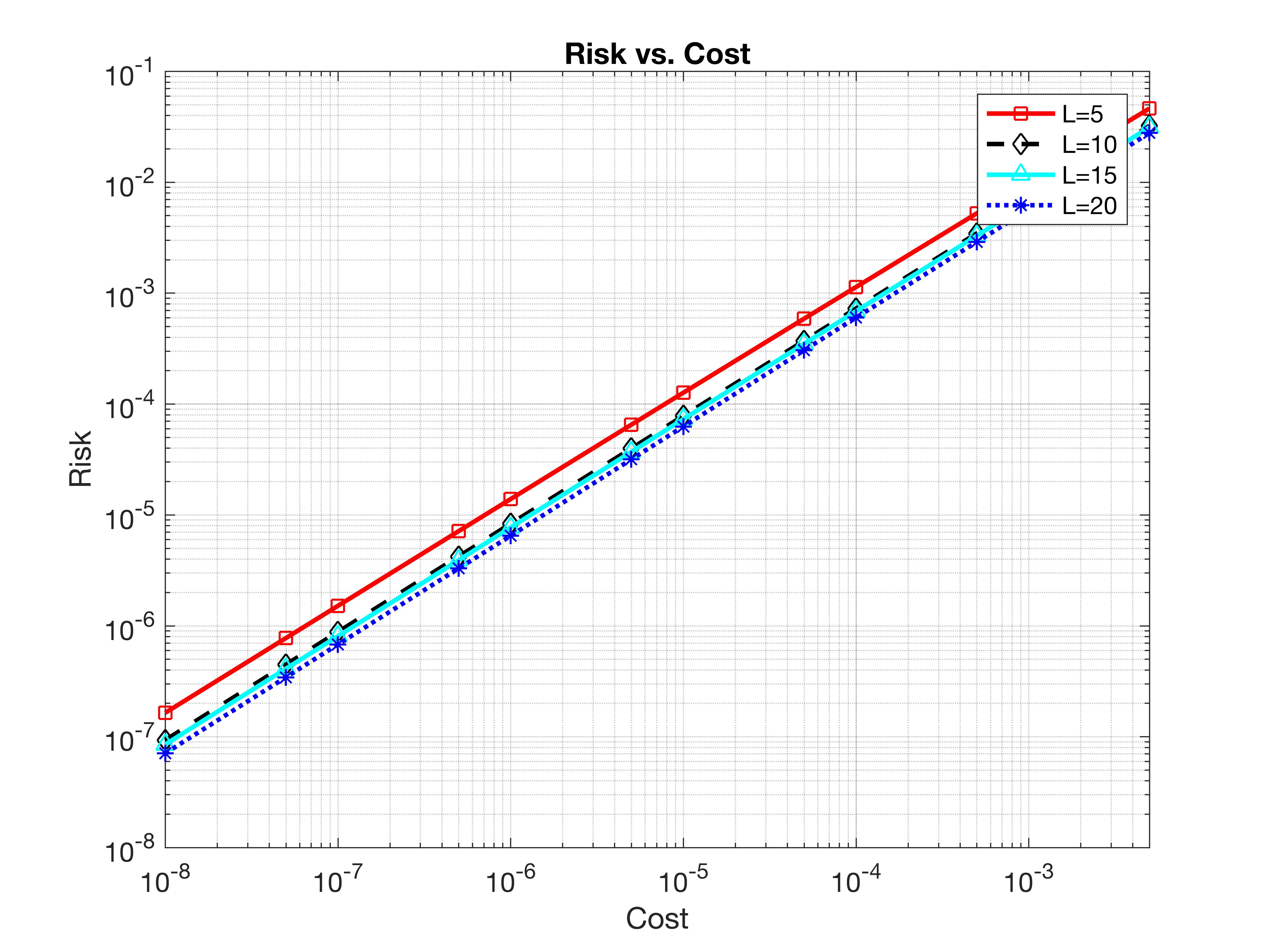}
  \caption{Performance of DCT: risk vs. cost $c$ for different number of sensors~$L$}
  \label{fig:DCT1}
\end{minipage} \hspace{5pt}
\begin{minipage}{.48\textwidth}
  \centering
  \includegraphics[width=\linewidth]{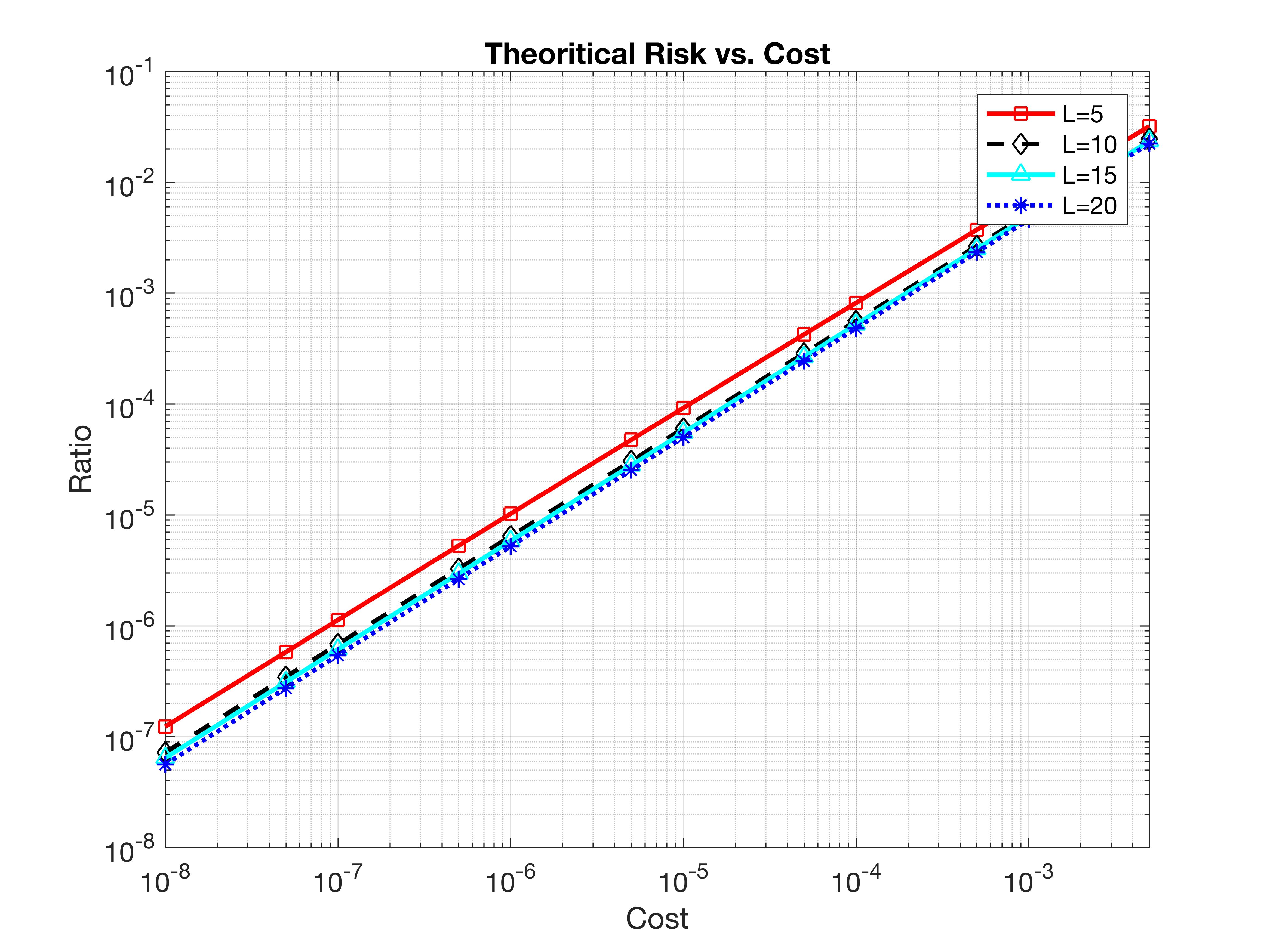}
  \caption{Performance of DCT according to Theorem \ref{DCT}: risk vs. cost $c$ for different number of sensors~$L$}
  \label{fig:DCT2}
\end{minipage}
\end{figure*}

\begin{figure}
\includegraphics[width=.5\linewidth]{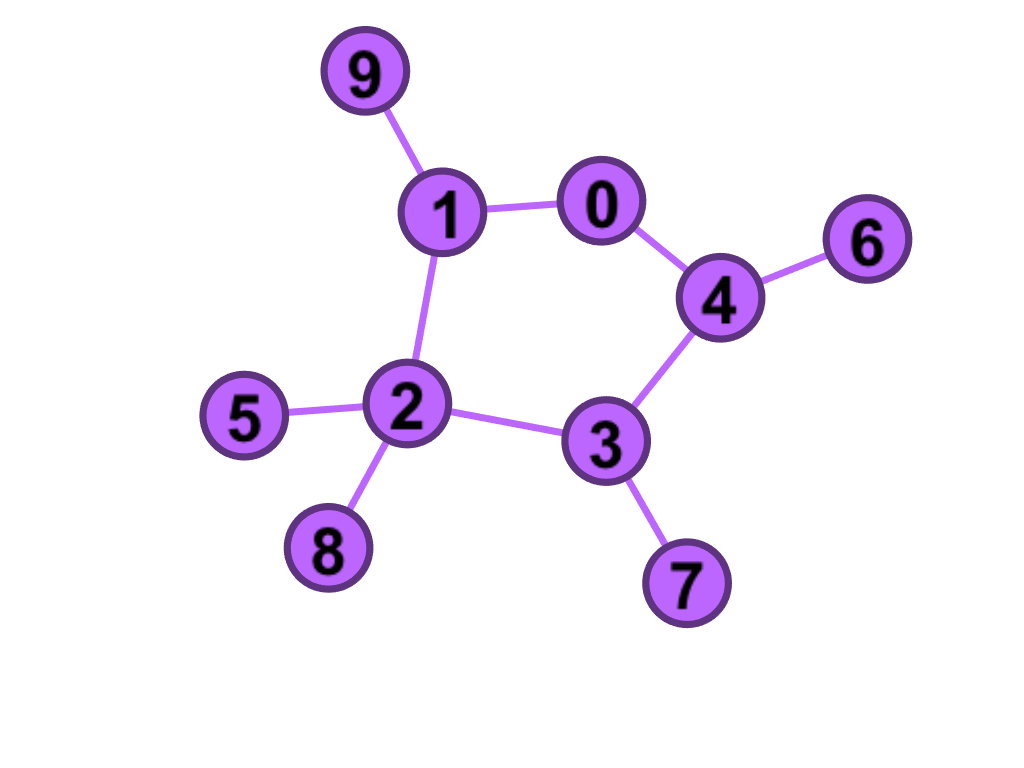}
\caption{An example of sensor network with $L=10$ nodes.}
\label{fig:performance2}
\end{figure}

\begin{figure*}
\centering
\begin{minipage}{0.48\textwidth}
  \centering
  \includegraphics[width=\linewidth]{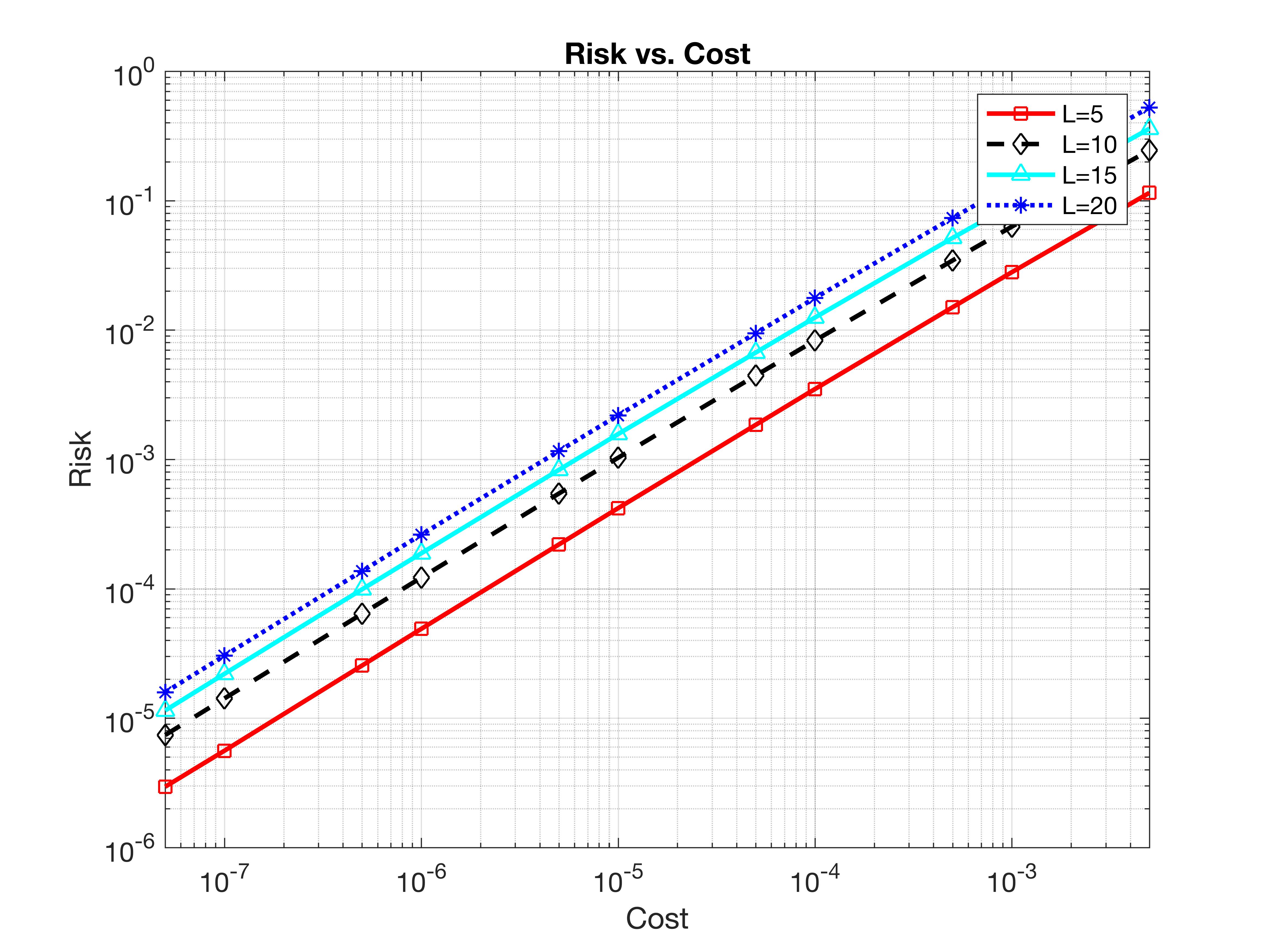}
  \caption{Performance of CCT for the ring with random attachments topology: risk vs. cost $c$ for different number of sensors~$L$}
  \label{fig:CCT1}
\end{minipage} \hspace{5pt}
\begin{minipage}{.48\textwidth}
  \centering
  \includegraphics[width=\linewidth]{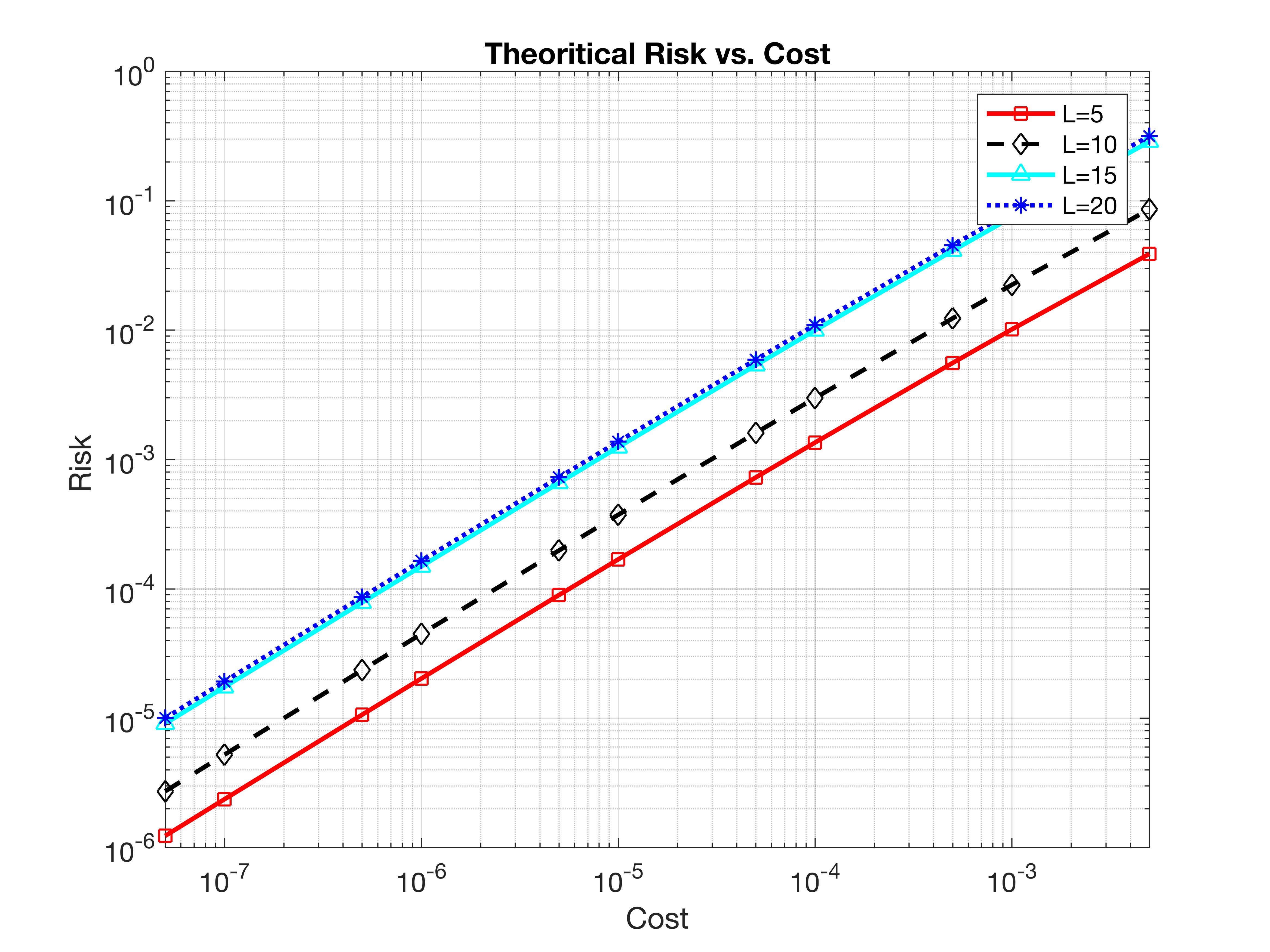}
  \caption{Performance of CCT according to Theorem \ref{CCT} for the ring with random attachments topology: risk vs. cost $c$ for different number of sensors~$L$}
  \label{fig:CCT2}
\end{minipage}
\end{figure*}
\begin{figure*}
\centering
\begin{minipage}{0.48\textwidth}
  \centering
  \includegraphics[width=\linewidth]{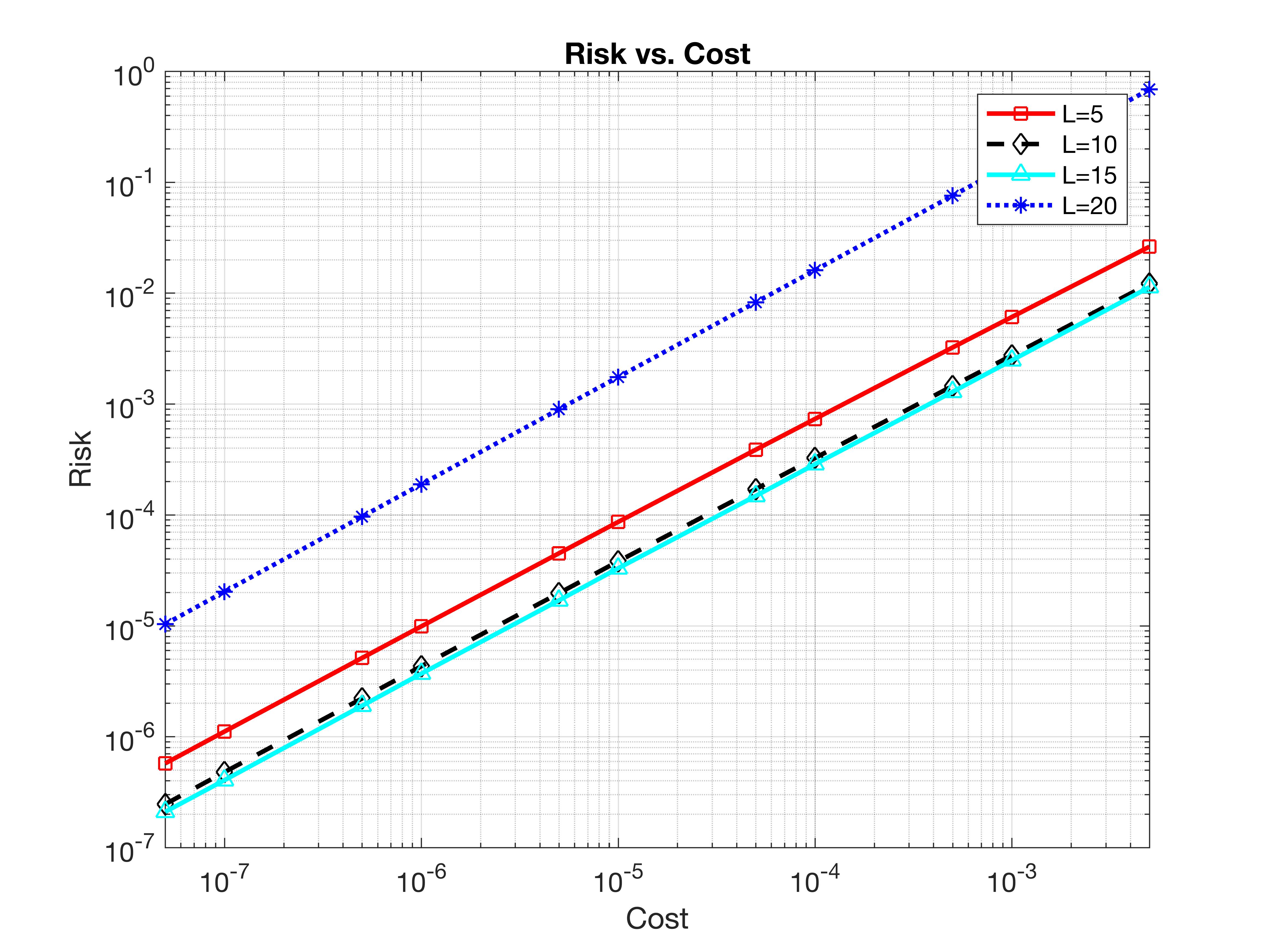}
  \caption{Performance of CCT for the tree topology: risk vs. cost $c$ for different number of sensors $L$}
  \label{fig:CCT3}
\end{minipage} \hspace{5pt}
\begin{minipage}{.48\textwidth}
  \centering
  \includegraphics[width=\linewidth]{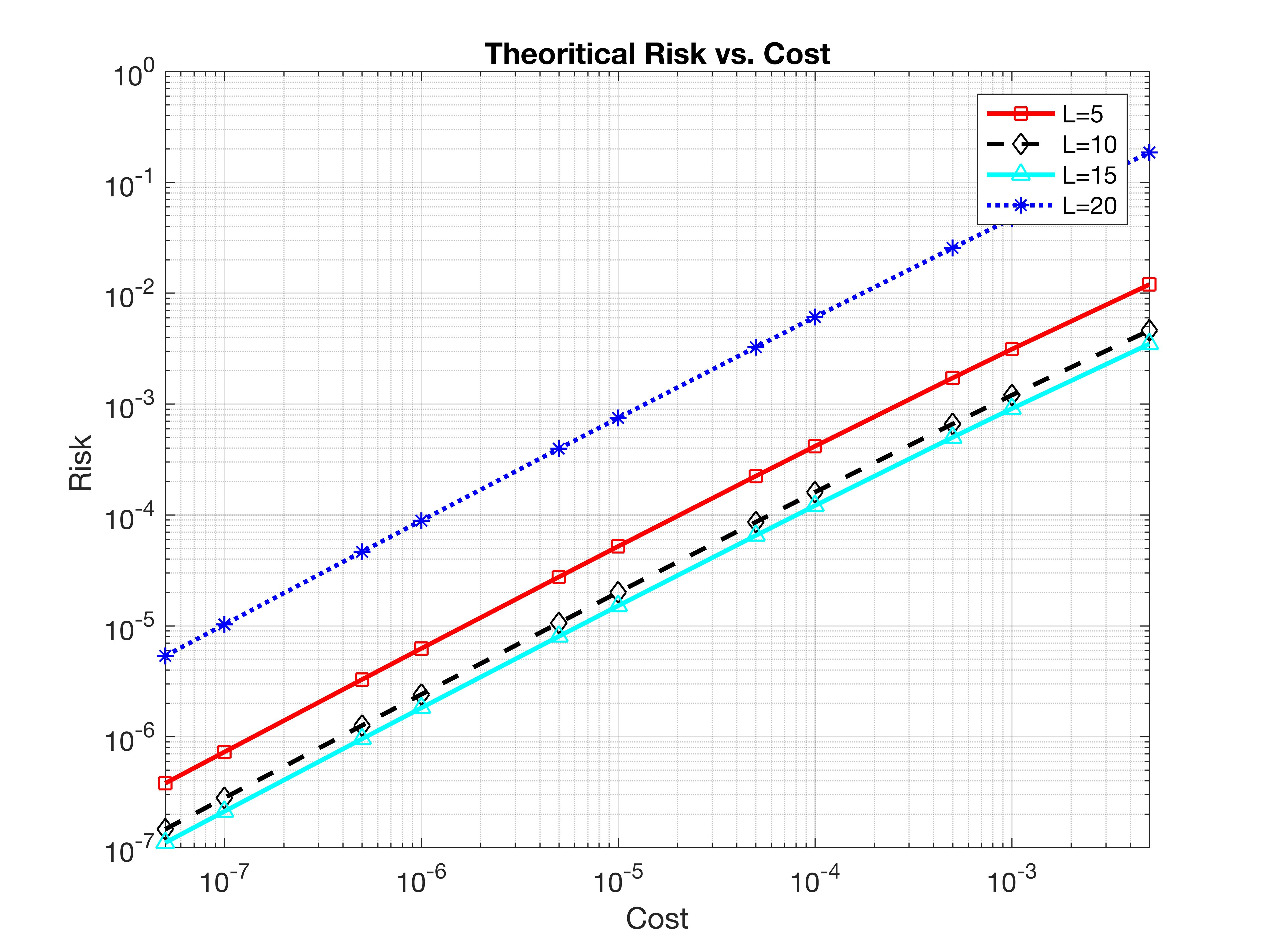}
  \caption{Performance of CCT according to Theorem \ref{CCT} for the tree topology: risk vs. cost $c$ for different number of sensors $L$}
  \label{fig:CCT4}
\end{minipage}
\end{figure*}

Figure~\ref{fig:DCT1} shows the risk of DCT in a fusion center based setup, as obtained by simulations. Figure~\ref{fig:DCT2} shows the corresponding value of the risk, as predicted by Theorem \ref{DCT}. The risk decreases as the observation cost~$c$ decreases. This is because the threshold in the triggering condition~(\ref{StoppingRuleDec11}) increases, which ensures that the nodes have a greater confidence about their local decision. On the other hand, the risk decreases by increasing the number of sensors $L$, because the cumulative capability of the network to detect the hypothesis, defined in (\ref{eq:CC}),  increases with $L$, and the task of hypothesis testing is divided among a larger number of sensors. Hence, the final decision can be reached more quickly, and this decreases the risk. The trends are in agreement with the theoretical results obtained for DCT. 

Our simulations also confirm the prediction that, on the average, only four channel usages 
are required, per single sensor, see (\ref{eq:average}) in Theorem \ref{mainDCT}. The results of these simulations are not reported here for the sake of brevity. We only mention that, on rare occasions, for individual realizations it may happen that the number of channel usages  is substantially larger than four ---a manifestation of the long-run phenomenon \cite[p.~110]{poorbook}. In practice, this can be remedied by resorting to a truncated version of the sequential test, for which the maximum number of probing actions is fixed, see \cite{tantaratana-poor-82,tantaratana-93} and references therein for a discussion, and see \cite{marano-03} for a simple implementation of truncation. A precise analysis of DCT using truncated tests is out of the scope of the present paper.

The performance of CCT is evaluated for two network topologies. In the first topology, given the number of network nodes $L$, $\lceil L/2 \rceil$ sensors are connected to form a ring, and the remaining sensors are randomly connected to the sensors in the ring. An example of network with $L=10$ is shown in Figure~\ref{fig:performance2}. In this case, the spanning height of the tree is linear in $L$. In the second topology, 
given the number of network nodes $L$, the nodes are connected to form a binary tree. In this case, the spanning height of the tree is  $\mathcal{O}(\log_2 L)$.

Figure \ref{fig:CCT1} shows the performance of CCT for the ring with random attachments, obtained by computer simulations. Figure~\ref{fig:CCT2} shows the value of risk according to Theorem \ref{CCT}. Like in the case of DCT, the risk of CCT decreases as the observation cost $c$ decreases. Instead, the behavior as function of $L$  is different. Unlike DCT, the risk of CCT increases by increasing the number of network nodes $L$. This effect can be explained by observing that in 
CCT there is a trade-off between the time required by the initialization phase and the time required by the test phase. For the considered network $\mathcal{G}$ and consensus matrix $W$, as the number of nodes $L$ increases, the consensus scheme in the initialization phase will require more time in comparison to the test phase. Additionally, the time required by the test phase decreases with $L$, for the same reasons as in the DCT case. 
Figures \ref{fig:CCT1} and \ref{fig:CCT2} show that the consensus between the sensors in the first phase of CCT becomes the dominating factor in the decision time. This is in agreement with the theoretical bounds provided in Theorem~\ref{CCT}. 

Figures \ref{fig:CCT3} and \ref{fig:CCT4} show the performance of CCT for the tree topology, via simulations and using the theoretical predictions of Theorem \ref{CCT}, respectively. The risk of CCT decreases as $c$ decreases. Unlike the ring topology with random attachments, the risk decreases by increasing $L$ until $L=15$, and then increases. In this setup, for the initial values of $L$, the time required by the test phase is larger than the time for the initialization phase, hence, it is the dominating factor in the decision time of CCT. 
On the contrary, for $L=20$, the time of the initialization phase becomes dominant, which explains why the risk increases with $L$.  
Finally, comparing Figures \ref{fig:CCT3} and \ref{fig:CCT4}, we see that the theoretical values of the risk are close to the results of numerical simulations. 

\section{ Extension to Channels with Quantized Messages and Link Failures } \label{sec:nonideal}
In the previous sections, we have assumed a communication model carrying  real numbers over ideal links, without errors. This models a situation where  transmission are finely quantized and adequately protected 
against errors. We now wish to explicitly take into account  
the effect of data quantization, 
and of  link failures leading to packet erasures. 

\subsection{ {Channels with Quantized Messages}}
We start by considering channels supporting  quantized messages, rather than  real numbers, as described in Section~\ref{sec:ProblemFormulation}.
We extend our previous results   by describing the key changes to both DCT and CCT  formulations.
\subsubsection{DCT with quantized messages} \label{subsDCT}
In the initialization phase, the vectors $v_{\ell}$ and $I$ 
need to be quantized using $C$ bits before transmission. Accordingly, at the sensor nodes we construct the quantized vector    $\floor{v_{\ell}}=[\floor{v_{1,\ell}},\ldots,\floor{v_{M,\ell}}]$ and at the fusion center we construct the  corresponding vector $\floor{I} = [\floor{I(1)},\dots, \floor{I(M)}]$.
By \eqref{eq:CC}  we have that for all $i\in[M]$ and $\ell\in[L]$,  $v_{i,\ell}\leq I(i)$. It follows that to construct the first vector we can divide the interval $[0,\max_{i}I(i)]$  uniformly into $Q$  sub-intervals, where $Q=2^{C/M}$, and 
let  $\floor{v_{i,\ell}}$  be the nearest  value among the $Q$ quantization levels smaller than $v_{i\ell}$. 
In this way, the difference between any two   contiguous quantization levels for $v_{i,\ell}$ is  
\begin{equation}\label{eq:QuantizationStep}
    \Delta \Big (\max_{i}I(i), Q \Big )=\frac{\max_{i}I(i)}{Q}. 
\end{equation}
 The quantized vector $\floor{v_{\ell}}=[\floor{v_{1,\ell}},\ldots,\floor{v_{M,\ell}}]$ is then sent by each node to the fusion center using $M\cdot \log_2 Q=C$  bits in one transmission.
On the other hand, for the second vector we let, for all $i \in [M]$
\begin{equation}\label{eq:newCC}
    \floor{I(i)}=\sum_{\ell =1}^{L}\floor{v_{i,\ell}}.
\end{equation}
Since $\floor{v_{i,\ell}}$ lies in the interval $[0,\max_{i}I(i)]$ and $\sum_{\ell=1}^Lv_{i,\ell} =  I(i)$, then $\floor{I(i)}$   also corresponds to
a quantization level of the interval $[0,\max_{i}I(i)]$ when this is uniformly divided into $Q$  sub-intervals. It follows that 
the fusion center can send the vector $\floor{I}$ to each node using $C$ bits in one transmission.  

Upon reception of $\floor{I}$ from the fusion center, every node $\ell$ computes a   vector $\rho_{\ell}=[\rho_{1,\ell},\dots,\rho_{M,\ell}]$, where for all $i\in [M]$ 
\begin{equation}\label{eq:newrho}
    \rho_{i,\ell}=\frac{v_{i,\ell}}{\sum_{\Tilde{\ell}=1}^{L}\floor{v_{i,\Tilde\ell}}}=\frac{v_{i,\ell}}{\floor{I(i)}}, 
\end{equation}
and uses it in the test phase for the determination of the threshold in  (\ref{StoppingRuleDec11}). 
In the test phase, each local preference  can be communicated using $\log_{2}M$ bits   and in the stopping phase the halting message can be communicated using a single bit. 


\subsubsection{CCT with quantized messages} \label{subsCCT}
In the initialization phase, we need to send   $z_{\ell}^n$  and  $\hat{I}_{\ell}^{n}$    over the channel at each transmission using $C$ bits.
Since the initialization phase terminates when $z^n_{\ell}>L+1$ (see Algorithm \ref{alg:Phase1}), it follows that at most $\log_{2}(L+2)$ bits are needed to communicate $z^n_{\ell}$. The remaining bits $\tilde{C}=C-\log_{2}(L+2)$ bits can then used to communicate the vector $\hat{I}_{\ell}^{n}$. 
Similar to DCT, we divide the interval $[0,\max_i I(i)]$ uniformly into $\tilde{Q}=2^{\tilde{C}/M}$ sub-intervals so that the  difference between any two adjacent quantization levels is
\begin{equation}\label{eq:quantizedConsensusLevel}
    \Delta \Big(\max_i I(i), \tilde{Q} \Big)=\frac{ \max_i I(i)}{\tilde{Q}}. 
\end{equation}
We let the initial estimate $\hat{I}_{\ell}^{0}=[\floor{v_{1,\ell}},\ldots,\floor{v_{M,\ell}}]$, where $\floor{v_{i,\ell}}$  is the nearest lower value among the $\tilde{Q}$ quantization levels representing $v_{i,\ell}$. The consensus protocol is then modified as follows:
\begin{equation}\label{eq:quantizedConsensus}
    \hat{I}_{\ell}^{n+1}= \Big \lfloor w_{\ell,\ell}\cdot \hat{I}_{\ell}^{n} +\sum_{j \in  \mathcal{N}_{\ell}}w_{\ell,j}\cdot \hat{I}_{j}^{n} \Big \rfloor.
\end{equation}
It follows that every node $\ell$ performs a  convex combination of the quantized self-estimate $\hat{I}_{\ell}^{n}$ and the quantized estimates $\{\hat{I}_{j}^{n}\}_{j \in {\cal N}_\ell}$ from its neighbors and the updated estimate $\hat{I}_{\ell}^{n+1}$ is a quantized version of this convex combination. The stopping rule of the initialization phase  remains the same as stated in Algorithm \ref{alg:Phase1}. In the following phases, we let $\floor{\hat{I}_{\ell}}=[\floor{\hat{I}_{\ell}(1)},\ldots, \floor{\hat{I}_{\ell}(M)}]$ denote  the estimate of vector $I$ using \eqref{eq:quantizedConsensus} at node $\ell$ at the end of the initialization phase.

In the test phase of CCT, the SCT is performed locally  using the result of the consensus algorithm to select the threshold, namely $\gamma=\hat{\rho}_{i^*_n,\ell}|\log c|$ and $\hat{\rho}_{i^*_n,\ell}=v_{i^*_n,\ell}/\floor{\hat{I}_{\ell}(i^*_n)}$. 
Finally, in the stopping phase of CCT (see Algorithm \ref{alg:Phase3}), the variable $d_{\ell}^{n}$ and local decision $\hat{H}^{n}_{\ell}$ are communicated over the channel. Since the 
stopping phase terminates when $d_{\ell}^{n}>L+1$, {no more than $\log_2(L+2)$ bits are needed} to communicate~$d_{\ell}^{n}$. 
The local preference $\hat{H}^{n}_{\ell}$ can be communicated by $\log_2 M$ bits.

\subsection{Performance analysis for {Channels with Quantized Messages}} \label{substheory}
In this section, we extend the results in Theorem \ref{DCT} and Theorem \ref{CCT} to channels 
with {quantized messages.} 

\begin{theorem}\label{DCT_quantized} \emph{(Direct).} Letting 
\begin{equation} \label{eq:fun}
    f(Q)=\frac{L\max_{i} I(i)}{Q },
\end{equation}
and assuming $C$ is sufficiently large such that for all $i\in [M]$, $f(Q)<I(i)$, the following statements hold for DCT: 
 \\
$(i)$ For all $c \in (0,1)$ and all $ i\in [M]$, 
the probability that the DCT takes an incorrect decision is 
\[{\mathbb{P}}^{{\cal D}}_i(\hat{H}\neq h_i)\leq \min\{(M-1)c,1\}.\]
\\
$(ii)$ For all $\ell \in [L]$, $i,j\in[M]$ and $u\in S$, if $\mathbb{E}\big[\log {p_{i,\ell}^{u}(Y)}/{p_{j,\ell}^{u}(Y)}\big]^2 < \infty$, then
the expected decision time is
\begin{equation} \label{EDCT1_1}
\mathbb{E}_{i}^{{\cal D}}[N]\leq (1+o(1))\frac{\abs{\log c}}{I(i)-f(Q)}, \;\; \mbox{ as }  c\to 0.
\end{equation}
$(iii)$ Combining $(i)$ and $(ii)$, the risk defined in~(\ref{risk}) is
\begin{equation} \label{RDCT_1}
\mathbb{R}_{i}^{{\cal D}}\leq (1+o(1))\frac{c\,\abs{\log c}}{I(i)-f(Q)}, \;\; \mbox{ as }  c \to 0.
\end{equation}
$(iv)$ For all $\ell \in [L]$, $i,j\in[M]$, $u\in S$ and {all integers} $r\geq 2$, if $\mathbb{E}\big[\abs{\log {p_{i,\ell}^{u}(Y)}/{p_{j,\ell}^{u}(Y)}}^{r+1}\big] < \infty$, then the $r^{th}$ moment of the decision time $N$ is 
\begin{align}\label{eq:momentsDCT_1}\small
    &\mathbb{E}_{i}^{{\cal D}}[N^{r}]
      \leq \Bigg( (1+o(1))\frac{c\,\abs{\log c}}{I(i)-f(Q)}\Bigg)^{r}, \mbox{ as }  c \to 0.
\end{align}
\end{theorem}

By Theorem \ref{DCT_quantized}, it follows that the performance of DCT  depends on the number  of quantization levels through the function $f(Q)$. As $Q\to \infty$, we have that $f(Q)\to 0$ and the results of Theorem~\ref{DCT} are recovered. Since as $Q\to \infty$, real numbers can be communicated perfectly over the channels, in this case DCT incurs  no loss of asymptotic performance. We can then view $f(Q)$ as quantifying the loss in performance of DCT due to quantization. This is also evident by combining \eqref{eq:QuantizationStep}  and \eqref{eq:newCC}, which show that the  quantization error      $|I(i) - \floor{I(i)}|$ is at most  $f(Q)$. Our theorem statement, by assuming that $f(Q)<I(i)$, ensures that  this error is smaller than $I(i)$. Since $Q= 2^{C/M}$, this constraint can be satisfied by having $C$ sufficiently large.  

Next, we consider the CCT case. We make the following assumptions that are commonly adopted in the literature of consensus over channels with quantized messages.

\begin{assumption}\emph{\cite[{Assumption 1}]{nedic2009distributed}}\label{assumption:capacity} The matrix $W$ is doubly stochastic, i.e. (\ref{eq:SC11}) and (\ref{eq:SC22})  holds, with positive diagonal entries. In addition, there exists a constant $\alpha>0$ such that if $w_{i,j}>0$, then $w_{i,j}>\alpha$. 
\end{assumption}
The double stochastic assumption on the weight matrix $W$ guarantees that the average of the sensor values remains the same at each consensus iteration. 
The second part of Assumption \ref{assumption:capacity} ensures that each sensor gives a non-negligible 
weight to its values and to the values of its neighbors at each time.
\begin{assumption}\emph{\cite[{Assumption 4}]{nedic2009distributed}}\label{assumption:initialValue} For all $\ell$ and $i$, $v_{i,\ell}$ is a multiple of $M/\tilde{Q}$. 
\end{assumption}

The above assumption states that the  values of vector $\hat{I}_{\ell}^{0}$ are already quantized, namely $ \hat{I}_{\ell}^{0}=[\floor{v_{1,\ell}},\ldots,\floor{v_{M,\ell}}] = [{v_{1,\ell}},\ldots,{v_{M,\ell}}]$. 


\begin{theorem}\label{CCT_quantized} \emph{(Direct).}
Let 
\begin{equation} \label{eq:Qfun}
    g(Q,c,\alpha)=\frac{L}{Q}\bigg( 2\frac{L^2}{\alpha}\log(\min(Q^2,L^4/ c^2)\cdot \max_j I^2(j))+1+h^{\mathcal{G}}\Bigg(\frac{-\log(d^{\mathcal{G}})}{\log\big(1-\eta(W^{h^{\mathcal{G}}})\big)}+1\Bigg)+L+1\bigg).
\end{equation}
Assume that $C$ is sufficiently large such that for all $i\in [M]$, $g(Q,c,\alpha)<I(i)$ and $C>\log_2(L+2)+\log_2 M$,
and Assumptions \ref{assumption:capacity} and \ref{assumption:initialValue} hold. Then, we have: \\
$(i)$ For all $c \in (0,1)$ and all $ i\in [M]$, 
the probability that CCT takes an incorrect decision is
\[{\mathbb{P}}^{{\cal C}}_i(\hat{H}\neq h_i)\leq \min\left \{(M-1)c^{\frac{I(i)}{I(i)+g(\tilde{Q},c,\alpha)}},1\right \}.\]
$(ii)$ For all $\ell \in [L]$, $i,j\in[M]$ and $u\in S$, if $\mathbb{E}\big[\log {p_{i,\ell}^{u}(Y)}/{p_{j,\ell}^{u}(Y)}\big]^2 < \infty$, then 
the expected decision time is
\begin{equation} \label{EDCT2}
\begin{split}
    &\mathbb{E}_{i}^{{\cal C}}[N]
    \leq (1+o(1)) \Bigg(\tilde Q\cdot g(\tilde Q,c,\alpha)+\frac{\abs{\log c}}{I(i)-g(\tilde Q,c,\alpha)}\Bigg), 
\end{split}
\end{equation}
as $c \to 0$.\\
$(iii)$ Combining $(i)$ and $(ii)$, the risk defined in~(\ref{risk}) is
\begin{equation} \label{RDCT2}
\begin{split}
  \mathbb{R}_{i}^{{\cal C}}&\leq\hspace{-2pt} (1+o(1))\Bigg(\frac{\tilde Q\cdot g(\tilde Q,c,\alpha)}{c}+\frac{1}{I(i)-g(\tilde Q,c,\alpha)}\hspace{-3pt}\Bigg)
\cdot c \,\abs{\log c},  \qquad
\end{split}
\end{equation} \\
\textnormal{\emph{as} $c\to 0$}. \\
$(iv)$ For all $\ell \in [L]$, $i,j\in[M]$, $u\in S$ and all integers $r\geq 2$, if $\mathbb{E}\big[\abs{\log {p_{i,\ell}^{u}(Y)}/{p_{j,\ell}^{u}(Y)}}^{r+1}\big] < \infty$, then the $r^{th}$ moment of the expected decision time is 
\begin{align}\label{eq:moments}\small
    &\mathbb{E}_{i}^{{\cal C}}[N^{r}]\nonumber
      \leq\hspace{-3pt} \Bigg[\hspace{-3pt} (1+o(1)) \Bigg(\tilde Q\cdot g(\tilde Q,c,\alpha)+\frac{\abs{\log c}}{I(i)-g(\tilde Q,c,\alpha)}\Bigg)\hspace{-3pt}\Bigg]^{r}\hspace{-3pt},
\end{align}
as $c\to 0$.\\
\end{theorem}

By Theorem \ref{CCT_quantized}, it follows that the performance of CCT  depends on the number $\tilde Q$ of quantization levels through {the function $g(\tilde Q,c,\alpha)$.} As $Q\to \infty$, $\tilde Q\to\infty$ and $g(\tilde Q,c,\alpha)\to 0$. The time {required by the} initialization phase is given by $\Tilde{Q}\cdot g(\tilde Q,c,\alpha)=\mathcal{O}(|\log(c)|)$ as $Q\to\infty$, which is of the same order as  ${h^{\mathcal{G}}\cdot |\log(c/\max_{j\in [L]} I(j))|}/{\log \big(1-\eta(W^{h^{\mathcal{G}}})\big)}=\mathcal{O}(|\log(c)|)$ 
{appearing in} Theorem \ref{CCT}.  
As $Q\to\infty$, Theorem~\ref{CCT_quantized} {retcovers the same optimality of CCT expressed by~Theorem \ref{main2}.} 
In conclusion, $g(\tilde Q,c,\alpha)$   quantifies, in terms of the relevant system parameters, the loss in asymptotic performance of CCT due to quantization. In this case,  the  error for    $|I(i) - \floor{\hat{I}_{\ell}(i)}|$ is at most  $g(\tilde Q,c,\alpha)$ and our theorem   assumes that this error is smaller than $I(i)$. Since $\tilde{Q}= 2^{\tilde C/M}$, this constraint can be satisfied by having $C$ sufficiently large. The additional capacity constraint $C>\log_2(L+2)+\log_2 M$ in the statement of the theorem is due to the transmission of $d_\ell^n$ and the local preference $\hat{H}^{n}_{\ell}$. 

\subsection{{Channels with Quantized Messages and Erasures}}
{In this section, we consider both quantized channels and $\epsilon$-random packet erasures, as described in Section \ref{sec:ProblemFormulation}. We extend our previous results by describing  key changes to both DCT and CCT. }

\subsubsection{DCT with Quantization and Erasures}
In the initialization phase each node~$\ell$ communicates the vector $\floor{v_\ell}$ to the fusion center using a packet of $C$ bits.   The expected time for   successful transmission of the packet is $1/(1-\epsilon)$. After receiving the vector $\floor{v_{\ell}}$ from all the nodes, the fusion center communicates the  vector $\floor{I}=[\floor{I(1)},\ldots,\floor{I(M)}]$  back to each node $\ell$, which requires an expected time $1/(1-\epsilon)$ as well.

In the test phase, each local preference is communicated using a packet of $\log_{2} M $ bits to the fusion center, also with an expected time $1/(1-\epsilon)$.

The final decision $\hat{H}$ at the fusion center is made in favor of hypothesis $h_i$  when the local decisions received from all the network nodes are in favor of the hypothesis $h_i$. Given the local decision $h_i$ is reached at all the nodes, the expected time for reaching the final decision $\hat{H}$ is $1/(1-\epsilon)^{L}$, as it is required that all the links be 
simultaneously active. Upon taking the final decision, the fusion center sends a halting message to each node $\ell$.  

\subsubsection{CCT with Quantization and Erasures}
In this case, at each time step $n$, we consider the time-varying graph $\mathcal{G(L,E}(n))$, where $\mathcal{E}(n)\subseteq \mathcal{E}$ denotes the set of communication links where a packet can be sent successfully. 

In the initialization phase of CCT, since the graph is time-varying, the weight matrix $W=W(n)$ also varies over time.
This matrix can be expressed as \cite{kar2008sensor}
\begin{equation}
    W(n)=U_{L\times L}-\beta \Bar{L}(n),
\end{equation}
where $\beta$ is a design parameter, $U_{L\times L}$ is the identity matrix of dimension $L\times L$,  $\Bar{L}(n)$ is the $L\times L$ dimensional Laplacian matrix 
of $\mathcal{G(L,E}(n))$ \cite{kar2008sensor}, {with entries:}
\begin{equation}
\Bar{l}_{i,j}(n)\hspace{-3pt}=\hspace{-2pt}\begin{cases}
        \sum_{j^\prime\neq i}\mathbf{1}((i,j^\prime)\in \mathcal{E}(n))
    \quad\hspace{-.5pt}\mbox{ if } i=j,  \\
         -1 \qquad\quad\quad\qquad\quad\quad\hspace{12pt}\mbox{if } (i,j)\in \mathcal{E}(n),  \\
         0 \qquad\quad\quad\quad\quad\qquad\quad\quad\hspace{-4pt}\mbox{          otherwise,}
       \end{cases}
\end{equation}
{where $\mathbf{1}(\cdot)$ denotes the indicator function.}
Each node $i$ can compute locally the {values} $\Bar{l}_{i,j}(n)$, based on whether a packet is received from node $j$ at time $n$. Since $\Bar{l}_{i,j}(n)=\Bar{l}_{j,i}(n)$, it follows that $W(n)$ is a symmetric matrix, where \cite{kar2008sensor}
\begin{equation}
{w}_{i,j}(n)\hspace{-3pt}=\hspace{-2pt}\begin{cases}
        1-\beta\sum_{j^\prime\neq i}\mathbf{1}((i,j^\prime)\in \mathcal{E}(n))
    \quad\hspace{-.5pt}\mbox{ if } i=j,  \\
         \beta \qquad\quad\quad\qquad\quad\quad\hspace{1pt}\mbox{if } (i,j)\in \mathcal{E}(n),  \\
         0 \qquad\quad\quad\quad\quad\qquad\quad\quad\hspace{-1pt}\mbox{          otherwise}.
       \end{cases}
\end{equation}
Then, as in~(\ref{eq:quantizedConsensus}), node $\ell$ updates its quantized estimate at time step $n$ as
\begin{equation}\label{eq:consensus3}
    \hat{I}_{\ell}^{n+1}=\Big \lfloor w_{\ell,\ell}(n)\cdot \hat{I}_{\ell}^{n} +\sum_{j \in  \mathcal{N}_{\ell}}w_{\ell,j}(n)\cdot \hat{I}_{j}^{n} \Big \rfloor.
\end{equation}
{Whenever links are active, the information communicated over the channels is of the same form as that over channels with quantized messages.} The stopping rule of this phase remains the same as stated in Algorithm \ref{alg:Phase1}. In the following phases, we let $\floor{\hat{I}^\epsilon_{\ell}}=[\floor{\hat{I}_{\ell}^\epsilon(1)},\ldots, \floor{\hat{I}_{\ell}^\epsilon (M)}]$ denote  the estimate of vector $I$ using \eqref{eq:consensus3} at node $\ell$ at the end of the initialization phase in this channel model. 

In the test phase of CCT, the SCT is performed locally  using the result of the consensus algorithm to select the threshold, namely $\gamma=\hat{\rho}^\epsilon_{i^*_n,\ell}|\log c|$ and $\hat{\rho}^\epsilon_{i^*_n,\ell}=v_{i^*_n,\ell}/\floor{\hat{I}^\epsilon_{\ell}(i^*_n)}$.

Finally, in the stopping phase of CCT (see Algorithm \ref{alg:Phase3}), the variable $d_{\ell}^{n}$ and the local decision $\hat{H}^{n}_{\ell}$ are communicated over channel by $\log_2(L+2)+\log_2 M$ bits. {Of course, these communications are successful only when the link between the nodes is active.} 

\subsection{Performance Analysis for Channels with {Quantized Messages and  Erasures}}
In this section, we  extend the results of Theorem \ref{DCT_quantized} and Theorem \ref{CCT_quantized} to channels with {quantized messages and 
erasures. }

\begin{theorem}\label{DCT_quantized_linkFailure}
{In the presence of channel with quantized messages and $\epsilon$-random packet erasures, Theorem \ref{DCT_quantized} holds unmodified.}
\end{theorem}

{Intuitively, the reason  why the results of Theorem \ref{DCT_quantized} hold unmodified is  as follows. Link failures only delay the communication} of the quantized information over the channel, which impacts the decision time. Note that the expected time for communication of $\floor{v_{\ell}}$ from all the nodes is at most $L/(1-\epsilon)$, {as is the expected time to communicate the response vector to all the nodes.} Given the same local decision is reached at the nodes, the expected time to reach the final decision is $1/(1-\epsilon)^L$. Likewise, the expected time to communicate the halting message is $L/(1-\epsilon)$. All these delays introduced by the $\epsilon$-erasure channel are finite and independent of $c$, and 
are {embodied in the terms $o(1)$ appearing in the statement of Theorem~\ref{DCT_quantized_linkFailure}.} 

{Next, we give a lemma needed to provide the performance guarantees of CCT.} 

\begin{lemma}\label{thm:propertiesOfW}For all   $n$ and $0<\beta<1/(2|\mathcal{E}|)$, the following holds:\\
$(i)$ $W(n)$ is a doubly stochastic matrix i.e. (\ref{eq:SC11}) and (\ref{eq:SC22})  holds.\\
$(ii)$ For all $i,j\in [L]$, if $w_{i,j}(n)>0$, then $w_{i,j}(n)>\min{(1-\mathcal{D(G)}\beta,\beta)}$, where  $\mathcal{D(G)}=\max_{s}\sum_{j\neq s}\mathbf{1}((j,s)\in \mathcal{E})$ is the maximum node degree   in the graph $\mathcal{G(V,E)}$.\\
$(iii)$ The spectral radius verifies
\begin{equation}\label{eq:spectralRadiusCheck}
R\left(W(n)-\frac{\textbf{{1}}_{L\times1}\cdot \textbf{{1}}_{1\times L}}{L}\right)< 1.
\end{equation}
\end{lemma}

\begin{theorem}\label{CCT_quantized_linkFailure} \emph{(Direct).} Let 
\begin{equation} \label{eq:needthis0}
    h(Q,c,\alpha,\epsilon)=\frac{g(Q,c,\alpha)(2-|\mathcal{E}|\epsilon)}{(1-|\mathcal{E}|\epsilon)^2},
\end{equation}
\begin{equation} \label{eq:needthis2}
  q({Q},c,\alpha,\epsilon )=\frac{Q\cdot g({Q},c,\alpha)}{L (2-|\mathcal{E}|\epsilon)},
\end{equation}
$\epsilon<1/|\mathcal{E}|$, and $0<\beta<1/(2|\mathcal{E}|)$.
{Assume that $C$ is sufficiently large such that for all $i\in [M]$, $h(\tilde Q,c, \min{(1-\mathcal{D(G)}\beta,\beta)},\epsilon) < I(i)$ and $C>\log_2(L+2)+\log_2 M$, and Assumption \ref{assumption:initialValue} is verified. Then the following statements hold for CCT:} \\
$(i)$ For all $c \in (0,\sqrt{({1-|\mathcal{E}|\epsilon})/{2}})$ and all $ i\in [M]$, 
the probability that CCT takes an incorrect decision is
\begin{equation}
\begin{split}
    {\mathbb{P}}^{{\cal C}}_i(\hat{H}\neq h_i)&\leq \min \{(M-1)(1- \exp(-2 q(\Tilde{Q},c,\min{(1-\mathcal{D(G)}\beta,\beta),\epsilon} ))\cdot c^{I(i)/(I(i)+h(\Tilde{Q},c,\min{(1-\mathcal{D(G)}\beta,\beta)},\epsilon))}\\
    &+\exp(-2 q(\Tilde{Q},c,\min{(1-\mathcal{D(G)}\beta,\beta),\epsilon} ),1 \}.
\end{split}
\end{equation}
$(ii)$ For all $\ell \in [L]$, $i,j\in[M]$ and $u\in S$, if $\mathbb{E}\big[\log {p_{i,\ell}^{u}(Y)}/{p_{j,\ell}^{u}(Y)}\big]^2 < \infty$, then 
the expected decision time is
\begin{align} \label{EDCT1_new_error-quantized}
    \mathbb{E}_{i}^{{\cal C}}[N\big|\{\floor{\hat{I}^{\epsilon}_{\ell}}\}_{\ell\in [L]}]
    &\leq (1+o(1)) \Bigg({\tilde Q\cdot h(\tilde Q,c,\min{(1-\mathcal{D(G)}\beta,\beta)},\epsilon})+\frac{\abs{\log c}}{\min_{\ell\in [L]}\floor{\hat{I}^\epsilon_{\ell}(i)}}\Bigg)\\
    \label{EDCT1_new_error-quantized1}
    &\leq (1+o(1)) \Bigg({\tilde Q\cdot h(\tilde Q,c,\min{(1-\mathcal{D(G)}\beta,\beta)},\epsilon})+\frac{\abs{\log c}}{I(i)-h(\tilde Q,c,\min{(1-\mathcal{D(G)}\beta,\beta)},\epsilon)}\Bigg), 
\end{align}
with probability at least $1-\exp{(-2q(\Tilde{Q},c,\min{(1-\mathcal{D(G)}\beta,\beta)},\epsilon))}$, as $c \to 0$. \\
$(iii)$ Combining $(i)$ and $(ii)$, the risk  is
\begin{equation} \label{RDCT1}
\begin{split}
  \mathbb{R}_{i}^{{\cal C}}&\leq\hspace{-2pt} (1+o(1))\Bigg(\frac{\tilde Q\cdot h(\tilde Q,c,\min{(1-\mathcal{D(G)}\beta,\beta)},\epsilon)}{c }+\frac{1}{I(i)-h(\tilde Q,c, \min{(1-\mathcal{D(G)}\beta,\beta)},\epsilon)}\hspace{-3pt}\Bigg)
\cdot c\abs{\log c},  \qquad
\end{split}
\end{equation} \\
with probability $1-\exp{(-2q(\Tilde{Q},c,\min{(1-\mathcal{D(G)}\beta,\beta)},\epsilon))}$, \textnormal{\emph{as} $c\to 0$}.\\
$(iv)$ For all $\ell \in [L]$, $i,j\in[M]$, $u\in S$ and all integers $r\geq 2$, if $\mathbb{E}\big[\abs{\log {p_{i,\ell}^{u}(Y)}/{p_{j,\ell}^{u}(Y)}}^{r+1}\big] < \infty$, then the $r^{th}$ moment of the expected decision time is 
\begin{align}\label{eq:moments}\small
    &\mathbb{E}_{i}^{{\cal C}}[N^{r}\big|\{\floor{\hat{I}^{\epsilon}_{\ell}}\}_{\ell\in [L]}]\nonumber
      \leq\hspace{-3pt} \Bigg(\hspace{-3pt} (1+o(1)) \Bigg({\tilde{Q}\cdot h(\tilde Q,c,\min{(1-\mathcal{D(G)}\beta,\beta)},\epsilon)}+\frac{\abs{\log c}}{I(i)-h(\tilde Q,c, \min{(1-\mathcal{D(G)}\beta,\beta)},\epsilon)}\Bigg)\hspace{-3pt}\Bigg)^{r}\hspace{-3pt}.
\end{align}
with probability at least $1-\exp{(-2q(\Tilde{Q},c,\min{(1-\mathcal{D(G)}\beta,\beta)},\epsilon))}$, as $c\to 0$.\\
\end{theorem}
 We point out that when estimating  the vector $\floor{\hat{I}^{\epsilon}_{\ell}}$ in the initialization phase of CCT, the $\epsilon$-random  erasure  model introduces additional randomness. For this reason, \eqref{EDCT1_new_error-quantized} represents the conditional expected decision time given $\{\floor{\hat{I}^{\epsilon}_{\ell}}\}_{\ell\in [L]}$. To obtain \eqref{EDCT1_new_error-quantized1},  we use the fact that for all 
 $\ell\in [L]$, we have that the random variable
 \begin{equation}
 \floor{\hat{I}^{\epsilon}_{\ell}(i
)} \in  [I(i)-h(\tilde Q,c, \min{(1-\mathcal{D(G)}\beta,\beta)},\epsilon), I(i)+h(\tilde Q,c, \min{(1-\mathcal{D(G)}\beta,\beta)},\epsilon)]
\end{equation}
with probability at least $1-\exp{(-2q(\Tilde{Q},c,\min{(1-\mathcal{D(G)}\beta,\beta)},\epsilon))}$ (see \eqref{eq:estimationErrorCCT_error} in  Appendix \ref{app:CCT_quantized_linkFailure}). 


In Theorem \ref{CCT_quantized_linkFailure},  performance guarantees are provided with high probability and this probability depends on the number of quantization levels and on the packet erasures through $q(\Tilde{Q},c,\min{(1-\mathcal{D(G)}\beta,\beta)},\epsilon)$. As $c\to 0$ and $Q\to \infty$ {(in arbitrary order)}, we have that $q(\Tilde{Q},c,\min{(1-\mathcal{D(G)}\beta,\beta)},\epsilon)\to \infty$ and   $1-\exp{(-2q(\Tilde{Q},c,\min{(1-\mathcal{D(G)}\beta,\beta)},\epsilon))}$ converges to one. 
Additionally, the performance of CCT also depends on $h(\tilde Q,c, \min{(1-\mathcal{D(G)}\beta,\beta)},\epsilon)$. As $Q\to \infty$,  we have $g(\tilde Q,c,\alpha)\to 0,$ which implies $h(\tilde Q,c, \min{(1-\mathcal{D(G)}\beta,\beta)},\epsilon)\to 0$. Finally,  the time required to complete the initialization phase is given  by $\Tilde{Q}\cdot h(\tilde Q,c, \min{(1-\mathcal{D(G)}\beta,\beta)},\epsilon)=\mathcal{O}(|\log(c)|)$ as $Q\to\infty$. 
 
As $Q\to\infty$,  Theorem~\ref{CCT_quantized_linkFailure}  {recovers the  same optimality of CCT expressed} in Theorem \ref{main2}. 
The quantity $h(\tilde Q,c, \min{(1-\mathcal{D(G)}\beta,\beta)},\epsilon)$ {quantifies} the loss in performance of CCT due to both quantization and random packet erasures.  In this case, since $\floor{\hat{I}^{\epsilon}_{\ell}}$ is a random variable,  the  error for    $|I(i) - \floor{\hat{I}^{\epsilon}_{\ell}(i)}|$ is at most  $h(\tilde Q,c, \min{(1-\mathcal{D(G)}\beta,\beta)},\epsilon)$ with high probability, and our theorem   assumes that this error is smaller than $I(i)$. Since $\tilde{Q}= 2^{\tilde C/M}$, this constraint can be satisfied by having $C$ sufficiently large. The additional capacity constraint $C>\log_2(L+2)+\log_2 M$ in the statement of the theorem is due to the transmission of $d_\ell^n$ and the local preference $\hat{H}^{n}_{\ell}$. 

\section{{Summary}}\label{sec:Conclusion}

Networked sensor systems are becoming increasingly popular for inference problems due to  their improved robusteness, scalability, versatility, and performance. Initial implementations were based on inexpensive small sensors, with extremely limited hardware/software capabilities. Progressively, these  devices acquired more and more functionalities, and are nowadays capable of active sensing, namely they can adapt the probing signal on the basis of previous measurements, in order to optimize their sensing capability. Thus, individual sensors have become intelligent devices which continuously learn from the past and can decide their future actions, in a closed-loop adaptive scheme.

We considered two network configurations of these ``intelligent'' sensors: a star network configuration with a fusion center, and a general network configuration that is fully distributed. In the first configuration, the  fusion center  coordinates the actions of the remote nodes, and takes the final decision. The second configuration does not have a central coordination, and  all the processing takes place at the nodes: they actively collect measurements, exchange information with  immediate neighbors, and collectively take a decision.

For the first configuration   we proposed a sequential  adaptive decision system --- referred to as DCT --- which  operates in three phases. 
 First, there is a round of communication between the fusion center and the remote nodes, needed to define the relative capability of each node to detect the hypotheses. This capability is then used to  divide the decision task among the nodes. Each node begins to continuously sense the environment, and makes the central entity aware about decisions that are locally believed to be sufficiently reliable. The final decision is taken by the fusion center on the basis of these \emph{local suggestions} about the true hypothesis.

We provided a theoretical analysis of  detection performance and expected time to reach a decision. 
We show that the test is asymptotically optimal in terms of detection performance (risk), as the observation cost per unit time tends to zero. 

For the second configuration, we exploit ideas from the DCT implementation, combined with gossip protocols that use consensus techniques, to design a fully distributed adaptive sequential decision system, which is referred to as CCT. Our approach is markedly different from those usually exploited in the literature, where real-valued belief vectors are continuously exchanged over the network to reach   consensus. 

Our CCT works in three phases. In the first phase, a consensus about the relative capability of the nodes to detect the state of nature is achieved by means of gossip protocols with local information exchange. In the second phase, nodes implement the Chernoff test and, once  all the network nodes reach their preference, the final decision is reached in a  distributed way in the  third phase of operation, by diffusing messages across the network that percolate the information of whether the other sensors have terminated their share of task.
We prove  the asymptotic optimality of CCT, up to a multiplicative factor in terms of both risk and higher moments of the decision time.

\appendices
\section{Proof of Theorem\ref{converse}} \label{app:th2}
\begin{proof}
Let $H^*=h_{i}$ be the true hypothesis. The proof of Theorem \ref{converse} consists of two parts.  First, for all $0<\epsilon<1$, we show that for the probability of error 
to be close to zero, {the log-likelihood ratio between $h_i$ and all $h_m\neq h_i$, 
should be greater than} $-(1-\epsilon)\log c $ with high probability  as $c\to 0$. Namely, the inequality
\begin{equation}
S^{N}(h_{i},h_{m}) = \sum_{\ell=1}^{L}\sum_{k=1}^{N}\log\frac{p_{i,\ell}^{u_{k,\ell}}{(y_{k,\ell})}}{p_{m,\ell}^{u_{k,\ell}}{(y_{k,\ell})}}\geq -(1-\epsilon)\log c  
\label{eq:2}
\end{equation}
must hold with high probability, as $c\to 0$. Second, we show that for all  $0<\epsilon<1$ and $n< -(1-\epsilon)\log  c /I(i)$, it is unlikely that such inequality is not satisfied for some hypothesis $h_m\neq h_i$. 

We start by defining two sets of hypotheses  $\mathcal{H}^{\prime}_{0}=\{h_{i}\}$ and $\mathcal{H}^{\prime}_{1}=\{ h_{m}\}_{m \neq i}$. By (\ref{eq:assumption}), both type $\RN{1}$ and type $\RN{2}$ error probabilities  of the hypothesis test $\mathcal{H}^{\prime}_{0}$ vs. $\mathcal{H}^{\prime}_{1}$ are $O(-c\log c )$. Thus, by \cite[Lemma 4]{chernoff1959sequential}, for all hypotheses $h_m\neq h_i$ and $0<\epsilon<1$, we have
\begin{equation}
\label{eq:1}
\mathbb{P}_{i}\Big(S^{N}(h_{i},h_{m})\leq -(1-\epsilon)\log c \Big) = {\cal O}(-c^{\epsilon}\log c ).
\end{equation}
Therefore, as $c\to 0$, the probability in (\ref{eq:1}) tends to 0, which concludes the first part of the proof.\\
Now, we show that for all $\epsilon > 0$, we have
\begin{equation}
\label{eq:3}
\lim_{n^{\prime} \rightarrow \infty} \mathbb{P}_i\left(\max_{1\leq n \leq n^{\prime}}\min_{m\neq i} S^{n}(h_{i},h_{m})\geq n^{\prime}(I(i)+\epsilon)\right)= 0.
\end{equation}
 We have 
\begin{align}
S^{n}(h_{i},h_{m}) &= \sum_{\ell=1}^{L}\sum_{k=1}^{n} \Bigg( \log\frac{p_{i,\ell}^{u_{k,\ell}}{(y_{k,\ell})}}{p_{m,\ell}^{u_{k,\ell}}{(y_{k,\ell})}}\nonumber\\
&-D(p_{i,\ell}^{u_{k,\ell}}||p_{m,\ell}^{u_{k,\ell}}) \Bigg)\nonumber\\
&+\sum_{\ell=1}^{L}\sum_{k=1}^{n}D(p_{i,\ell}^{u_{k,\ell}}||p_{m,\ell}^{u_{k,\ell}})\nonumber\\
&=M_{1}^{n}+M_{2}^{n},
\end{align}
where
\begin{equation}
\begin{split}
M_{1}^{n}&= \sum_{\ell=1}^{L}\sum_{k=1}^{n}\Bigg(\log\frac{p_{i,\ell}^{u_{k,\ell}}{(y_{k,\ell})}}{p_{m,\ell}^{u_{k,\ell}}{(y_{k,\ell})}}-D(p_{i,\ell}^{u_{k,\ell}}||p_{m,\ell}^{u_{k,\ell}}) \Bigg),\\
\end{split}
\end{equation}
is a martingale with mean 0, and 
\[M_{2}^{n}=\sum_{\ell=1}^{L}\sum_{k=1}^{n}D(p_{i,\ell}^{u_{k,\ell}}||p_{m,\ell}^{u_{k,\ell}}).\]
Then, for all $1\leq n \leq n^{\prime}$, we have
\begin{align}
\min_{m\neq i}M_{2}^{n}&\stackrel{}{=}\min_{m\neq i}\sum_{\ell=1}^{L}\sum_{k=1}^{n}D(p_{i,\ell}^{u_{k,\ell}}||p_{m,\ell}^{u_{k,\ell}})\nonumber\\
&\stackrel{(a)}{\leq}\sum_{\ell=1}^{L}\sum_{k=1}^{n} v_{i,\ell}\nonumber\\
&\stackrel{(b)}{=} nI(i)\nonumber\\
&\stackrel{(c)}{\leq} n^{\prime} I(i),
\label{eq:cumulativeCapability_1}
\end{align}
where $(a)$ follows from the definition of $v_{i,\ell}$ in 
(\ref{eq:value}), $(b)$ follows from the definition of $I(i)$ in (\ref{eq:CC}), and $(c)$ follows from $n \leq n^{\prime}$. Now, using \eqref{eq:cumulativeCapability_1},
if the event in (\ref{eq:3}) occurs for a fixed $n_{1}$,  i.e. 
\[\min_{m\neq i}(M_{1}^{n_{1}}+M_{2}^{n_{1}})\geq n^{\prime}\big(I(i)+\epsilon\big),\]
then there exists a hypothesis $h_{m}$ such that $M_{1}^{n_{1}}\geq n^{\prime}\epsilon$. Thus, there exists a constant $K^{\prime}>0$ such that the probability in (\ref{eq:3}) becomes
\begin{align}
&\mathbb{P}_i\bigg(\max_{1\leq n \leq n^{\prime}}\min_{m\neq i}S^{n}(h_{i},h_{m})\geq n^{\prime}(I(i)+\epsilon)\bigg)\nonumber\\
&\leq \sum_{m\neq i}\mathbb{P}_i\bigg(\max_{1\leq n\leq n^{\prime}}M_{1}^{n}\geq n^{\prime}\epsilon\bigg)\nonumber\\
&\stackrel{(a)}{\leq}\frac{(M-1)K^{\prime}}{n^{\prime}\epsilon^{2}},
\end{align}
where $(a)$ follows from the fact $M_{1}^{n}$ is a martingale with mean zero and using  the Doob-Kolmogorov extension of Chebyshev's inequality\cite{doob1953stochastic}. 
Thus, (\ref{eq:3}) follows. As discussed in \cite[Theorem 2]{chernoff1959sequential}, for  $n_0= -(1-\epsilon)\log  c /(I(i)+\epsilon)$,
we have 
\begin{align}
&\mathbb{P}_i(N\leq n_{0})\nonumber\\
&\leq \mathbb{P}_i\Big( N \leq n_{0} \mbox{ and } \forall  m\neq i: \nonumber\\
&\qquad S^{N}(h_{i},h_{m})\geq n_{0}(I(i)+\epsilon) \Big)\nonumber\\
&+\mathbb{P}_i\Big(\exists m \neq i:
  S^{N}(h_{i},h_{m})
\leq n_{0}(I(i)+\epsilon)\Big)\nonumber\\
&\leq \mathbb{P}_i\Big( \max_{1\leq n \leq n_{0}}\min_{m\neq i} S^{n}(h_{i},h_{m})\geq n_{0}(I(i)+\epsilon)\Big)\nonumber\\
&+\mathbb{P}_i\Big(\exists m \neq i:
  S^{N}(h_{i},h_{m})
\leq n_{0}(I(i)+\epsilon)\Big).
\label{eq:5}
\end{align}
The first and the second terms at the right-hand side of (\ref{eq:5}) approach   zero  by  (\ref{eq:3}) and (\ref{eq:1})  respectively. Now, using (\ref{eq:5}), we also have 
\begin{equation}
\begin{split}
\mathbb{P}_i(N^r\leq n_{0}^r)
&=\mathbb{P}_i(N\leq n_{0})\to 0,\\
\end{split}
\end{equation}
as $c\rightarrow 0$.
(\ref{Econverse}) follows upon observing that as $c\to 0$, $\mathbb{E}_i[N^r]\geq n_0^r$ which is 
\[\mathbb{E}_i[N^r]\geq \Bigg((1+o(1))\frac{\abs{\log c}}{I(i)}\Bigg)^{r}.\]
The proof of (\ref{Rconverse}) is straightforward.
\end{proof}
\section{Proofs for DCT and CCT}
\subsection{Proof of Theorem \ref{DCT}}\label{app:DCT}
\begin{proof}
To prove Theorem \ref{DCT}, we need some   additional notation. Let $A_{n,j}$  be  the set of sample paths where the decision made by the fusion center is in favor of $h_{j}$ at the $n^{th}$ step, and we indicate a single sample path  as $\{(u^{n}_{1},y^{n}_{1})\ldots(u^{n}_{L},y^{n}_{L})\}$. We indicate by $A_{n,j,\ell}$   the set of sample paths  in $A_{n,j}$ corresponding to the $\ell^{th}$ node.  Finally, we define
\begin{equation}\label{SRsensor}
     N_{i,\ell} =
     \inf\left\{n: \log\frac{\mathbb{P}(H^{*}=h_{i^*_{n},}|y^{n+1}_{\ell},u^{n+1}_{\ell})}{\max_{j \neq i^*_{n}}\mathbb{P}(H^{*}=h_{j}|y^{n+1}_{\ell},u^{n+1}_{\ell})} \geq \rho_{i,\ell}\,\abs{\log c}\right\}
     =\inf\left\{n: \sum_{k=1}^{n} \log \frac{p_{i,\ell}^{u_{k,\ell}}(y_{k,\ell})}{\max_{j\neq i}p_{j,\ell}^{u_{k,\ell}}(y_{k,\ell})}\geq \rho_{i,\ell}\,\abs{\log c}\right\}.  \nonumber
 \end{equation}
 
 The proof consists of two parts. First, we write  ${\mathbb{P}}^{{\cal D}}_i(\hat{H}\neq h_i)$   as the probability of a countable union of disjoint sets of sample paths. An upper bound on  this probability then follows from an upper bound on the probability of these disjoint sets, in conjunction with  the union bound.  Second, we upper bound $\mathbb{E}_{i}^{{\cal D}}[N]$   by the sum of the expected time required to reach the threshold in  ~(\ref{StoppingRuleDec11}) at node $\ell$ for $H ^*=h_{i}$, and the expected delay between the time of reaching the threshold and the time when the final decision is taken  in favor of hypothesis $h_{i}$ at the fusion center. We then show that these expectations are the same at all the nodes, so that~(\ref{EDCT1}) follows.
 
 Consider the probability ${\mathbb{P}}^{\mathcal{D}}_i(\hat{H}= h_j)$.
This is the same as the probability of  the countable union of disjoint sets  $A_{n,j}$.
Thus, for all $j\neq i$, we can write 
 \begin{align}\label{probAnj1}
&{\mathbb{P}}_{i}^{{\cal D}}(A_{n,j}) \nonumber \\
&= \int_{A_{n,j}}\prod_{\ell=1}^{L}\prod_{k=1}^{n}p_{i,\ell}^{u_{k,\ell}}(y_{k,\ell})\,d{y_{1,\ell}(u_{1,\ell})}\ldots\,d{y_{n,\ell}(u_{n,\ell})} \nonumber \\
& \stackrel{(a)}{=}\prod_{\ell=1}^{L}\int_{A_{n,j,\ell}}\prod_{k=1}^{n}p_{i,\ell}^{u_{k,\ell}}(y_{k,\ell})\,d{y_{1,\ell}(u_{1,\ell})}.....\,d{y_{n,\ell}(u_{n,\ell})} \nonumber \\
& \stackrel{(b)}{\leq} \prod_{\ell=1}^{L}\int_{A_{n,j,\ell}}c^{\rho_{j,\ell}}\prod_{k=1}^{n}p_{j,\ell}^{u_{k,\ell}}(y_{k,\ell})\,d{y_{1,\ell}(u_{1,\ell})}\ldots\,d{y_{n,\ell}(u_{n,\ell})} \nonumber \\
& \stackrel{(c)}{=} c\prod_{\ell=1}^{L}\int_{A_{n,j,\ell}}\prod_{k=1}^{n}p_{j,\ell}^{u_{k,\ell}}(y_{k,\ell})\,d{y_{1,\ell}(u_{1,\ell})}\ldots\,d{y_{n,\ell}(u_{n,\ell})}\nonumber \\
& = c\prod_{\ell=1}^{L}{\mathbb{P}}^{\mathcal{D}}_{j}(\hat{H}=h_j \mbox{ at sample } n \mbox{ at } \ell^{th} \mbox{ sensor}) \nonumber \\
&= c \, \mathbb{P}^{\mathcal{D}}_{j}(\hat{H}=h_j \mbox{ at sample } n),
\end{align}
where $(a)$ follows from the definition of $A_{n,j,\ell}$; $(b)$ follows from the definition of $N_{i,\ell} $; $(c)$ follows from $\sum_{\ell=1}^{L} \rho_{j,\ell}=1$. Now, we can bound $\mathbb{P}^{{\cal D}}_i(\hat{H}\neq h_i)$ as follows
\begin{align}\label{eq:proberror2}
&{\mathbb{P}}^{{\cal D}}_i(\hat{H}\neq h_i) = \sum_{j\neq i}{\mathbb{P}}^{{\cal D}}_i(\hat{H} = h_j)= \sum_{j\neq i}\sum_{n= 1}^{\infty}{\mathbb{P}}^{\mathcal{D}}_{i}(A_{n,j}) \nonumber \\
&\leq \sum_{j\neq i}\sum_{n=1}^{\infty} c \,  {\mathbb{P}}_{j}^{{\cal D}}(\hat{H}=h_j \mbox{ at sample } n) \nonumber \\ 
&= \sum_{j\neq i} c \,  {\mathbb{P}}_{j}^{{\cal D}}(\hat{H}=h_j) \leq c \,(M-1),
\end{align}
where the first inequality follows from~(\ref{probAnj1}).
This proves part~$(i)$ of the theorem. 

Let us now define 
\[
\tau(N_{i,\ell}) = \sup\Big\{n\ : \hspace*{-5pt }\sum_{k=N_{i,\ell}+1}^{N_{i,\ell}+n} \hspace*{-5pt } \log \frac{p_{i,\ell}^{u_{k,\ell}}(y_{k,\ell})}{\max_{j\neq i}p_{j,\ell}^{u_{k,\ell}}(y_{k,\ell})} < 0\Big\}.
\]
The  condition in \eqref{StoppingRule} is satisfied for threshold in  (\ref{StoppingRuleDec11}) at the $\ell^{th}$ node for all $n>N_{i,\ell}+\tau(N_{i,\ell})$, yielding
\[
{\small N \hspace*{-3pt} \leq \hspace*{-3pt} \max_{1\leq \ell\leq L} \hspace*{-3pt} (N_{i,\ell}+\tau(N_{i,\ell})+3(M+1))
\hspace*{-2pt}\leq\hspace*{-2pt} \max_{1\leq \ell\leq L} \hspace*{-2pt}N_{i,\ell} + \sum_{\ell=1}^{L}\tau(N_{i,\ell}) \hspace*{-2pt}+\hspace*{-2pt}3(M+1), }
\]
where if $C\geq M$, then three time steps are needed to communicate $v_{\ell}$, $I$ and the halting message; otherwise at most $3(M+1)$ time steps are needed to communicate this information. 

Taking the expectation of both sides, we get
\begin{equation}\label{SampleSize}
\mathbb{E}^{{\cal D}}_i[N]\leq \mathbb{E}_i \left [\max_{1\leq \ell\leq L}N_{i,\ell} \right ] +\sum_{\ell=1}^{L}\mathbb{E}_i[\tau(N_{i,\ell})]+3(M+1). 
\end{equation}

We now bound the terms on the right-hand side of~(\ref{SampleSize}). 
Since each node performs the Chernoff test individually, for all $\ell\in [L]$ and $i\in[M]$, there exist two constants $K_{i,\ell}>0$ and $b_{i,\ell}>0$ such that for all $\epsilon>0$ and $n\geq (1+\epsilon)\abs{\log(c)}/I(i)$, we have \cite[Lemma 2]{chernoff1959sequential}
\begin{equation}\label{eq:conc}
    \mathbb{P}_{i}(N_{i,\ell}\geq n)\leq K_{i,\ell}\cdot e^{-b_{i,\ell}n}. 
\end{equation}
Thus, we have
\begin{equation}
\label{exTime}
\mathbb{E}_i[N_{i,\ell}]={(1+o(1))\abs{\log c}}/{I(i)},
\end{equation}
which is independent of $\ell$.
Using (\ref{eq:conc}), we also have that  for all $\epsilon>0$ and $n\geq (1+\epsilon)\abs{\log(c)}/I(i)$ ,
\begin{equation}
\begin{split}
    \mathbb{P}_{i}\bigg(\max_{1\leq \ell\leq L}N_{i,\ell}\geq n\bigg) &\leq \sum_{\ell=1}^{L}\mathbb{P}_{i}(N_{i,\ell}\geq n),\\
    &\leq L K_{i}e^{-b_{i}n},
\end{split}
\end{equation}
where $K_{i}=\max_{\ell}K_{i,\ell}$ and $b_{i}=\min_{\ell} b_{i,\ell}$. For all $r\geq 1$, we have the bound on the $r^{th}$ moment of $\max_{1\leq \ell\leq L}N_{i,\ell}$ ,i.e. 
\begin{equation}\label{SS2DCT}
    \mathbb{E}_i \left [\Big(\max_{1\leq \ell\leq L}N_{i,\ell}\Big)^r \right ]\leq \Bigg((1+o(1))\frac{\abs{\log c} }{I(i)}\Bigg)^{r}.
\end{equation}

Now, we bound the higher moments of $\tau(N_{i,
\ell})$. Let $N^{*}$ be the time instance such that for all $n\geq N^{*}$, the local decision $\hat{H}_{\ell}$ at node $\ell$ is $h^*$, i.e., $\hat{H}_{\ell}={h}^{*}$.  
Using \cite[Lemma 1]{chernoff1959sequential}, there exists $K>0$ and $b>0$ such that
\begin{equation}\label{eq:chernoffTime}
    \mathbb{P}_{i}(N^{*}\geq n)\leq k\exp{(-bn)},
\end{equation}
which implies $\mathbb{P}_{i}(N^{*}<\infty)=1$. Then, node $\ell$ following time $N^{*}$ selects the actions in an i.i.d.\ fashion according to the probability mass function given by (\ref{eq:pmf}). 

Let $G_{n,\ell}$ be the joint cumulative distribution function of the variables $(y_{n,\ell},u_{n,\ell})$ at round $n$ and node $\ell$ for the Chernoff test. Also, let $F_{\ell}$ be the joint cumulative distribution function of $(y_{n,\ell},u_{n,\ell})$ under the true hypothesis $h^*$ when the actions are selected according to $Q_{h^*}^{\ell}$ (see (\ref{eq:pmf})) at each round at sensor~$\ell$. Then, for all $n>N^*$, $G_{n,\ell}=F_{\ell}$. Since $\mathbb{P}_{i}(N^{*}<\infty)=1$, it follows that the distribution $G_{n,\ell}$ converges to $F_{\ell}$. 

Given that for all $n$, $(y_{n},u_{n})\sim F_{\ell}$ are i.i.d.\ random variables, we have that
\begin{equation}\label{eq:111}
    \mathbb{E}_i\Big[\log ( {p_{i,\ell}^{u_{i,\ell}}(y_{k,\ell})}/{\max_{j\neq i}p_{j,\ell}^{u_{i,\ell}}(y_{k,\ell}))}\Big]=v_{i,\ell}>0. 
\end{equation}
Additionally, using (\ref{eq:111}), finiteness of the $r+1^{st}$ moment of log-likelihood ratio for $r\geq 1$, and by Corollary \ref{corr:1}, Lemma \ref{lemma:BoundingTime} and \eqref{eq:AppendixC3} in Appendix \ref{app:C}, we have that 
\begin{equation}\label{eq:MomentFinite}
    \begin{split}
    \mathbb{E}_{i}[\tau(N_{i,\ell}) )^{r}]
    &<\infty,
    \end{split}
\end{equation}
where the expectation is with respect to $F_{\ell}$. 

We now note   that (\ref{eq:MomentFinite}) also holds when the expectation is with respect to $G_{n,\ell}$. To show this claim, first observe that $\mathbb{E}_{i}[\tau(N_{i,\ell}) )^{r}]$ is upper bounded by the two terms at the right hand side of (\ref{eq:timeBound}) in Corollary \ref{corr:1}. The first term is bounded, since the KL-divergence between any two probability measures is finite. 
The second term  can be split into two summations, one for $1\leq n\leq N^*$, and the other for $n \geq N^*+1$. The first summation is finite  since $N^*<\infty$ a.s., and the probability is at most one. By using Lemma \ref{lemma:BoundingTime} in Appendix \ref{app:C} and $G_{n,\ell}=F_{\ell}$,  the second summation is also finite. It follows that (\ref{eq:MomentFinite}) holds for the SCT.

Since $\mathbb{E}[\log ( {p_{i,\ell}^{u_{i,\ell}}(y_{k,\ell})}/{\max_{j\neq i}p_{j,\ell}^{u_{i,\ell}}(y_{k,\ell}))}]^2$ is finite, using (\ref{eq:MomentFinite}), the   term $\mathbb{E}_i[\tau(N_{i,\ell})]$ on the right-hand side of (\ref{SampleSize})  is finite and independent of $c$.
Now, combining equation (\ref{SampleSize}), (\ref{SS2DCT}) and the finiteness of $\mathbb{E}_i[\tau (N_{i,\ell})]$, as $c\to 0$ we get~(\ref{EDCT1}). Thus, part $(ii)$ of the theorem is proved. 

Now, 
\begin{equation}\label{eq:MomentDCT}
    \mathbb{E}^{\mathcal{D}}_{i}[N^{r}]\leq \mathbb{E}_{i}\Big[\Big(\max_{1\leq \ell\leq L} N_{i,\ell}+ \sum_{\ell\in[L]}\tau(N_{i,\ell}) +1\Big)^r\Big].
\end{equation}
The moments of $\sum_{\ell\in[L]}\tau(N_{i,\ell})$ are finite and independent of~$c$.
Hence, the dominant term, dependent on $c$, in the expansion of the right-hand side of (\ref{eq:MomentDCT}) is given only by $\max_{1\leq \ell\leq L} N_{i,\ell}$. It follows that using (\ref{eq:MomentFinite}) and (\ref{SS2DCT}),  as $c\to 0$, we have
\begin{equation}\label{eq:MomentFinalDCT}
\begin{split}
      \mathbb{E}_{i}^{\mathcal{D}}[N^{r}]
      \leq \hspace*{-5pt}\Bigg( (1+o(1)) \frac{\abs{\log c}}{I(i)}\Bigg)^{r},
\end{split}
\end{equation}
which proves part $(iv)$ of the theorem.
\end{proof}

\subsection{Proof of Theorem \ref{mainDCT}}
\label{app:DCTMain}
\begin{proof}
Combining Theorems \ref{converse} and \ref{DCT}, (\ref{RfinalDCT}) and (\ref{EfinalDCT})  follow  immediately. We then turn to the proof of (\ref{eq:average}).

For all $\ell \in [L]$, given that hypothesis $h_i$ is true,  we have that as $c \rightarrow 0$ the probability of incorrect detection tends to zero. It follows that  $\hat{H}=h_{i}$ and
\begin{equation}
    \begin{split}
        \mathbb{E}^{{\cal D}}_{i}[N]&= (1+o(1))\frac{\abs{\log c}}{I(i)},\\
        &=\mathbb{E}^{{\cal D}}_{i}[N_{i,\ell}],
    \end{split}
\end{equation}
 where the last equality follows from (\ref{exTime}). 
 Thus, as $c\to 0$  all the nodes reach the same local decision,   on average, at the same time,  and the average number of messages that each node sends to the fusion center to communicate this local decision is one. It follows that,  as $c \rightarrow 0$ , the total expected communication overhead is four: two in the initialization phase, one to communicate the local decision, and one to receive the halting message.
 \end{proof}
 \subsection{Proof of Theorem \ref{CCT}} 
 \begin{proof}
 Let $B_{n,j}$  be  the set of sample paths where the final decision $\hat{H}$ is initiated in favor of $h_{j}$ at the $n^{th}$ step, and we indicate a single sample path  as $\{(u^{n}_{1},y^{n}_{1})\ldots(u^{n}_{L},y^{n}_{L})\}$. We indicate by $B_{n,j,\ell}$   the set of sample paths  in $B_{n,j}$ corresponding to the $\ell^{th}$ node. $N^{c}$ denotes the time taken to terminate the initialization phase of CCT. Now, we define the two times associated with the test phase of CCT:
\begin{equation}\label{SRsensor_1}
     T_{i,\ell}=
     \inf\left\{n: \log\frac{\mathbb{P}(H^{*}=h_{i^*_{n},}|y^{n+1}_{\ell},u^{n+1}_{\ell})}{\max_{j \neq i^*_{n}}\mathbb{P}(H^{*}=h_{j}|y^{n+1}_{\ell},u^{n+1}_{\ell})} \geq \hat{\rho}^{(N^{c})}_{i,\ell}\,\abs{\log c}\right\} =\inf\left\{n: \sum_{k=1}^{n} \log \frac{p_{i,\ell}^{u_{k,\ell}}(y_{k,\ell})}{\max_{j\neq i}p_{j,\ell}^{u_{k,\ell}}(y_{k,\ell})}\geq \hat{\rho}^{(N^{c})}_{i,\ell}\,\abs{\log c}\right\},  \nonumber
 \end{equation}
 and 
 \[
\tau(T_{i,\ell}) = \sup\Big\{n\ : \hspace*{-5pt }\sum_{k=T_{i,\ell}+1}^{T_{i,\ell}+n} \hspace*{-5pt } \log \frac{p_{i,\ell}^{u_{k,\ell}}(y_{k,\ell})}{\max_{j\neq i}p_{j,\ell}^{u_{k,\ell}}(y_{k,\ell})} < 0\Big\}.
\]

The proof consists of two parts. First, we write 
${\mathbb{P}}^{{\mathcal{ C}}}_i(\hat{H}\neq h_i)$   as the probability of a countable union of disjoint sets of sample paths. An upper bound on  this probability then follows from an upper bound on the probability of these disjoint sets, in conjunction with  the union bound.  Second, $\mathbb{E}_{i}^{{\mathcal{ C}}}[N]$  is dependent on the time required to reach and detect the consensus during the initialization phase, the time required to reach the threshold in \eqref{StoppingRuleDec111} in the test phase, and the time required to reach and detect that the nodes have reached a common preference about a hypothesis in the stopping phase. The stopping time $N$ can be bounded as
\begin{equation}
\begin{split}
    N &\leq N^{c}+\max_{1\leq \ell\leq L}( T_{i,\ell}+\tau(T_{i,\ell}) )+N^{s},
\end{split}
\end{equation}
where $N^{s}$ is the time taken to detect the common preference about the hypothesis in the stopping phase of CCT.

Consider the probability ${\mathbb{P}_i^{\mathcal{C}}}(\hat{H}= h_j)$.
This is the same as the probability of  the countable union of disjoint sets  $B_{n,j}$.
Thus, for $j\neq i$, we can write 
 \begin{align}\label{probBnj}
&{\mathbb{P}}_{i}^{{\mathcal{C}}}(B_{n,j}) \nonumber \\
&= \int_{B_{n,j}}\prod_{\ell=1}^{L}\prod_{k=1}^{n}p_{i,\ell}^{u_{k,\ell}}(y_{k,\ell})\,d{y_{1,\ell}(u_{1,\ell})}\ldots\,d{y_{n,\ell}(u_{n,\ell})} \nonumber \\
& \stackrel{(a)}{=}\prod_{\ell=1}^{L}\int_{B_{n,j,\ell}}\prod_{k=1}^{n}p_{i,\ell}^{u_{k,\ell}}(y_{k,\ell})\,d{y_{1,\ell}(u_{1,\ell})}.....\,d{y_{n,\ell}(u_{n,\ell})} \nonumber \\
& \stackrel{(b)}{\leq} \prod_{\ell=1}^{L}\int_{B_{n,j,\ell}}c^{\hat{\rho}^{(n)}_{j,\ell}}\prod_{k=1}^{n}p_{j,\ell}^{u_{k,\ell}}(y_{k,\ell})\,d{y_{1,\ell}(u_{1,\ell})}\ldots\,d{y_{n,\ell}(u_{n,\ell})} \nonumber \\
& \stackrel{(c)}{\leq} c^{I(i)/(I(i)+c)}\cdot\nonumber \\
&\qquad\prod_{\ell=1}^{L}\int_{B_{n,j,\ell}}\prod_{k=1}^{n}p_{j,\ell}^{u_{k,\ell}}(y_{k,\ell})\,d{y_{1,\ell}(u_{1,\ell})}\ldots\,d{y_{n,\ell}(u_{n,\ell})}\nonumber \\
& = c^{I(i)/(I(i)+c)}\prod_{\ell=1}^{L}{\mathbb{P}_{j}^{\mathcal{ C}}}(\hat{H}=h_j \mbox{ at sample } n \mbox{ at } \ell^{th} \mbox{ sensor}) \nonumber \\
&= c^{I(i)/(I(i)+c)} \mathbb{P}^{\mathcal{C}}_{j}(\hat{H}=h_j \mbox{ at sample } n),
\end{align}
where $(a)$ follows from the definition of $B_{n,j,\ell}$; $(b)$ follows from the definition of $T_{i,\ell} $; $(c)$ follows from the fact that using Theorem \ref{th:1} and \eqref{eq:error2}, $I(i)/(I(i)+c)\leq\sum_{\ell=1}^{L} \hat{\rho}^{(n)}_{j,\ell}\leq I(i)/(I(i)-c)$ and $c<1$. Now, we can bound $\mathbb{P}^{{\cal C}}_i(\hat{H}\neq h_i)$ as follows
\begin{align}\label{eq:CCTPRobError2}
&{\mathbb{P}}^{{\mathcal{C}}}_i(\hat{H}\neq h_i) = \sum_{j\neq i}{\mathbb{P}}^{{\mathcal{C}}}_i(\hat{H} = h_j)= \sum_{j\neq i}\sum_{n= 1}^{\infty}{\mathbb{P}_{i}^{\mathcal{C}}}(B_{n,j}) \nonumber \\
&\leq \sum_{j\neq i}\sum_{n=1}^{\infty} c^{I(i)/(I(i)+c)} \,  {\mathbb{P}}_{j}^{{\mathcal{C}}}(\hat{H}=h_j \mbox{ at sample } n) \nonumber \\ 
&=  c^{I(i)/(I(i)+c)}  \,(M-1),
\end{align}
where the inequality in the chain follows by~(\ref{probBnj}).
This proves part $(i)$ of the theorem. 

Let us bound the time required to terminate the initialization phase, i.e., $N^{c}$. Since matrix $W$ in (\ref{eq:Consensus2}) is row stochastic (see (\ref{eq:SC22})) and the graph $\mathcal{G(N,E)}$ is connected,  the ergodic coefficient $\eta(W)\in (0,1)$ (see Lemma \ref{lemma:ErgodicCoeff}). It follows from \cite{dong2017flocking} that for all $k,n \in \mathbb{N}$ and $\ell,j\in [L]$, we have
\begin{equation}
\label{eq:ConvRate}
e^{k+n}_{\ell,j}\preceq \big(1-\eta(W^{n})\big)e^{k}_{\ell,j},
\end{equation}
where $e^{k}_{\ell,j} = \abs{\hat{I}^{k}_{\ell}-\hat{I}^{k}_{j}}$. Now, if the initialization phase reaches uniformly local $c$-consensus at time instance $k_{0}$, then for all $\ell,j\in[L]$, we have  (see (\ref{eq:error2}))
\[e^{k_{0}}_{\ell,j} \preceq c\cdot\textbf{1}_{1\times M}.\]
 Thus, there exists $k^\prime\in \mathbb{N}$ such that $h^{\mathcal{G}}\cdot k^\prime\leq k_{0}\leq h^{\mathcal{G}}\cdot( k^\prime+1)$. Using (\ref{eq:ConvRate}), for all $\ell,j\in[L]$, we have
\begin{align}
\label{eq:Ph1time1}
    e^{k_{0}}_{\ell,j}&\preceq e^{h^{\mathcal{G}}\cdot k^\prime}_{\ell,j}
    \stackrel{(a)}{\preceq} \Big(1-\eta\big(W^{h^{\mathcal{G}}}\big)\Big)^{k^\prime} e^{0}_{\ell,j}\nonumber\\
    &\stackrel{(b)}{\preceq} \Big(1-\eta\big(W^{h^{\mathcal{G}}}\big)\Big)^{k^\prime} I,
\end{align}
where (a) follows from $\hat{I}^{h^{\mathcal{G}}\cdot k^\prime}=W^{h^{\mathcal{G}}}\cdot \hat{I}^{h^{\mathcal{G}}\cdot (k^\prime-1)}$ and  Lemma~\ref{lemma:ErgodicCoeff}, and (b) follows from the fact that for all $\ell,j \in [L]$, $e^{0}_{\ell,j}\preceq I$. Since
for all $\ell,j\in [L]$, $e^{k_{0}}_{\ell,j} \preceq c\cdot \textbf{1}_{1\times M}$,  using (\ref{eq:Ph1time1})  we have
\begin{equation}
    \Big(1-\eta(W^{h^{\mathcal{G}}})\Big)^{k^\prime}\cdot I\preceq c,
\end{equation}
\begin{equation}
    k^{'}\leq \frac{\log(c/\max_{j\in [L]}I(j))}{\log\big(1-\eta(W^{h^{\mathcal{G}}})\big)}.
\end{equation}
Since $k_{0} \leq h^{\mathcal{G}}(k^\prime+1)$, we have
\begin{equation}
    \begin{split}
        k_{0}
        &\leq h^{\mathcal{G}}\Bigg(\frac{\log(c/\max_{j\in [L]}I(j))}{\log\big(1-\eta(W^{h^{\mathcal{G}}})\big)} +1\Bigg).
    \end{split}
\end{equation}
Now, let $k_{d}$ be the time to detect the   local $c$-consensus. From \cite{xie2017stop}, we have 
\begin{equation}\label{eq:timeToDetectConsensus}
    k_{d} \leq h^{\mathcal{G}}\Bigg(\frac{-\log(d^{\mathcal{G}})}{\log\big(1-\eta(W^{h^{\mathcal{G}}})\big)}+1\Bigg)+L+1.
\end{equation}
Now, the time $N^c$ for initialization phase is bounded as follows
\begin{align}
\label{eq:Ph1}
        N^{c} &\leq k_{0}+k_{d}\nonumber\\
                 &\leq h^{\mathcal{G}}\Bigg(\frac{\log(c/\max_{j\in [L]}I(j))}{\log\big(1-\eta(W^{h^{\mathcal{G}}})\big)} +1\Bigg)\nonumber\\
     &\qquad+h^{\mathcal{G}}\Bigg(\frac{-\log(d^{\mathcal{G}})}{\log\big(1-\eta(W^{h^{\mathcal{G}}})\big)}+1\Bigg)+L+1.
\end{align}

The expected time of the test phase of CCT is bounded above by $\mathbb{E}_i \left [\max_{1\leq \ell\leq L}( T_{i,\ell}+\tau(T_{i,\ell}) ) \right ]$. 
Using (\ref{SS2DCT}) and (\ref{eq:MomentFinite}) of Theorem \ref{DCT}, and the fact that $|\hat{I}_{\ell}(i)-I(i)|\leq c$ (see using Theorem \ref{th:1} and \eqref{eq:error2}), as $c\to 0$, we have 
\begin{eqnarray}
\label{eq:Ph2}
\mathbb{E}_i \left [\max_{1\leq \ell\leq L}( T_{i,\ell}+\tau(T_{i,\ell}) ) \right ] \leq \frac{\abs{\log c}}{I(i)-c}(1+o(1)). 
\end{eqnarray}

Now, we compute the time for the decision phase of CCT. 
The network will reach  the final decision for all $n> \max_{1\leq \ell\leq L}\tau(T_{i,\ell}) + k_{r}$, where $k_{r}$ is the time taken by the termination message $m_{t}^{(3)}$ to reach every node after its initiation at any node. Thus, the time  $N^{s}$ of the decision phase  is  bounded above as
\[N^{s} \leq \max_{1\leq \ell\leq L}\tau(T_{i,\ell}) + k_{r}. \]
Therefore, we have
\begin{equation}
\label{eq:Ph3}
  \mathbb{E}_{i}[ N^{s}] \leq \sum_{\ell=1}^{L}\mathbb{E}_{i}[ \tau(N_{i,\ell})] + \mathbb{E}_{i}[ k_{r}].
\end{equation}
Using (\ref{eq:MomentFinite}), the   term $\mathbb{E}_{i}[\tau(T_{i,\ell})]$ at the right-hand side of (\ref{eq:Ph3})  is finite and independent of $c$.
Additionally, $k_{r} <d^{\mathcal{G}}+1$. Thus, $\mathbb{E}_i[ N^{s}]$ is finite and independent of $c$.

Combining equations (\ref{eq:Ph1}), (\ref{eq:Ph2}), and the finiteness of  $\mathbb{E}_i[ N^{s}]$, we get that~(\ref{EDCT}) holds as $c\to 0$, proving part $(ii)$ of the theorem.

Now we derive the bounds for the higher moments of the decision time $N$. We have
\begin{align}
\label{eq:momentProof}
    N &\leq   N^{c}+\max_{1\leq \ell\leq L}( T_{i,\ell}+\tau(T_{i,\ell}) )+N^{s},\nonumber\\
    &\leq N^{c}+\max_{1\leq \ell\leq L}( T_{i,\ell})+ 2\max_{1\leq \ell\leq L}\tau(T_{i,\ell}) + k_{r},\nonumber\\
    &\leq N^{c}+\max_{1\leq \ell\leq L}( T_{i,\ell})+ 2\sum_{\ell\in[L]}\tau(T_{i,\ell}) + k_{r}.
\end{align}
Now, we present the bound on the $r^{th}$ moment of each term in the right-hand side of (\ref{eq:momentProof}). Using (\ref{eq:Ph1}), $N^{c}$ is bounded above by a constant. As $c\to 0$, we have
\begin{equation} \label{eq:MomentPh1}
\begin{split}
    (N^{c})^{r}&\leq 
    \Bigg((1+o(1))\frac{h^{\mathcal{G}}\log(c/\max_{j\in [L]}I(j))}{\log(1-\eta(W^{h^{\mathcal{G}}})}\Bigg)^{r}.
\end{split}
\end{equation}
Using (\ref{SS2DCT}) and the fact that $|\hat{I}_{\ell}(i)-I(i)|\leq c$, we have
\begin{align}\label{SS2}
\mathbb{E}_i \left [\max_{1\leq \ell\leq L}T_{i,\ell}^{r} \right ] 
&= \Bigg((1+o(1))\frac{\abs{\log(c)}}{I(i)-c}\Bigg)^{r}. 
\end{align}
Using (\ref{eq:MomentFinite}), the higher moments of the third term in the right-hand side of  (\ref{eq:momentProof}) are finite and independent of $c$ by definition of $\tau(T_{i,\ell})$. 
Additionally, $k_{r}\leq L+1<\infty$.
Now, 
\begin{equation}\label{eq:Moment}
    \mathbb{E}^{\mathcal{C}}_{i}[N^{r}]\leq \mathbb{E}_{i}\Big[N^{c}+\max_{1\leq \ell\leq L}( T_{i,\ell})+ 2\sum_{\ell\in[L]}\tau(T_{i,\ell}) + k_{r}\Big]^{r}.
\end{equation}
The moments of $\sum_{\ell\in[L]}\tau(T_{i,\ell}) + k_{r}$ are finite and independent of $c$.
The dominant terms, dependent on $c$, in the expansion of the right-hand side of (\ref{eq:Moment}) depend only on $N^{c}+\max_{1\leq \ell\leq L}( T_{i,\ell})$. Therefore, as $c\to 0$, we have
\begin{equation}\label{eq:MomentFinal}
\begin{split}
      &\mathbb{E}_{i}^{\mathcal{C}}[N^{r}]\\&
      \leq \hspace*{-5pt}\Bigg( (1+o(1)) \Bigg(\frac{h^{\mathcal{G}}\cdot \log(c/\max_{j\in [L]} I(j))}{\log (1-\eta(W^{h^{\mathcal{G}}}))}+\frac{\abs{\log c}}{I(i)-c}\Bigg)\Bigg)^{r},
\end{split}
\end{equation}
which proves part $(iv)$ of the theorem.
 \end{proof}
 \subsection{Proof of Theorem \ref{main2}}
\label{app:CCTMain}
\begin{proof}
Combining Theorems \ref{converse} and \ref{CCT}, (\ref{RfinalCCT}) and (\ref{EfinalCCT})  follow  immediately. 
\end{proof}
 \subsection{Decision Phase of CCT}\label{app:Phase3CCT}
\begin{lemma}\label{Lemma:stoppingrule}
If $d_{\ell}^{N}>L+1$, then there exists a time $k\leq N$ at which the local decision of all the nodes are the same, i.e., $\min_{j\in[L]}x_{j}^{k} \geq 1$. This decision is the same as the local decision $\hat{H}_{\ell}^{N}$ of node $\ell$ at time   $N$.
\end{lemma}
\begin{proof}
At time   $N$ and node $\ell$, if $d_{\ell}^{N}>L+1$, then for all $k\in\mathcal{N}_{\ell}$, $d_{k}^{N-1}>L$ and $x_{k}^{N-2}\geq L$. If the shortest distance between the node $\ell$ and $j$ is $s_{\ell,j}$, then $d_{j}^{N-s_{\ell,j}}>L-s_{\ell,j}+1 $. Thus, for all $j\in [L]$, $d_{j}^{N-d^{\mathcal{G}}}>d_{j}^{N-s_{\ell,j}}+s_{\ell,j}-d^{\mathcal{G}}>1$ as $s_{\ell,j}\leq d^{\mathcal{G}}\leq L$. This implies, for all $j\in [L]$, $x_{j}^{N-d^{\mathcal{G}}-1}\geq 1$. Thus, the first statement of the claim follows.

Now, we  prove the second statement by contradiction. Let the decision at time   $N-d^{\mathcal{G}}-1$, for all $j\in[L]$, be $\hat{H}_{j}^{N-d^{\mathcal{G}}-1}=h^{\prime}$ which is different from $\hat{H}_{\ell}^{N}$. At sensor $\ell$, let the decision change from $h^{\prime}$ to $\hat{H}_{\ell}^{N}$ at time   $n$. Then,
\[N-d^{\mathcal{G}}-1<n\leq N.\]
Therefore,
\[x_{\ell}^{n}=1,\]
which implies
\[d_{\ell}^{n+1}\leq 2.\]
Now,
\begin{align}
d_{\ell}^{N}&\leq d_{\ell}^{n+1}+N-n-1\nonumber\\
&\leq 2+N-n-1\nonumber\\
&<2+d^{\mathcal{G}}\nonumber\\
&\leq 2+L.
\end{align}
However, $d_{\ell}^{N}\geq L+2$ by the statement of the Lemma. Hence, by contradiction, we conclude that the second statement of our claim holds.
\end{proof}
\subsection{Proof of Theorem \ref{DCT_quantized}}
\begin{proof}
The proof of the theorem is exactly along the same lines as the proof of Theorem \ref{DCT}. The key difference lies in the computation of the constant ${\rho_{i,\ell}}$. Due to quantization into $Q$ sub-intervals, we have
\begin{equation}
   v_{i,\ell}-\Delta(\max_i I(i),Q)\leq \floor{v_{i,\ell}}\leq v_{i,\ell},
\end{equation}
which implies 
\begin{equation}\label{eq:quantizedV}
   v_{i,\ell}-f(Q)/L \leq \floor{v_{i,\ell}}\leq v_{i,\ell}.
\end{equation}
Using \eqref{eq:quantizedV}, $\floor{I(i)}$ can be bounded as follows:
\begin{equation}\label{eq:quantizedI}
   I(i)-f(Q) \leq \floor{I(i)}\leq I(i),
\end{equation}
which implies that  $\rho_{i,\ell}$ in (\ref{eq:newrho}) verifies
\begin{equation}\label{eq:quantizationrho}
\begin{split}
       \frac{v_{i,\ell}}{I(i)}\leq {\rho_{i,\ell}} &\leq \frac{v_{i,\ell}}{I(i)-f(Q)}.\\
\end{split}
\end{equation}

For part $(i)$,  using the lower bound from (\ref{eq:quantizationrho}) in (\ref{probAnj1}), we have 
\begin{equation}
   {\mathbb{P}}_{i}^{{\cal D}}(A_{n,j})\leq c^{\sum_{\ell}v_{i,\ell}/I(i)} \, \mathbb{P}^{\mathcal{D}}_{j}(\hat{H}=h_j \mbox{ at sample } n)= c \, \mathbb{P}^{\mathcal{D}}_{j}(\hat{H}=h_j \mbox{ at sample } n).
\end{equation}
Now, the result in part $(i)$ follows similar to
(\ref{eq:proberror2}). 

For part $(ii)$, $(iii)$ and $(iv)$, since $C>\log_2 M$, the local decisions can be communicated at each time instance. Using \eqref{eq:quantizedI} and $f(Q)\leq I(i)$, for all $r\geq 1$, (\ref{SS2DCT}) becomes 
\begin{equation}
    \mathbb{E}_i \left [\Big(\max_{1\leq \ell\leq L}N_{i,\ell}\Big)^r \right ]\leq \Bigg((1+o(1))\frac{\abs{\log c} }{{\floor{I(i)}}}\Bigg)^{r}\leq\Bigg((1+o(1))\frac{\abs{\log c} }{{I(i)}-f(Q)}\Bigg)^{r}.
\end{equation}
Now, similar to  \eqref{eq:MomentFinalDCT}, we have
\begin{equation}
\begin{split}
      \mathbb{E}_{i}^{\mathcal{D}}[N^{r}]
      \leq \hspace*{-5pt}\Bigg( (1+o(1)) \frac{\abs{\log c}}{I(i)-f(Q)}\Bigg)^{r}.
\end{split}
\end{equation}
Hence, part $(ii)$, $(iii)$ and $(iv)$ follows. 
\end{proof}
\subsection{Proof of Theorem \ref{CCT_quantized}}
\begin{proof}
The proof of the theorem is   along the same lines as Theorem \ref{CCT}. The key difference lies in the computation of the constant $\hat{\rho}_{i,\ell}$. 

Firstly, we will upper bound and lower bound $\hat{\rho}_{i,\ell}$ in terms of $I(i)$ and $g(\tilde Q,c,\alpha)$. 
Since Assumptions \ref{assumption:capacity} and  \ref{assumption:initialValue} hold, and the graph $\mathcal{G}$ is strongly connected, 
using \cite[Proposition 5]{nedic2009distributed}, { the time $k_0$ to reach \emph{local} $c$-consensus}  is 
\begin{equation}\label{eq:consensus_quantized_1}
    k_0\leq 2\frac{L^2}{\alpha}\log(\min(\tilde Q^2,L^4/c^2)\cdot \max_j I^2(j))+1.
\end{equation}
 Using $C>\log_2(L+2)$ and \eqref{eq:timeToDetectConsensus}, time $k_d$ to detect the consensus is 
\begin{equation} \label{eq:consensus_quantized_2}
    k_{d} \leq h^{\mathcal{G}}\Bigg(\frac{-\log(d^{\mathcal{G}})}{\log\big(1-\eta(W^{h^{\mathcal{G}}})\big)}+1\Bigg)+L+1.
\end{equation}

Now, using \cite[Proposition 7]{nedic2009distributed} and the fact that the average decreases by at most $1/\tilde Q$ in each iteration of consensus, for all $i\in [M]$ and $\ell \in [L]$, the error in estimation $\floor{\hat{I}_{\ell}(i)}$ at the end of initialization phase is at most
\begin{align}\label{eq:estimationErrorCCT}
       |\floor{\hat{I}_{\ell}(i)}-I(i)|&\leq \frac{L}{\tilde Q}(k_0+k_d),\nonumber\\
       &\leq \frac{L}{\tilde Q}\bigg( 2\frac{L^2}{\alpha}\log(\min(\tilde Q^2,L^4/ c^2)\cdot \max_j I^2(j))+1+h^{\mathcal{G}}\Bigg(\frac{-\log(d^{\mathcal{G}})}{\log\big(1-\eta(W^{h^{\mathcal{G}}})\big)}+1\Bigg)+L+1\bigg),\nonumber\\
       &=g(\tilde Q,c,\alpha).
\end{align}
This implies 
\begin{equation}\label{eq:errorInEstCCTquantized}
    \frac{v_{i,\ell}}{I(i)+g(\tilde Q,c,\alpha)} \leq \hat{\rho}_{i,\ell}\leq \frac{v_{i,\ell}}{I(i)-g(\tilde Q,c,\alpha)}.
\end{equation}

Thus, for part $(i)$, using the lower bound from \eqref{eq:errorInEstCCTquantized} in (\ref{probBnj}), we have
\begin{align}
    {\mathbb{P}}_{i}^{{\mathcal{C}}}(B_{n,j})\leq c^{I(i)/(I(i)+g(\tilde Q,c,\alpha))} \mathbb{P}^{\mathcal{C}}_{j}(\hat{H}=h_j \mbox{ at sample } n).
\end{align}
The result in part $(i)$ follows similar to (\ref{eq:CCTPRobError2}). 

For part $(ii)$, $(iii)$ and $(iv)$, 
the time required in the initialization phase is at most $k_0+k_d$ and can be bounded using
\eqref{eq:consensus_quantized_1} and \eqref{eq:consensus_quantized_2}. For test phase, using \eqref{eq:errorInEstCCTquantized} and $g(\tilde Q,c,\alpha)<I(i)$, for all $r\geq 1$, (\ref{SS2}) becomes 
\begin{align}
\mathbb{E}_i \left [\max_{1\leq \ell\leq L}T_{i,\ell}^{r} \right ] 
&= \Bigg((1+o(1))\frac{\abs{\log(c)}}{\min_{\ell\in [L]}\hat{I}_{\ell}(i)}\Bigg)^{r},\nonumber\\
&\leq \Bigg((1+o(1))\frac{\abs{\log(c)}}{{I}(i)-g(\tilde Q,c,\alpha)}\Bigg)^{r}
\end{align}
For decision phase, 
since $C>\log_2 (L+2)+\log_2 M$, the local decisions and $d_{\ell}^n$ can be communicated at each time instance. Hence, the time $N^{s}$ of decision phase is finite from Theorem \ref{CCT}. Similar to (\ref{eq:MomentFinal}), the result follows by combining the time for all the three phases of CCT. 
\end{proof}
\subsection{Proof of Theorem \ref{DCT_quantized_linkFailure}}
\begin{proof}
For part $(i)$, since the vectors $v_{\ell}$ and $I$ are communicated using $Q$ levels of quantization, the proof is exactly same as that of  part $(i)$ in Theorem \ref{DCT_quantized}. 

For part $(ii)$, $(iii)$ and $(iv)$, the additional delays in comparison to the setting in Theorem \ref{DCT_quantized} are the time to communicate the vectors $\floor{v_{\ell}}$ to the fusion center,
the time to communicate vector $\floor{I}$ to the nodes, and time to make a final decision given the same preferences about the hypothesis are reached at the nodes. Since each link is active with probability $1-\epsilon$, the expected time to communicate the vectors $\floor{v_{\ell}}$ and $\floor{I}$ is at most 
\begin{equation}\label{eq:check11}
    \frac{2L}{1-\epsilon}.
\end{equation}
Given that all the local preferences are reached at the nodes, i.e., $n>\max_{i}\tau(N_{i,\ell})$, the probability that all these preferences are received at the same time instances at the fusion center is $(1-\epsilon)^{L}$, which corresponds to all the links being active at the same time. The expected decision time following $n>\max_{i}\tau(N_{i,\ell})$ is 
\begin{equation}\label{eq:check12}
    \frac{1}{(1-\epsilon)^{L}}.
\end{equation}
Combining the delays in (\ref{eq:check11}) and (\ref{eq:check12}), and  Theorem \ref{DCT_quantized}, the statement of the theorem follows.
\end{proof}
\subsection{Proof of Lemma \ref{thm:propertiesOfW}}
\begin{proof}
For part $(i)$, for all $\ell\in[L]$,
\begin{align}
        \sum_{j=1}^{L} w_{\ell,j}(n)
        &=w_{\ell,\ell}(n)+\sum_{j\neq \ell} w_{\ell,j}(n) \nonumber\\
        &=1-\beta \sum_{j\neq \ell}\mathbf{1}((j,\ell)\in\mathcal{E}(n))+\beta \sum_{j\neq \ell}\mathbf{1}((j,\ell)\in\mathcal{E}(n))\nonumber\\
        &=1.
\end{align}
Since $w_{i,j}(n)=w_{j,i}(n)$, we have 
\begin{equation}
    \sum_{\ell=1}^{L} w_{\ell,j}(n)=1.
\end{equation}
Hence, $W(n)$ is doubly stochastic. 

For part $(ii)$, for all $(i,j)\in \mathcal{E}(n)$, we have
\begin{equation}
    w_{i,j}(n)\geq \min{(\beta,1-\beta \sum_{\ell\neq i}\mathbf{1}((i,\ell)\in\mathcal{E}(n)))}. 
\end{equation}
Thus, for all $(i,j)\in \mathcal{E}(n)$, we have 
\begin{equation}
    w_{i,j}(n)>\min{(1-\mathcal{D(G)}\beta,\beta)},
\end{equation}
where  $\mathcal{D(G)}=\max_{s}\sum_{j\neq s}\mathbf{1}((j,s)\in \mathcal{E})$. 

For part $(iii)$, note that the eigenvalues of $\Bar{L}(n)$ are non-negative {and recall that the sum of the diagonal elements of $\Bar{L}(n)$ is 
equal to the sum of the its eigenvalues~\cite{Johnson-Horn}}. Let $\lambda$ is an eigenvalue of $\Bar{L}(n)$. Then,
\begin{equation}
    \lambda \leq 2 |\mathcal{E}|,
\end{equation}
because $|\mathcal{E}(n)|\leq |\mathcal{E}|$. The eigenvalues of $W(n)$ are of the form
$1-\beta \lambda$. Since $0<\lambda\leq 2|\mathcal{E}|$, then for all $0<\beta<1/(2|\mathcal{E}|)$, we have
\begin{equation}
    0<1-\beta \lambda<1,
\end{equation}
which implies $R(W(n))<1$. To show (\ref{eq:spectralRadiusCheck}), let $\Bar{\lambda}$ be an eigenvalue of $W(n)-{(\textbf{{1}}_{L\times1}\cdot \textbf{{1}}_{1\times L})}/{L}$ and not an eigenvalue of $W(n)$. Then
\begin{align}
        \mbox{det}\bigg(\Bar{\lambda}U_{L\times L}-W(n)+\frac{\textbf{{1}}_{L\times1}\cdot \textbf{{1}}_{1\times L}}{L}\bigg)&\stackrel{(a)}{=} \mbox{det}(\Bar{\lambda}U_{L\times L}-W(n))
       \cdot\bigg(1 + \frac{\textbf{{1}}_{1\times L}\cdot (
       \Bar{\lambda}U_{L\times L}-W(n))^{-1}\cdot \textbf{{1}}_{L\times1} }{L}\bigg) \nonumber\\
       &\stackrel{(b)}{=}\mbox{det}(\Bar{\lambda}U_{L\times L}-W(n)) \bigg(1 + \frac{\textbf{{1}}_{1\times L}\cdot \textbf{{1}}_{L\times 1}}{(\Bar{\lambda}-1)L}\bigg)\nonumber\\
       &{=}\mbox{det}(\Bar{\lambda}U_{L\times L}-W(n))\bigg(1 + \frac{1}{(\Bar{\lambda}-1)}\bigg),
\end{align}
where $\mbox{det}(\cdot)$ denotes the determinant of a matrix, {In the above equation, $(a)$} follows from the fact that $(\Bar{\lambda}U_{L\times L}-W(n))$ is non-singular because $\Bar{\lambda}$ is not an eigenvalue of $W(n)$, and {exploits the matrix determinant lemma:} if $A$ is a non-singular matrix of dimension $L\times L$ and $c$ and $d$ are column vectors of dimension $L\times 1$, then~\cite{Johnson-Horn}
\begin{equation}
    \mbox{det}(A+cd^T)=\mbox{det}(A)(1+d^TA^{-1}c).
\end{equation}
{Also, in the above,} $(b)$ follows from the fact $(\Bar{\lambda}U_{L\times L}-W(n))$ is non-singular and doubly stochastic, which implies
\begin{equation}
\begin{split}
      (\Bar{\lambda}U_{L\times L}-W(n))\textbf{{1}}_{L\times1}&= \Bar{\lambda}\textbf{{1}}_{L\times1} - \textbf{{1}}_{L\times1} \\
      &=(\Bar{\lambda}-1)\textbf{{1}}_{L\times1}.
\end{split}
\end{equation}

Since $\Bar{\lambda}$ is an eigenvalue of $W(n)-{(\textbf{{1}}_{L\times1}\cdot \textbf{{1}}_{1\times L})}/{L}$, we have
\begin{equation}
    \bigg(1 + \frac{1}{(\Bar{\lambda}-1)}\bigg)=0,
\end{equation}
which implies $\Bar{\lambda}=0$. Combining the facts that $\Bar{\lambda}<1$ and $R(W(n))<1$, the claim in $(iii)$ follows.
\end{proof}
\subsection{Proof of Theorem \ref{CCT_quantized_linkFailure}}\label{app:CCT_quantized_linkFailure}
\begin{proof}
The proof of the theorem is along the same lines as the proof of Theorem \ref{CCT_quantized}. The key difference is that, unlike $\hat{\rho}_{i,\ell}$, in this case $\hat{\rho}^{\epsilon}_{i,\ell}$ is a random variable, and the randomness is introduced by the time-varying topology of the network due to $\epsilon$-random packet erasures. 

We derive the upper and lower bound on $\hat{\rho}^{\epsilon}_{i,\ell}$ with high probability in terms of $I(i)$ and $h(\tilde Q,c,\min{(1-\mathcal{D(G)}\beta,\beta)},\epsilon)$. First, we establish that for all $n$, $W(n)$ satisfies Assumption \ref{assumption:capacity}. Second, we show that for all $n$, the resulting graph  $\mathcal{G(V, E}(n))$ is strongly connected with probability at least $1-\mathcal{|E|}\epsilon$. Using these two results, similar to Theorem \ref{CCT_quantized}, we bound the time to reach consensus $K_0+K_d$ which is now a random variable (see (\ref{eq:consensus_quantized_1}) and (\ref{eq:consensus_quantized_2})) and the estimation error (see (\ref{eq:estimationErrorCCT})). The rest of the  proof is similar to that of Theorem \ref{CCT_quantized}.


For all $n$, given an edge $e\in \mathcal{E}$, the probability that $e\notin \mathcal{E}(n)$ is $\epsilon$, as the link failures are independent and identically distributed across time and independent of other links. Thus, the probability that the graph $\mathcal{G(V, E}(n))$ is strongly connected is
\begin{align}\label{eq:connected graph1}
    \mathbb{P}(\mathcal{G(V, E}(n)) \mbox{ is strongly connected})&\geq \mathbb{P}( \mbox{for all } e\in \mathcal{E}, e\in \mathcal{E}(n))\nonumber\\
    &\geq (1-\epsilon)^{|\mathcal{E}|}\nonumber\\
    &\geq 1-|\mathcal{E}|\epsilon.
\end{align}

Since Assumptions \ref{assumption:capacity} and  \ref{assumption:initialValue} hold, and the graph $\mathcal{G(V, E}(n))$ is strongly connected with probability at least $ 1-|\mathcal{E}|\epsilon$, 
using Lemma \ref{thm:propertiesOfW},\cite[Proposition 5]{nedic2009distributed} and (\ref{eq:consensus_quantized_1}), the number of time steps satisfying the property that  $\mathcal{G(V, E}(n))$ is strongly connected and that are required to converge to  {\emph{local} $c$-consensus} is at most
\begin{equation}
    2\frac{ L^2}{\min{(1-\mathcal{D(G)}\beta,\beta)}}\log(\min(\tilde Q^2,L^4/c^2)\cdot \max_j I^2(j))+1.
\end{equation}
Now, using (\ref{eq:connected graph1}), 
 $\mathbb{E}[K_0]$ to reach {\emph{local} $c$-consensus} is 
\begin{equation}\label{eq:consensus_quantized_linkFailue_1}
    \mathbb{E}[K_0]\leq \frac{1}{(1-|\mathcal{E}|\epsilon)} \bigg( 2\frac{L^2}{\min{(1-\mathcal{D(G)}\beta,\beta)}}\log(\min(\tilde Q^2,L^4/c^2)\cdot \max_j I^2(j))+1\bigg).
\end{equation}
Now, similar to (\ref{eq:consensus_quantized_2}),  $\mathbb{E}[K_d]$ to detect consensus is 
\begin{equation} \label{eq:consensus_quantized_linkFailue_detection_1}
   \mathbb{E}[ K_{d}] \leq \frac{1}{(1-|\mathcal{E}|\epsilon)}\bigg( h^{\mathcal{G}}\Bigg(\frac{-\log(d^{\mathcal{G}})}{\log\big(1-\eta(W^{h^{\mathcal{G}}})\big)}+1\Bigg)+L+1\bigg).
\end{equation}
{To obtain} the high probability bound on the estimation error of vector $I$, let us introduce a sequence of Bernoulli i.i.d random variables $\{Z_{n}\}_{n=1}^\infty$ with probability of success  $\mathbb{P}(Z_n=1) =1-|\mathcal{E}|\epsilon$. Then, with probability one, we have 
\begin{equation}\label{eq:newRV_K}
K_0\leq \min \left \{n\geq 1: \sum_{k=1}^{n}Z_k>2\frac{ L^2}{\min{(1-\mathcal{D(G)}\beta,\beta)}}\log(\min(\tilde Q^2,L^4/c^2)\cdot \max_j I^2(j))+1 \right \}. 
\end{equation}
Let $\delta=1/(1-|\mathcal{E}|\epsilon)$ and
\[N_0=\frac{1}{(1-|\mathcal{E}|\epsilon)}\bigg(2\frac{ L^2}{\min{(1-\mathcal{D(G)}\beta,\beta)}}\log(\min(\tilde Q^2,L^4/c^2)\cdot \max_j I^2(j))+1\bigg).\]
Using Hoeffding's inequality \cite{hoeffding1994probability}, we have
\begin{align}\label{eq:highprobabilitybounerror}
&\mathbb P\bigg(K_0\geq N_0(1+\delta)\bigg)\nonumber\\
&\stackrel{(a)}{\leq}  
    \mathbb P\bigg(\sum_{n=1}^{N_0(1+\delta)}Z_n\leq 2\frac{ L^2}{\min{(1-\mathcal{D(G)}\beta,\beta)}}\log(\min(\tilde Q^2,L^4/c^2)\cdot \max_j I^2(j))+1  \bigg)\nonumber\\
    &= \mathbb P\bigg(\sum_{n=1}^{N_0(1+\delta)}Z_n-N_0(1+\delta)(1-|\mathcal{E}|\epsilon)\leq 2\frac{ L^2}{\min{(1-\mathcal{D(G)}\beta,\beta)}}\log(\min(\tilde Q^2,L^4/c^2)\cdot \max_j I^2(j))+1 -N_0(1+\delta)(1-|\mathcal{E}|\epsilon)\bigg)\nonumber\\
    &\stackrel{(b)}{=} \mathbb P\bigg(\sum_{n=1}^{N_0(1+\delta)}Z_n - N_0 (1+\delta)(1-|\mathcal{E}|\epsilon)\leq -N_0 \delta(1-|\mathcal{E}|\epsilon)\bigg) \nonumber\\
    &\leq \exp{(-2 \delta^2(1-|\mathcal{E}|\epsilon)^2 N_0(1+\delta)/(1+\delta)^2)} \nonumber\\
    &=\exp{(-2 \delta^2(1-|\mathcal{E}|\epsilon)^2 N_0/(1+\delta))} \nonumber\\
    &\stackrel{}{=} \exp{(-2 (1-|\mathcal{E}|\epsilon) N_0/(2-|\mathcal{E}|\epsilon))},
\end{align}
where $(a)$ follows from \eqref{eq:newRV_K}, and $(b)$ follows from the definition of $N_0$.

Similarly, we can show that for
$\delta=1/(1-|\mathcal{E}|\epsilon)$ and
\[N^\prime_{0}=\frac{1}{(1-|\mathcal{E}|\epsilon)} \bigg(h^{\mathcal{G}}\bigg(\frac{-\log(d^{\mathcal{G}})}{\log\big(1-\eta(W^{h^{\mathcal{G}}})\big)}+1\bigg)+L+1\bigg),\]
we have
\begin{equation}\label{eq:timetodetectconsensus}
   \mathbb P(K_d>N_{0}^\prime (1+\delta))\leq \exp{(-2(1-|\mathcal{E}|\epsilon)N_0^\prime/(2-|\mathcal{E}|\epsilon))}.
\end{equation}
Thus, similar to (\ref{eq:estimationErrorCCT}), using (\ref{eq:highprobabilitybounerror}) and (\ref{eq:timetodetectconsensus}), we have that with probability one, the error in the estimation of $\floor{\hat{I}^\epsilon_{\ell}(i)}$ at the end of initialization phase is 
\begin{align}
    |\floor{\hat{I}^\epsilon_{\ell}(i)}-I(i)|&\leq \frac{L}{\tilde{Q}}(K_0+K_d),\nonumber.\\
\end{align}
This implies that using \eqref{eq:highprobabilitybounerror} and  \eqref{eq:timetodetectconsensus} , we have
\begin{align}\label{eq:estimationErrorCCT_error}
      |\floor{\hat{I}^\epsilon_{\ell}(i)}-I(i)|&\leq \frac{L}{\tilde{Q}}(K_0+K_d),\nonumber\\
      &\leq \frac{L(1+\delta)}{\tilde Q(1-|\mathcal{E}|\epsilon)}\bigg(
      2\frac{ L^2}{\min{(1-\mathcal{D(G)}\beta,\beta)}}\log(\min(\tilde Q^2,L^4/c^2)\cdot \max_j I^2(j))+1\nonumber\\
      &\qquad + h^{\mathcal{G}}\Bigg(\frac{-\log(d^{\mathcal{G}})}{\log\big(1-\eta(W^{h^{\mathcal{G}}})\big)}+1\Bigg)+L+1\bigg),\nonumber\\
      &=g(\Tilde{Q},c,\min{(1-\mathcal{D(G)}\beta,\beta)})(2-|\mathcal{E}|\epsilon)/(1-|\mathcal{E}|\epsilon)^2,\nonumber\\
      &=h(\Tilde{Q},c,\min{(1-\mathcal{D(G)}\beta,\beta)},\epsilon)
\end{align}
with probability at least 
\[1-\exp{\Bigg(-\frac{2}{(2-|\mathcal{E}|\epsilon)}\bigg(
      2\frac{ L^2}{\min{(1-\mathcal{D(G)}\beta,\beta)}}\log(\min(\tilde Q^2,L^4/c^2)\cdot \max_j I^2(j))+1 
      + h^{\mathcal{G}}\Bigg(\frac{-\log(d^{\mathcal{G}})}{\log\big(1-\eta(W^{h^{\mathcal{G}}})\big)}+1\Bigg)+L+1\bigg)
    \Bigg)},\]
\begin{equation}
\begin{split}
    &=1-\exp{(-2 \tilde{Q} g(\Tilde{Q},c,\min{(1-\mathcal{D(G)}\beta,\beta)})/ L (2-|\mathcal{E}|\epsilon))}, \\
    &= 1- \exp(-2 q(\Tilde{Q},c,\min{(1-\mathcal{D(G)}\beta,\beta),\epsilon} ),
\end{split}
\end{equation} 
as $K_0$ and $K_d$ are independent. 

Now, for part $(i)$, using the lower bound from (\ref{eq:estimationErrorCCT_error}) in (\ref{probBnj}), we have ${\mathbb{P}}_{i}^{{\mathcal{C}}}(B_{n,j})$ is at most
\begin{align}
    &(1- \exp(-2 q(\Tilde{Q},c,\min{(1-\mathcal{D(G)}\beta,\beta),\epsilon} ))\cdot c^{I(i)/(I(i)+h(\Tilde{Q},c,\min{(1-\mathcal{D(G)}\beta,\beta)},\epsilon))}\cdot \mathbb{P}^{\mathcal{C}}_{j}(\hat{H}=h_j \mbox{ at sample } n)\nonumber\\
    &+\exp(-2 q(\Tilde{Q},c,\min{(1-\mathcal{D(G)}\beta,\beta),\epsilon} )\cdot \mathbb{P}^{\mathcal{C}}_{j}(\hat{H}=h_j \mbox{ at sample } n) 
\end{align}
The result in part $(i)$ follows similar to (\ref{eq:CCTPRobError2}). 

Consider next parts $(ii)$, $(iii)$ and $(iv)$. For the consensus phase, the {expected time $ \mathbb{E}[K_0+K_d]$ required is upper bounded by the right hand sides of (\ref{eq:consensus_quantized_linkFailue_1}) and (\ref{eq:consensus_quantized_linkFailue_detection_1}).} For the test phase, using 
$h(\Tilde{Q},c,\min{(1-\mathcal{D(G)}\beta,\beta)},\epsilon)<I(i)$, for all $r\geq 1$, (\ref{SS2}) {becomes}
\begin{align}
\mathbb{E}_i \left [\max_{1\leq \ell\leq L}T_{i,\ell}^{r}\bigg| \{\floor{\hat{I}^\epsilon_{\ell}}\}_{\ell\in [L]} \right ] 
&= \Bigg((1+o(1))\frac{\abs{\log(c)}}{\min_{\ell\in [L]}\floor{\hat{I}^{\epsilon}_{\ell}(i)}}\Bigg)^{r},\nonumber\\
&\stackrel{(a)}{\leq} \Bigg((1+o(1))\frac{\abs{\log(c)}}{{I}(i)-h(\Tilde{Q},c,\min{(1-\mathcal{D(G)}\beta,\beta)},\epsilon)}\Bigg)^{r},
\end{align}
with probability $(1-\exp{(-2q(\Tilde{Q},c,\min{(1-\mathcal{D(G)}\beta,\beta)},\epsilon))})$, where $(a)$ follows from (\ref{eq:estimationErrorCCT_error}). For the stopping phase, 
since $C>\log_2 (L+2)+\log_2 M$, the local decisions and $d_{\ell}^n$ can be communicated at each time instance. Hence, the time $N^{s}$ of the decision phase is finite from Theorem \ref{CCT} and the fact that the probability the graph is strongly connected at each time instance is at least $(1-|\mathcal{E}|\epsilon)>0$. Similar to (\ref{eq:MomentFinal}), the result follows by combining the time for all the three phases of CCT. 
\end{proof}
\section{Proof of Miscellaneous Results}\label{app:C}
In this section, we present results used to bound time $\tau({N}_{i,\ell})$ in Theorem \ref{DCT} and \ref{CCT} (see (\ref{eq:MomentFinite})). Let $X_{1}\ldots X_{n}$ be i.i.d.\ random variables and let the time
\begin{equation}
    T = \sup\Big\{n\ : \sum_{k=1}^{n} X_{k} > 0\Big\}.
\end{equation}
This is the last $n$ at which
\begin{equation}
S_{n} >0, 
\label{eq:2_1}
\end{equation} 
where $S_{n} = \sum_{k=1}^{n} X_{k}$, $n\ge 1$, {and $S_0=0$}.
\begin{lemma}\label{lemma:finiteStoppingTime}
For all $r\geq 1$, if $\mathbb{E}\big[\abs{X_{1}}^{r+1}\big] < \infty$ and {$\mathbb{E}\big[X_{1}\big] \le -\mu_{0}<0$}, then 

{\begin{align}
  \mathbb{E}[T^r] &\leq r \Bigg(\frac{2}{
  \mu_{0}
  }\Bigg)^{r}\mathbb{E}[(S^{*})^r]\nonumber\\
     &\qquad+ \sum_{k=1}^{\infty} r k^{r-1}\mathbb{P}\big(S_{k}+k{\mu_{0}}/2>0\big),
\end{align}}
where {$S^{*}=\max_{j\ge0}S_{j}$}.
\end{lemma}
\begin{proof}
We have 
\begin{align}
&  \mathbb{E}[T^r]\nonumber\\
&{\le \sum_{k=1}^{\infty}rk^{r-1}\mathbb{P}( T\geq k)}\nonumber\\
&=\sum_{k=1}^{\infty}rk^{r-1}\mathbb{P}(\max_{j\geq k}S_{j}>0)\nonumber\\
&=\sum_{k=1}^{\infty}rk^{r-1}\mathbb{P} \bigg(\max_{j\geq k}\big(S_{j}-S_{k}\big)+S_{k}>0\bigg)\nonumber\\
&=\sum_{k=1}^{\infty}rk^{r-1}\mathbb{P}\bigg(S^{*}+S_{k}>0\bigg)
\end{align}
{where $S^{*}$ is an independent copy of $\max_{j\ge 0} S_j$, therefore  we loosely use the same symbol.}

{Now, elaborating as done in~\cite[Theor.~D]{kiefer1963asymptotically}:
\begin{align} \label{eq:myeq1}
&\sum_{k=1}^{\infty}rk^{r-1}\mathbb{P}\bigg(S^{*}+S_{k}>0\bigg)\nonumber\\
&=\int_{0}^{\infty}\sum_{k=1}^{\floor{2\xi/\mu_{0}}}rk^{r-1}\mathbb{P}\big(S_{k}>-\xi\big) \, d\mathbb{P}\big(S^{*}\leq \xi\big)\nonumber\\
&+\int_{0}^{\infty}\sum_{k=\floor{2\xi/\mu_{0}}+1}^{\infty}rk^{r-1}\mathbb{P}\big(S_{k}+\mu_{0}k/2>\mu_{0}k/2-\xi\big) \, d\mathbb{P}\big(S^{*}\leq \xi\big).
\end{align}
The first integral at the right-hand side of~(\ref{eq:myeq1}) can be bounded as follows:
\begin{align} \label{eq:myintegral1}
&\int_{0}^{\infty}\sum_{k=1}^{\floor{2\xi/\mu_{0}}}rk^{r-1}\mathbb{P}\big(S_{k}>-\xi\big) \, d\mathbb{P}\big(S^{*}\leq \xi\big) \nonumber \\
&\le \int_{0}^{\infty}\sum_{k=1}^{\floor{2\xi/\mu_{0}}}rk^{r-1}\, d\mathbb{P}\big(S^{*}\leq \xi\big) \nonumber \\
& \le \int_{0}^{\infty} r \left (2\xi/\mu_{0} \right )^{r} \, d\mathbb{P}\big(S^{*}\leq \xi\big) \nonumber \\
& = r \left ( \frac{2}{\mu_0}\right )^r \mathbb{E}\big[(S^{*})^r\big].
\end{align}
To bound the second integral at the right-hand side of~(\ref{eq:myeq1}), we have 
\begin{align} \label{eq:myintegral2}
&\int_{0}^{\infty}\sum_{k=\floor{2\xi/\mu_{0}}+1}^{\infty}rk^{r-1}\mathbb{P}\big(S_{k}+\mu_{0}k/2>\mu_{0}k/2-\xi\big) \, d\mathbb{P}\big(S^{*}\leq \xi\big) \nonumber \\
& \le \int_{0}^{\infty}\sum_{k=\floor{2\xi/\mu_{0}}+1}^{\infty}rk^{r-1}\mathbb{P}\big(S_{k}+\mu_{0}k/2>0\big) \, d\mathbb{P}\big(S^{*}\leq \xi\big) \nonumber \\
&\le \int_{0}^{\infty}\sum_{k=1}^{\infty}rk^{r-1}\mathbb{P}\big(S_{k}+\mu_{0}k/2>0\big) \, d\mathbb{P}\big(S^{*}\leq \xi\big) \nonumber \\
&\le \sum_{k=1}^{\infty}rk^{r-1}\mathbb{P}\big(S_{k}+\mu_{0}k/2>0\big),
\end{align}
where the first inequality follows from the fact that integration variable verifies $\xi \le {\mu_0 k}/2$. 

The claim of the Lemma now follows by~(\ref{eq:myintegral1}) and~(\ref{eq:myintegral2}).}
\end{proof}

\begin{corollary}\label{corr:1}
{ Let $X_1,\dots,X_n$ be i.i.d. random variables such that $\mathbb{E}[\abs{X_1}^{r+1}]<\infty$ and $\mathbb{E}[X_1]\ge \mu_0>0$. Also, let
$S_0=0$, $S_n=\sum_{k=1}^n X_k$, $n\ge1$, and}
\begin{equation}
    T = \sup\Big\{n\ : \sum_{k=1}^{n} X_{k} < 0\Big\}.
\end{equation}
{Then for all $r\geq 1$, we have
\begin{align}\label{eq:timeBound}
  \mathbb{E}[T^{r}]&\leq r \Big(\frac{2}{\mu_{0}}\Big)^{r}\mathbb{E}\bigg[\Big(-\min_{j\ge0}S_{j}\Big)^{r}\bigg]\nonumber\\
     &\qquad+ \sum_{k=1}^{\infty}r k^{r-1}\mathbb{P}\big(S_{k}-k{\mu_{0}}/2<0\big).
\end{align}}
\end{corollary}
\begin{proof}
The proof follows from replacing $X_{k}$ by $-X_{k}$ in Lemma \ref{lemma:finiteStoppingTime}.
\end{proof}

\begin{lemma}\label{lemma:BoundingTime}
Let $X_{1},\ldots,X_{n}$ be a sequence of independent and  identically distributed random variables with  {zero mean} and finite {$(r+1)^{th}$ absolute moment, $\mathbb{E}[\abs{X}^{r+1}] < \infty$, for all $r\geq 1$.} We have
\begin{equation}\label{eq:Lemm6}
    \sum_{n=1}^{\infty}n^{r-1}\mathbb{P}\Big(\Big|\sum_{k=1}^{n}X_{k}\Big |>n\Big)<\infty
\end{equation}
\end{lemma}
\begin{proof} 

The proof technique is borrowed from~\cite{erdos1949theorem}.
Event $A=\{\abs{\sum_{k=1}^{n}X_{k}}>n\}$ is written as a subset of the union of three events i.e. $A \subset A^{(1)}_{n}\cup A^{(2)}_{n}\cup A^{(3)}_{n}$. We bound the probability of these three events, and show that for all $i\in[3]$, we have
\begin{equation}\label{eq:Summary}
    \sum_{n=1}^{\infty}n^{r-1}\mathbb{P}( A^{(i)}_{n}) < \infty.
\end{equation}
Thus, (\ref{eq:Lemm6}) follows from (\ref{eq:Summary}).

Let $2^{i}\leq n < 2^{i+1}$, {with $i\ge 0$}. The events $A^{(1)}_{n}$, $A^{(2)}_{n}$ and $A^{(3)}_{n}$ are defined as follows:
\begin{align}
&A^{(1)}_{n}=\{\mbox{There exists }  k\leq n \mbox{ such that } \abs{X_{k}}>2^{i-2} \},\nonumber\\
&A^{(2)}_{n}=\{\mbox{There exists at least two integers }  k_{1},k_{2}\leq n \mbox{ such that } \nonumber\\
&\qquad\qquad\abs{X_{k_{1}}}>n^{4/5} \mbox{and } \abs{X_{k_{2}}}>n^{4/5}  \},\nonumber\\
&A^{(3)}_{n}=\Big \{\Big |\sum_{k\in N^{\prime} }X_{k}\Big|>2^{i-2} \Big \},\nonumber
\end{align}
where {$N^{\prime}=[n]\backslash \{k\leq n : \abs{X_k}>n^{4/5}\}$.} If the event $A^{(1)}_{n}\cup A^{(2)}_{n}\cup A^{(3)}_{n}$ does not occur, then we have
\[\bigg | \sum_{k=1}^n X_{k} \bigg | \le 2^{i-2}+2^{i-2}<n.\]
Hence, $A \subset A^{(1)}_{n}\cup A^{(2)}_{n}\cup A^{(3)}_{n}$, and 
$\mathbb{P}(A)\leq \mathbb{P}(A^{(1)}_{n}) +\mathbb{P}(A^{(2)}_{n})+\mathbb{P}(A^{(3)}_{n})$.
Therefore,
\begin{align}\label{eq:finalEq}
    \sum_{n=1}^{\infty}n^{r-1}\mathbb{P}(A)&\leq \sum_{n=1}^{\infty}n^{r-1}\mathbb{P}(A^{(1)}_{n}) +\sum_{n=1}^{\infty}n^{r-1}\mathbb{P}(A^{(2)}_{n})\nonumber\\
    &\qquad\qquad+\sum_{n=1}^{\infty}n^{r-1}\mathbb{P}(A^{(3)}_{n}).
\end{align}
Now, we bound the probability of all three events at the right-hand side of the above equation.

Let {$a_{i}=\mathbb{P}(|X_{k}| > 2^{i})$}. We have
\begin{align}\label{eq:Event1a}
\sum_{i=0}^{\infty}2^{i(r+1)-1}a_{i}
&\stackrel{(a)}{\leq} \sum_{i=0}^{\infty}2^{i(r+1)}(a_{i}-a_{i+1})\nonumber\\
&\stackrel{(b)}{\leq} {\mathbb{E}[\abs{X_{k}}^{r+1}]}\stackrel{(c)}{<}\infty,
\end{align} 
where {$(a)$~follows from $\frac 1 2 \sum_{i=0}^\infty 2^{i(r+1)}a_i \ge \frac {1}{2^{r+1}} \sum_{i=1}^\infty 2^{i(r+1)}a_{i}=\sum_{i=0}^\infty 2^{i(r+1)}a_{i+1}$ $\Leftrightarrow$
$\sum_{i=0}^\infty 2^{i(r+1)}a_i- \frac 1 2 \sum_{i=0}^\infty 2^{i(r+1)}a_i \ge \sum_{i=0}^\infty 2^{i(r+1)}a_{i+1}$,}
$(b)$~follows from the definitions of {$a_i$ and $\mathbb{E}[\abs{X_{k}}^{r+1}]$, after exploiting $\int_{y_1}^{y_{2}}ydy \geq \int_{y_1}^{y_{2}}y_1dy$}, and $(c)$~follows 
from the assumption of the lemma. 
Thus, using (\ref{eq:Event1a}), we have 
\begin{equation}
\label{eq:Event1b}
\sum_{i=0}^{\infty}2^{i(r+1)}a_{i} <\infty.
\end{equation}
Now, we bound the probability of the event {at the right-hand side of~(\ref{eq:finalEq}) that involves} $A^{(1)}_{n}$: 
\begin{align}
\label{eq:Event1c}
&{\sum_{n=1}^{\infty}n^{r-1}\mathbb{P}(A_n^{(1)} )} \nonumber \\
&{=\sum_{n=1}^{\infty}n^{r-1}\mathbb{P}(\exists k\leq n: \abs{X_{k}} > {2^{i-2} \textnormal{ where $i$ verifies } 2^i \le n < 2^{i+1}} )} \nonumber \\
&\stackrel{(a)}{\leq} {\sum_{n=1}^{\infty}n^{r}  \mathbb{P}(\abs{X_{k}}> {2^{i-2}}\textnormal{ where $i$ verifies } 2^i \le n < 2^{i+1}})\nonumber\\
&{= \sum_{i=0}^{\infty}\sum_{2^{i}\leq n<2^{i+1}} n^{r} a_{i-2}}\nonumber\\
&\leq \sum_{i=0}^{\infty}\sum_{2^{i}\leq n<2^{i+1}}2^{(i+1)r}a_{i-2}\nonumber\\
&=\sum_{i=0}^{\infty}2^{i(r+1)+r}a_{i-2}\nonumber\\
&<\infty,
\end{align}
where (a) follows from the union bound and the fact that $X_k$ are i.i.d, and
the last inequality follows from (\ref{eq:Event1b}).

Since the $(r+1)^{st}$ absolute moment is finite,   for all $k\in\mathbb{N}$  there exists a finite constant $K>0$ such that
\begin{equation}
\label{eq:LemmaMoment}
\mathbb{P}(\abs{X_{k}} \geq u )\leq K/u^{r+1}.
\end{equation}
Now, we bound the probability of event $A^{(2)}_{n}$
\begin{align}\label{eq:Event1}
\mathbb{P}(A^{(2)}_{n})
&\stackrel{(a)}{\leq} \sum_{1\leq k_{1}<k_{2}\leq n}\mathbb{P}(\abs{X_{k_{1}}}>n^{4/5} 
\mbox{and } \abs{X_{k_{2}}}>n^{4/5})\nonumber\\
&\stackrel{(b)}{\leq} n^{2}\cdot\mathbb{P}(\abs{X_{1}}>n^{4/5}) \, \mathbb{P}( \abs{X_{2}}>n^{4/5})\nonumber\\
&\stackrel{(c)}{\leq} K^{2}\cdot n^{2}\cdot n^{-4(r+1)/5}\cdot n^{-4(r+1)/5}
\end{align}
where $(a)$ follows from the definition of the event and the union bound, $(b)$ follows from the independence of the random variables and 
{a bound on} the number of possible combinations of $k_1$ and $k_2$, and $(c)$ follows from (\ref{eq:LemmaMoment}). Therefore, we have
\begin{align}\label{eq:Event2c}
    \sum_{n=1}^{\infty}n^{r-1}\mathbb{P}( A^{(2)}_{n}) &\stackrel{(a)}{\leq} {\sum_{n=1}^{\infty}K^{2}\cdot n^{r-1} \cdot n^{2}\cdot n^{-8(r+1)/5}}\nonumber\\
    &=\sum_{n=1}^{\infty} K^{2}\cdot n^{-3r/5-3/5}\nonumber \\
    &\stackrel{(b)}{<}\infty,
\end{align}
where $(a)$ follows from (\ref{eq:Event1}), and $(b)$ follows as $r\geq 1$.

Now, we bound the probability of event $A^{(3)}_{n}$. Let 
\begin{equation}\label{eq:event3}
    X_{k}^{+}= \begin{cases}
                         X_{k} \qquad  |X_{k}| \tcdkb{<} n^{4/5}\qquad \\
                         0 \qquad \mbox{otherwise.}
                        \end{cases}
\end{equation}

Now, there exist finite positive constants {$K^{(1)},K^{(2)}$, such that
\begin{align}\label{eq:Event3}
&\mathbb{E}\bigg[\Big|\sum_{k=1}^{n}X_{k}^{+}\Big|^{2r}\bigg]\nonumber\\
&\stackrel{(a)}{\leq}\mathbb{E}\bigg[\sum_{k=1}^{n}|X_{k}^+|^{2r}\bigg]+ \sum_{1 \leq k_1,k_2\leq n}\mathbb{E}\big[|X_{k_{1}}^+|^{2r-1}\big]\mathbb{E}\big[|X_{k_{2}}^+|\big] +\ldots\nonumber\\
&\stackrel{(b)}{\leq}\mathbb{E}\bigg[\sum_{k=1}^{n}n^{4(r-1)/5}|X_{k}^+|^{r+1}\bigg]+ \sum_{1 \leq k_1,k_2\leq n}\mathbb{E}\big[n^{4(r-2)/5}|X_{k_{1}}^+|^{r+1}\big]\mathbb{E}\big[|X_{k_{2}}^+|\big] +\ldots\nonumber\\
&\stackrel{(c)}{\leq} K^{(1)}\cdot n^{4r/5}r^{2r}\bigg(n^{-4/5}+n^{-8/5}+\ldots\bigg)\nonumber\\
&\stackrel{(d)}{\leq} K^{(1)}\cdot n^{4r/5}r^{2r} \frac{n^{-4/5}}{1-n^{-4/5}}\leq K^{(2)} \cdot n^{4(r-1)/5},
\end{align}
}
where $(a)$ follows from the multinomial expansion of {$(\sum_{k=1}^{n}\abs{X_{k}^+})^{2r}$}, and the independence of the random variables, $(b)$ follows from (\ref{eq:event3}), and $(c)$  follows from the following facts:
\begin{itemize}
    \item $(r+1)^{st}$ {absolute moment of $X_{k}^+$ is finite;}
    \item the coefficient of multinomial expansion is of the form $2r!/(k_1!\ldots k_n!)$ such that $k_1+\ldots+k_n=2r$ and can be bounded as ${\cal O}(2r^{2r})$ independent of $n$; 
    \item the largest coefficient of $n$ in the expansion is $n^{4r/5}$ present in the first term in $(b)$
    \item the remaining coefficient of $n$ will form a finite geometric progression, i.e. $n^{-4/5},n^{-8/5},\ldots$,
\end{itemize}
and $(d)$ follows from the fact that sum of the geometric progression $n^{-4/5},n^{-8/5},\ldots$ can be bounded by ${n^{-4/5}}/(1-n^{-4/5})$.

Thus, using (\ref{eq:Event3}), there exists $K^{(3)}>0$ such that
\begin{equation}
\label{eq:Event3a}
\mathbb{P}\bigg (\Big | \sum_{k=1}^{n}X_{k}^+ \Big | >n/8 \bigg)\leq K^{(3)} \cdot n^{4(r-1)/5}/n^{2r},
\end{equation}
and
\begin{align}\label{eq:Event3b}
\mathbb{P}(A^{(3)}_{n})
&=\mathbb{P}\bigg (\Big |\sum_{k=1}^{n}X_{k}^{+} \Big |>2^{i-2} \bigg)\nonumber\\
&\stackrel{(a)}{\le} \mathbb{P}\bigg (\Big |\sum_{k=1}^{n}X_{k}^{+} \Big |>n/8 \bigg)\nonumber\\
&\stackrel{(b)}{\le} K^{(3)} \cdot n^{4(r-1)/5}\cdot n^{-2r},
\end{align}
where~$(a)$ follows by $n/8<2^{i-2}$, 
and~$(b)$ follows from (\ref{eq:Event3a}). Thus, we have
\begin{align}\label{eq:Event3c}
&\sum_{n=1}^{\infty}n^{r-1}\mathbb{P}(A^{(3)}_{n})\nonumber\\
&\stackrel{(a)}{\le} \sum_{n=1}^{\infty} n^{r-1} K^{(3)} \cdot n^{4(r-1)/5}\cdot n^{-2r} 
\nonumber \\
&= \sum_{n=1}^{\infty} K^{(3)} \cdot n^{-r/5}n^{-9/5}
\stackrel{(b)}{<}\infty,
\end{align}
where $(a)$~follows from (\ref{eq:Event3b}), and $(b)$~from the convergence of summation for $r\geq 1$.
Finally, using (\ref{eq:Event1c}), (\ref{eq:Event2c}) and (\ref{eq:Event3c}), (\ref{eq:Lemm6}) follows.

\end{proof}

{Now, we combine the results in Corollary \ref{corr:1} and Lemma \ref{lemma:BoundingTime}. Let $X_1,\dots,X_n$ be i.i.d. random variables with $\mathbb{E}[X_1]=\mu_x> 0$, and $\mathbb{E}[\abs{X_1}^{r+1}]< \infty$ for all $r \ge 1$. Let $S_0=0$, $S_n=\sum_{k=1}^n X_k$ for $n\ge 1$, and $T=\sup \Big \{ n \, : \, S_n < 0\Big \}$. From \eqref{eq:timeBound}:
\begin{align} \label{eq:neq1}
    E[T^r] &\le r \Big(\frac{2}{\mu_{x}}\Big)^{r}\mathbb{E}\bigg[\Big(-\min_{j\ge0}S_{j}\Big)^{r}\bigg] \nonumber \\
    &+ \sum_{k=1}^{\infty}r k^{r-1}\mathbb{P}\big(S_{k}-k{\mu_{x}}/2<0\big).
\end{align}
The first term at the right hand side of~\eqref{eq:neq1} is finite because of the assumptions $\mu_x >0$ and $\mathbb{E}[\abs{X_1}^{r+1}]<\infty$ \cite{{kiefer1963asymptotically}}. }
{The second term can be bounded as follows:
\begin{align}
   &\sum_{k=1}^{\infty}r k^{r-1}\mathbb{P}\big(S_{k}-k{\mu_{x}}/2<0\big) \nonumber \\
   &= \sum_{k=1}^{\infty}r k^{r-1}\mathbb{P}\big(2 k - 2 S_k/\mu_x>k\big)
   \nonumber \\
   &\le\sum_{k=1}^{\infty}r k^{r-1}\mathbb{P}\Big(\Big | 2 k - 2 S_k/\mu_x \Big |>k\Big)
   < \infty,
\end{align}
where the last inequality follows by Lemma \ref{lemma:BoundingTime} applied to the zero-mean i.i.d. variables $\{2 - 2 X_i/ \mu_x\}_{i=1}^\infty$. Thus, we arrive at
\begin{equation}\label{eq:AppendixC3}
    \mathbb{E}[T^r]<\infty.
\end{equation}}

\label{eq:AppendixC2_OLD}

\section*{Acknowledgment}
This work was partially supported by NSF awards number CNS-1446891 and    CCF-1717942.

\ifCLASSOPTIONcaptionsoff
  \newpage
\fi



%

\ifCLASSOPTIONcaptionsoff
  \newpage
\fi

\bibliographystyle{IEEEtran}
\bibliography{IEEEfull,chernoff}

\end{document}